\documentclass[11pt]{article}

\usepackage[T1]{fontenc}
\usepackage{palatino}
\usepackage{amsmath}
\usepackage{graphicx}
\usepackage[colorlinks=true, allcolors=blue]{hyperref}
\newcommand{\pred}[1]{\ensuremath{\mathsf{#1}}}

\usepackage{fullpage,amsfonts,latexsym,amsthm,graphicx,amssymb,amsmath,enumitem,mathrsfs}
\usepackage{cancel,color,soul}

\usepackage{xcolor}
\definecolor{almostblack}{rgb}{0.6, 0.0, 0.0}

\definecolor{almostblackk}{rgb}{0.3, 0, 0.0}

\usepackage{hyperref,cleveref}
\hypersetup{
    colorlinks=true,
    citecolor=almostblack,
    linkcolor=almostblack,
    filecolor=almostblack,      
    urlcolor=almostblack,
}

\mathchardef\mhyphen="2D
\usepackage{mathtools}

\newtheorem{theorem}{Theorem}[section]
\newtheorem{claim}{Claim}[section]
\newtheorem{subclaim}{SubClaim}[section]

\newtheorem{definition}{Definition}[section]
\newtheorem{corollary}{Corollary}[section]
\newtheorem{lemma}{Lemma}[section]

\newtheorem{assumption}{Assumption}
\newtheorem{conjecture}{Conjecture}

\newlength{\protowidth}

\definecolor{blue(ncs)}{rgb}{0.0, 0.53, 0.74}

\usepackage{amssymb}
\usepackage{amsmath}
\usepackage{amsthm}
\usepackage{amsfonts}
\usepackage{xspace}
\usepackage{bm}
\usepackage{bbm}
\usepackage{boxedminipage}
\usepackage{xparse}
\usepackage{float}
\usepackage{multirow}
\usepackage{graphicx}
\usepackage{caption}
\usepackage{subcaption}

\usepackage{tikz}
\usetikzlibrary{positioning}
\usetikzlibrary{fit}
\usetikzlibrary{calc}
\usetikzlibrary{backgrounds}
\usetikzlibrary{shapes}
\usetikzlibrary{patterns}
\usetikzlibrary{matrix}

\usepackage{enumitem}

\usepackage{geometry}
\newcommand{\io}{\mathsf{io}}

\newcommand{\sdotfill}{\textcolor[rgb]{0.2,0.2,0.2}{\dotfill}} 

\newcommand{\namedref}[2]{\hyperref[#2]{#1~\ref*{#2}}\xspace}

\def\poly{\ensuremath{\mathsf{poly}}\xspace}
\def\negl{\ensuremath{\mathsf{negl}}\xspace}
\newcommand{\secpar}{\ensuremath{{n}}\xspace}

\newcommand{\sR}{\ensuremath{\mathsf{R}}}

  \newcommand{\cA}{\ensuremath{{\mathcal A}}\xspace}
  \newcommand{\cB}{\ensuremath{{\mathcal B}}\xspace}
  \newcommand{\cC}{\ensuremath{{\mathcal C}}\xspace}
  \newcommand{\cD}{\ensuremath{{\mathcal D}}\xspace}
  \newcommand{\cE}{\ensuremath{{\mathcal E}}\xspace}
  
  \newcommand{\cG}{\ensuremath{{\mathcal G}}\xspace}
  \newcommand{\cH}{\ensuremath{{\mathcal H}}\xspace}
  
  \newcommand{\cL}{\ensuremath{{\mathcal L}}\xspace}
  \newcommand{\cM}{\ensuremath{{\mathcal M}}\xspace}
  \newcommand{\cN}{\ensuremath{{\mathcal N}}\xspace}
  \newcommand{\cO}{\ensuremath{{\mathcal O}}\xspace}
  
  \newcommand{\cQ}{\ensuremath{{\mathcal Q}}\xspace}
  \newcommand{\cR}{\ensuremath{{\mathcal R}}\xspace}
  \newcommand{\cS}{\ensuremath{{\mathcal S}}\xspace}
  
  \newcommand{\cU}{\ensuremath{{\mathcal U}}\xspace}

  \newcommand{\cX}{\ensuremath{{\mathcal X}}\xspace}
  \newcommand{\cY}{\ensuremath{{\mathcal Y}}\xspace}

  \newcommand{\bbA}{\ensuremath{{\mathbb A}}\xspace}
  \newcommand{\bbB}{\ensuremath{{\mathbb B}}\xspace}
  \newcommand{\bbC}{\ensuremath{{\mathbb C}}\xspace}
  \newcommand{\bbD}{\ensuremath{{\mathbb D}}\xspace}
  \newcommand{\bbE}{\ensuremath{{\mathbb E}}\xspace}
  \newcommand{\bbF}{\ensuremath{{\mathbb F}}\xspace}
  \newcommand{\bbG}{\ensuremath{{\mathbb G}}\xspace}
  
  \newcommand{\bbI}{\ensuremath{{\mathbb I}}\xspace}

  \newcommand{\bbL}{\ensuremath{{\mathbb L}}\xspace}
  
  \newcommand{\bbN}{\ensuremath{{\mathbb N}}\xspace}
  \newcommand{\bbO}{\ensuremath{{\mathbb O}}\xspace}

  \newcommand{\bbR}{\ensuremath{{\mathbb R}}\xspace}
  \newcommand{\bbS}{\ensuremath{{\mathbb S}}\xspace}
  
    \newcommand{\bbU}{\ensuremath{{\mathbb U}}\xspace}
    \newcommand{\bbV}{\ensuremath{{\mathbb V}}\xspace}

  \renewcommand{\Pr}[0]{\mathrm{Pr}\xspace}

\newcommand{\kabir}[1]{{{\textcolor{red}{}}}}
\newcommand{\dakshita}[1]{{{\textcolor{blue}{}}}}

\newcommand{\bin}{\{0,1\}}

\newcommand{\sample}{\leftarrow}

\definecolor{lg}{gray}{0.89}

\DeclareMathOperator*{\Prr}{Pr}

\newcommand{\s}{\mathsf{s}}
\newcommand{\ct}{\mathsf{c}}

\newcommand{\Gen}{\ensuremath{\mathsf{Samp}}\xspace}

\newcommand{\Ver}{\ensuremath{\pred{Ver}}\xspace}
\newcommand{\TD}{\ensuremath{\pred{TD}}\xspace}
\newcommand{\SD}{\ensuremath{\pred{SD}}\xspace}

\newcommand{\x}{\ensuremath{k}\xspace}

\newcommand{\bbZ}{\ensuremath{\mathbb{Z}}\xspace}
\newcommand{\sA}{\ensuremath{\mathsf{A}}\xspace}

\newcommand{\sK}{\ensuremath{\mathsf{K}}\xspace}

\newcommand{\sS}{\ensuremath{\mathsf{S}}\xspace}
\newcommand{\sV}{\ensuremath{\mathsf{V}}\xspace}
\newcommand{\sX}{\ensuremath{\mathsf{X}}\xspace}
\newcommand{\sY}{\ensuremath{\mathsf{Y}}\xspace}
\newcommand{\sZ}{\ensuremath{\mathsf{Z}}\xspace}

\newcommand{\Supp}{\ensuremath{\mathsf{Supp}}\xspace}

\newcommand{\tr}{\mathsf{Tr}}

\newcommand{\ket}[1]{|#1\rangle}
\newcommand{\bra}[1]{\langle #1 |}
\newcommand{\ketbra}[1]{\ket{#1}\bra{#1}}
\newcommand{\tM}{\widetilde{M}}
\newcommand{\ta}{\widetilde{a}}
\newcommand{\ttt}{\widetilde{t}}

\newcommand{\tp}{\widetilde{p}}

\newcommand{\tA}{\widetilde{A}}
\newcommand{\tB}{\widetilde{B}}
\newcommand{\tD}{\widetilde{D}}

\newcommand{\thres}{t}

\newcommand{\adv}[1]{\ket{\tau_\mathsf{#1}}}
\newcommand{\junk}[1]{\ket{\mathsf{junk}_{#1}}}
\newcommand{\arctanb}{\mathsf{arctan2}}
\newcommand{\sol}{\mathsf{sol}}

\usepackage{ wasysym }

\title{
{Founding Quantum Cryptography on Quantum Advantage\\ \vspace{2mm}
\Large{\em or, 
Towards 
Cryptography from 
$\#\mathsf{P}$-Hardness}}
} 

\author{Dakshita Khurana\thanks{UIUC and NTT Research. Email:~\texttt{dakshita@illinois.edu}}
\and
Kabir Tomer\thanks{UIUC. Email:~\texttt{ktomer2@illinois.edu}}
}

\date{}
\begin{document}

\maketitle

\thispagestyle{empty}
\begin{abstract}
Recent oracle separations [Kretschmer, TQC'21, Kretschmer et. al., STOC'23]
have raised the tantalizing possibility of 
building quantum cryptography from sources of hardness that persist even if the polynomial hierarchy collapses.
We realize this possibility by building quantum bit commitments and secure computation from {\em unrelativized}, well-studied mathematical problems that are conjectured to be hard for $\mathsf{P}^{\#\mathsf{P}}$ -- such as approximating the permanents of complex Gaussian matrices, or approximating the output probabilities of random quantum circuits. 
Indeed, we show that as long as {\em any one of the conjectures} underlying sampling-based quantum advantage (e.g., BosonSampling [Aaronson-Arkhipov, STOC'11], Random Circuit Sampling [Boixo et. al., Nature Physics 2018], IQP [Bremner, Jozsa and Shepherd, Proc. Royal Society of London 2010]) is true, quantum cryptography can be based on the extremely mild assumption that $\mathsf{P}^{\#\mathsf{P}} \not\subseteq (\io)\mathsf{BQP}/\mathsf{qpoly}$.

Our techniques uncover strong connections between the hardness of approximating the probabilities of outcomes of quantum processes, the existence of ``one-way'' state synthesis problems, and the existence of useful cryptographic primitives such as one-way puzzles and quantum bit commitments.
Specifically, we prove that the following hardness assumptions are equivalent under $\mathsf{BQP}$ reductions.
\begin{itemize}
    \item The {\bf hardness of approximating the probabilities} of outcomes of certain efficiently sampleable distributions. 
    That is, there exist 
    quantumly efficiently sampleable distributions for which it is hard to approximate the probability assigned to a randomly chosen string in the support of the distribution (upto inverse polynomial relative error).
    \item The existence of {\bf one-way puzzles}, where a quantum sampler outputs a pair of classical strings -- a puzzle and its key -- and where the hardness lies in finding the key corresponding to a random puzzle. These are known to imply quantum bit commitments [Khurana and Tomer, STOC'24].
    \item The existence of {\bf state puzzles}, or one-way state synthesis, 
    where it is hard to synthesize a secret quantum state given a public classical identifier. These capture the hardness of search problems with quantum secrets and classical challenges.   
\end{itemize}
These are the first constructions of quantum cryptographic primitives (one-way puzzles, quantum bit commitments, state puzzles) from well-studied mathematical assumptions that do not imply the existence of classical cryptography.

Along the way, we also show that distributions that admit efficient quantum samplers but cannot be pseudo-deterministically efficiently sampled imply quantum commitments. 
\end{abstract}

\newpage

\thispagestyle{empty}
\tableofcontents
\newpage

\pagenumbering{arabic}
\section{Introduction}
Nearly all of modern classical cryptography relies on unproven computational hardness. 
Decades of studying cryptosystems based on various concrete mathematical problems led to the emergence of complexity-based cryptography. Modern cryptography allows us to categorize the hardness offered by mathematical problems into {\em generic} cryptographic primitives, studying how abstract primitives are reducible to one another, and setting aside which concrete implementation of the generic primitive is used. For example, a one-way function is a classically efficiently computable function that is hard to invert. 
Concrete candidates for one-way functions are known based on a variety of algebraic problems such as the hardness of discrete logarithms, quadratic residuosity or learning with errors.
The existence of one-way functions is a fundamental hardness assumption, necessary for the existence of modern classical cryptography~\cite{STOC:LubRac86,FOCS:ImpLub89,STOC:ImpLevLub89}. \\

\noindent{\bf Hardness in Quantum Cryptography.}
At the same time, the breakthrough ideas of Weisner~\cite{Wiesner1983ConjugateC}, Bennett and Brassard~\cite{BenBra84} demonstrated the possibility of quantum cryptography -- specifically key agreement -- without the need for unproven assumptions, and based solely on the nature of quantum information. Unfortunately, it was soon discovered~\cite{LoChau97,Mayers97} that other fundamental quantum cryptographic primitives, including bit commitments and secure computation, {\em require} some form of computational hardness. It is now known~\cite{GLSV,C:BCKM21b} that (post-quantum) {\em one-way} functions suffice to build bit commitments and quantum secure computation. However, it is plausible that one-way functions are not {\em necessary} for quantum cryptography, and that sources of hardness even milder than the existence of one-way functions could suffice. Can we precisely quantify the amount of hardness that is necessary for quantum cryptography?

Recent works have attempted to address this problem by introducing quantum relaxations of generic classical primitives, and building quantum cryptography from these relaxations. 
Some examples are pseudorandom quantum states~\cite{C:JiLiuSon18} and one-way state generators/one-way puzzles~\cite{C:MorYam22,KT24}, that have been introduced as quantum analogues of 
pseudorandom generators and one-way functions respectively. 
All of these primitives are cryptographically ``useful'' in that they imply quantum bit commitments and secure computation~\cite{AQY21,C:MorYam22,KT24, BJ24}.
Furthermore, {\em relative to appropriately chosen oracles}, these primitives are weaker than their classical counterparts --  more precisely, there exist oracles relative to which secure instantiations of one-way states/pseudorandom states exist, even when $\mathsf{BQP = QMA}$ (respectively, $\mathsf{P=NP}$)~\cite{Kretschmer21,KQST}.
So while classical cryptography would completely break down if $\mathsf{P} = \mathsf{NP}$, there is the exciting  possibility that quantum cryptography would continue to exist!\\

\noindent{\bf Beyond Oracle Worlds.}
At this point it is natural to ask if there are {\em any} concrete candidates for quantum cryptosystems in the {\em real} (unrelativized) world, based on mathematical assumptions that do not imply classical cryptography. Unfortunately, the answer so far has been a resounding no.

Let us explain what we're looking for in more detail.
The gold standard in modern cryptography is to base security on mathematical conjectures whose statement is scientifically interesting independently of the cryptographic application itself\footnote{For instance, random circuits have been conjectured to output pseudorandom states~\cite{AQY21}--- but this essentially assumes that a given construction is secure, and to our knowledge, has not been cryptanalyzed or theoretically studied. It is also known to not be true for certain architectures (e.g., BosonSampling outputs are not pseudorandom~\cite{AA14}), and it is unclear why this conjecture would be separated from the existence of one-way functions.}. We want to avoid assuming that a suggested scheme itself is secure; since such assumptions are construction-dependent and the proclaimed guarantee of provable security essentially loses its meaning (we refer the reader to a survey by Goldwasser and Kalai~\cite{GoldKal} for a more thought-provoking analysis).
Underlying modern classical cryptography is a bedrock of concrete mathematical problems that have been introduced and cryptanalyzed extensively and often independently of their cryptographic application.

Going back to the state of affairs in quantum cryptography: so far, all provably secure constructions of quantum cryptosystems, from well-studied assumptions, rely on the existence of one-way functions.
At the same time, building on all the excitement about the possibility of cryptography without one-way functions, there is a large body of work conceptualizing  quantum generalizations of classical cryptographic primitives, and reducing them to one another. Besides quantum commitments, secure computation, pseudorandom and one-way states discussed above, other examples include quantum public-key and private-key encryption, signature schemes, etc.\footnote{We refer the reader to the graph at \url{https://sattath.github.io/qcrypto-graph/} for several additional examples.}
In fact, there is even a name for a world in which one-way functions do not exist and yet secure quantum cryptography is possible: {\em Microcrypt}.

However, in the absence of real instantiations, a sensible objection to this body of work is that we may just be building fictional castles in the air --  with no hope of securely realizing these cryptosystems without assuming the existence of one-way functions. Our work fundamentally refutes this objection.

We provide the first {\em unrelativized} constructions of quantum bit commitments and secure computation from well-studied mathematical assumptions that do not imply classical cryptography.
This addresses a major unresolved problem in the area, namely offering unrelativized evidence that the hardness required for quantum cryptography is weaker than that required for classical cryptography. \\

\noindent{\bf Hardness beyond the Polynomial Hierarchy.}
Recall that relative to certain oracles, quantum bit commitments exist even if $\mathsf{P} = \mathsf{NP}$, i.e., even if all languages in the polynomial hierarchy ($\mathsf{PH}$) can be efficiently decided. 
This indicates that we may be able to build quantum cryptography from problems that lie outside the polynomial hierarchy, and plausibly continue to be hard even if $\mathsf{PH}$ collapses. One natural complexity class that is believed to not be contained in $\mathsf{PH}$ is ${\#\mathsf{P}}$, known to contain problems such as finding the permanent of a given real or complex valued matrix. Besides being $\#\mathsf{P}$-hard to compute, permanents also have worst-case to average-case reductions, opening up the splendid possibility of basing quantum cryptography directly on worst-case hardness.

But actually building cryptosystems from the hardness of computing permanents turns out to be tricky due to the following conceptual barrier.
In most natural constructions of (quantum) cryptosystems, honest parties need to be able to sample hard problems together with their solutions -- whereas it appears unlikely that random matrices can be quantumly efficiently sampled together with their permanents. 
By exploiting connections with sampling-based quantum advantage and using some indirection, we overcome this barrier. This allows us to obtain quantum cryptography from the average-case hardness of {\em approximating} permanents of complex-valued matrices, or the hardness of  approximating probabilities of outcomes of random circuits.

\subsection{Our Results}
Our first conceptual contribution is a connection between {\em sampling-based quantum advantage} and a quantum cryptographic primitive called a {\em one-way puzzle}. One-way puzzles are notable for implying quantum commitments~\cite{KT24}, which in turn imply quantum secure computation~\cite{C:BCKM21b,GLSV,AQY21}.

A one-way puzzle~\cite{KT24} is a simple, cryptographically useful primitive with classical outputs -- analogous to a randomized variant of a one-way function. It consists of a quantum polynomial time algorithm $\mathsf{Samp}$ and a Boolean function $\mathsf{Ver}$. $\mathsf{Samp}$ outputs a pair of classical strings -- a puzzle and key $(s,k)$ -- satisfying $\Ver(s,k) = 1$. The security guarantee is that given a ``puzzle'' $s$, it is (quantum) computationally infeasible to find a key $k$ such that $\Ver(s,k) = 1$, except with negligible probability.

Through a connection with quantum advantage, our first result yields multiple unrelativized instantiations of one-way puzzles and quantum commitments (and therefore also secure computation) without one-way functions, based on well-studied mathematical problems that are conjectured to be $\mathsf{P}^{\#\mathsf{P}}$-hard. 
We discuss these in the following subsection, where we begin by briefly describing what we mean by quantum advantage.

\subsubsection{One-Way Puzzles and Sampling-Based Quantum Advantage}
More than two decades ago, it was observed~\cite{TV} that quantum computers can sample from distributions that likely cannot be reproduced on any classical device.
Subsequent works have solidified complexity-theoretic evidence that the output distributions of a variety of quantum computations -- including several types of non-universal computations -- may be computationally intractable to simulate on a classical device (see, e.g., ~\cite{shepherd2009temporally,BJS,AA11,bremneraverage,fujiicommuting,boixoetal,BFNV19,kondofocs,boulandnoise,krovi22,movassagh,Zlokapa2023} and a recent survey~\cite{qsurvey}).
We show that any of these quantum computations also yield quantum cryptography without one-way functions, under the same complexity conjectures that have been studied in context of quantum advantage and the 
mild additional assumption that $\mathsf{P}^{\#\mathsf{P}}\not\subseteq\mathsf{(io)BQP}/\mathsf{qpoly}$.

At first, building quantum cryptography only from quantum advantage may appear surprising or even unlikely -- quantum advantage is all about tasks that are hard for {\em classical} computers, whereas quantum cryptography asks for hardness against {\em quantum} computers. Perhaps one may be able to obtain (only) {\em classically secure} one-way puzzles from advantage, but how would one possibly obtain hardness against quantum machines from sampling tasks that are hard only for classical machines?

To understand why, we will peel back the layers a little bit. A key idea in sampling-based quantum advantage (originating in~\cite{AA11}) is to relate the hardness of classical sampling to the hardness of computing the probabilities of outcomes of quantumly efficiently sampled distributions. 
Specifically, major existing proposals for sampling-based quantum advantage assume that the following average-case problem is $\#\mathsf{P}$-hard: 
\begin{center}
{\em Given the description of a quantum sampler, approximate the probability assigned by\\ the sampler's output distribution to a uniformly chosen string in its support.}
\end{center}
Here, the approximation is required to have inverse polynomially small relative error. In addition, outputs of the sampler are assumed to satisfy a natural {\em anti-concentration} property -- requiring that not all of the hardness of approximation should lie on points that have extremely tiny (e.g., doubly exponentially small) probability mass.
We capture this with the following assumption that admits instantiations from many concrete conjectures corresponding to different frameworks for quantum advantage such as BosonSampling, universal random circuits, IQP circuit sampling, etc.
 
\begin{assumption}[$\#\mathsf{P}$-Hardness of Approximating Probabilities]
\label{con1}
There is a family of (uniform) efficiently sampleable distributions $\cC = \{\cC_n\}_{n\in\bbN}$ over quantum circuits $C$ where each $C$ has $n$-qubit outputs (and potentially additional junk qubits), such that there exist polynomials $p(\cdot)$ and $\gamma(\cdot)$ satisfying:
    \begin{enumerate}
        \item \textbf{Anticoncentration. }  For all large enough $n\in\bbN$$$\Prr_{\substack{C \leftarrow \cC_n\\x\leftarrow\bin^n}}\left[\Pr_C[x] \geq \frac{1}{p(n)\cdot2^n}\right] \geq \frac{1}{\gamma(n)}$$
        \item \textbf{ Hardness of Approximating Probabilities. } For any oracle $\cO$ satisfying that for all large enough $n \in \bbN$, 
        \[
            \Prr_{\substack{C \leftarrow \cC_n\\x\leftarrow\bin^n}}\left[ \left|\cO(C, x) - \Pr_{C}[x]\right| \leq \frac{\Pr_{C}[x]}{p(n)} \right] 
            \geq \frac{1}{\gamma(n)}-\frac{1}{p(n)}
        \]
        we have that $\mathsf{P}^{\#\mathsf{P}} \subseteq \mathsf{BPP}^{\cO}$.\footnote{Our theorem statements remain unchanged if the hardness reduction is a $\mathsf{BQP}^{\cO}$ (instead of a $\mathsf{BPP}^{\cO}$) machine. In fact, we could even allow the hardness reduction to be a $\mathsf{BPP}^{\mathsf{NP}^\cO}$ machine -- all this would change for our results is that the worst-case assumption that we make, which currently says $\mathsf{P}^{\#\mathsf{P}} \not\subseteq \mathsf{BQP}$ would change to $\mathsf{P}^{\#\mathsf{P}} \not\subseteq \mathsf{BPP}^{\mathsf{NP}^\mathsf{BQP}}$.} 
Here, $\Pr_{C}[x]$ denotes the probability of obtaining outcome $x$ when the output register of $C\ket{0}$ is measured in the standard basis.
    \end{enumerate}
\end{assumption}
As described above, different models of quantum advantage are based on the conjectured $\#\mathsf{P}$-hardness of computing output probabilities for different circuit families; all of these imply Assumption \ref{con1}. Some concrete, well-studied candidates include:
\begin{itemize}
\item {\bf Random Circuit Sampling.} One of the leading candidates for quantum advantage today assumes the hardness of    approximating output probabilities of random (universal) quantum circuits (see e.g.~\cite{boixoetal,BFNV19}). Here, the distibution $\cC$ corresponds to a circuit with gates drawn from some universal set and according to some specified architecture, and quantum advantage follows as long as Assumption \ref{con1} holds for $\cC$. 
Understanding the underlying architectures and attempting to prove/disprove the corresponding conjecture is now the focus of a large, and quickly growing, body of work (see e.g.~\cite{boixoetal,BFNV19,kondofocs,boulandnoise,krovi22,movassagh,Zlokapa2023}).

\item {\bf BosonSampling.} Aaronson and Arkhipov~\cite{AA11} related the task of finding probabilities of outcomes of a BosonSampling experiment to computing the permanents of appropriate random matrices. 
    They formulated the following two conjectures that together, form the complexity-theoretic basis for advantage based on BosonSampling.
    The {\em Permanent of Gaussians Conjecture (PGC)} posits that it is $\#\mathsf{P}$-hard to approximate the permanent of a matrix of independent random $\mathcal{N} (0, 1)$ Gaussian entries, and the {\em Permanent Anti-Concentration Conjecture (PACC)} says that with high probability over a matrix $A$ sampled randomly as above, $|\mathsf{Per}(A)| \geq \sqrt{n!}/\mathsf{poly}(n)$. The PGC and PACC  imply Assumption \ref{con1}.

\item {\bf IQP, and beyond.} Other non-universal models of quantum computation, such as Instantaneous Quantum Polynomial (IQP) and Deterministic Quantum Computation with one quantum bit (DQC1) are also candidates for advantage due to their potential ease of implementation on near-term quantum devices~\cite{shepherd2009temporally,knill,mori}. Again, quantum advantage assumes the hardness of approximating probabilities when $\cC$ corresponds to these types of circuits, along with anti-concentration of circuit outputs, which implies Assumption \ref{con1}. 
\end{itemize}

One way to state our main theorem is:
\begin{theorem} [Informal]
\label{thm1}
    Suppose Assumption \ref{con1} is true. Then, one-way puzzles (which imply quantum commitments) exist if and only if $\mathsf{P}^{\#\mathsf{P}} \not\subseteq \mathsf{(io)BQP}/\mathsf{qpoly}$. 
    
\end{theorem}

We also provide a slightly different statement of our main theorem. Assumption \ref{con1}, together with the (extremely mild) assumption that $\mathsf{P}^{\#\mathsf{P}} \not\subseteq \mathsf{ioBQP}/\mathsf{qpoly}$\footnote{We note that $\mathsf{P}^{\#\mathsf{P}} = \mathsf{P}^{\mathsf{PP}}$ and $\mathsf{PP} \subseteq \mathsf{BQP}/\mathsf{qpoly}$ implies a collapse of the counting hierarchy to $\mathsf{QMA}$~\cite{Aar06,Yirka}.} implies the following assumption, which is a reformulation of Assumption \ref{con1} to (1) require the adversary to succeed only on \textit{infinitely many} $n$ instead of \textit{all} large enough $n$, (2) only ask that the probabilities be hard to approximate for non-uniform QPT machines (instead of requiring the approximation to be $\#\mathsf{P}$-hard), and (3) remove the anti-concentration requirement (i.e., $\Pr_{C_n}[x] \geq \frac{1}{p(n)2^n}$) while instead, sampling challenge instances $x$ according to the output distribution of $C_n$. 
We state the assumption below.

\setulcolor{red}

\begin{assumption}[Native Approximation Hardness]
\label{asmptn1}
    There exists a family of (uniform) efficiently sampleable distributions $\cD = \{\cD_n\}_{n\in\bbN}$ over classical strings 
    such that
    there exists a polynomial $p(\cdot)$ such that 
    for all QPT $\cA = \{\cA_\secpar\}_{\secpar \in \bbN}$, every (non-uniform, quantum) advice ensemble $\ket{\tau} = \{\ket{\tau_n}\}_{n \in \mathbb{N}}$, and large enough $n \in \bbN$, 
\[
    \Prr_{x\leftarrow \cD_n}\left[ \left( \left|\cA(\ket{\tau},x) - \Pr_{\cD_n}[x]\right| > \frac{\Pr_{\cD_n}[x]}{p(n)} \right) \right] 
    \geq \frac{1}{p(n)}
\]
\end{assumption}

Our main theorem can be restated as:
\begin{theorem}
\label{infthm1}
(Informal)
    Assumption \ref{asmptn1} is equivalent to the existence of one-way puzzles, and implies the existence of quantum bit commitments.
\end{theorem}

One may wonder whether Assumptions \ref{con1} and \ref{asmptn1} really are {\em weaker} than the existence of (post-quantum) one-way functions. 
First, our equivalence in Theorem \ref{infthm1} also shows that Assumption \ref{asmptn1} is implied by the existence of one-way functions (which trivially imply puzzles).

Next, the 
the hardness of approximating probabilities in Assumption \ref{con1}
is based on conjectured hardness beyond the polynomial hierarchy. The resulting one-way puzzles and commitments therefore lie squarely in Microcrypt. 

Additionally, suppose that Assumption \ref{con1} is stated as a problem: ``for a particular quantumly sampleable distribution (say RCS/BosonSampling), approximate the probabilities of outputs upto low relative error in the average case''. Do we expect the hardness of this problem to imply the existence of one-way functions? 
Suppose there existed a one-way function $f$ and a BPP reduction that used an inverter for $f$ as a black-box to approximate probabilities for the distribution in Assumption~\ref{con1}. This would imply that these probabilities can be approximated in $\mathsf{BPP}^{\mathsf{NP}}$.
However, all existing strategies to {\em prove} that these distributions exhibit sampling based advantage (e.g.,~\cite{AA11}) rely on these probabilities being {\em hard to approximate, on average} by a $\mathsf{BPP}^{\mathsf{NP}}$ machine. Thus, conjectures made in context of quantum advantage already imply that these assumptions will not yield one-way functions (under BPP reductions)\footnote{Note that the typical way we build classical cryptography from hard mathematical problems {\em is} via BPP reductions. But as already discussed, the conjectured $\#\mathsf{P}$ hardness of the problem implies that we do not expect $\mathsf{BQP}$ or even $\mathsf{PH}$ reductions to be able to use an inverter for a one-way function to approximate these probabilities.}.


Furthermore, our results are tight -- quantum advantage is {\em necessary} for one-way puzzles in Microcrypt.
Namely, for any one-way puzzle that does not also imply a one-way function, the sampler $Q$ outputting puzzle and solution pairs cannot be (uniformly) classically simulated. 
Otherwise, a classical machine $C$ that samples from a distribution that is close in statistical distance to the output of $Q$ (for all large enough input lengths) will yield a one-way function as follows: on input $x$, return the puzzle output by $C(x)$. We formally prove that this yields a one-way function in Appendix \ref{app:owpowf}.
This also means that any improvements to our results, e.g., basing one-way puzzles {\em only} on $\#\mathsf{P}$-hardness (without the need for additional conjectures) would yield sampling-based quantum advantage from only minimal worst-case complexity assumptions such as the non-collapse of PH, and without the need for unproven conjectures. 

Finally, a powerful consequence of these results is that they serve as a tool to reduce the existence of other cryptosystems that become insecure in the presence of $\#\mathsf{P}$ oracles, to the existence of quantum commitments, as we will demonstrate with state synthesis hardness in Section \ref{sec:owpss}.

\subsubsection{One-way Puzzles and the Hardness of (Pseudo)-Deterministic Sampling}
Given the connection developed above between quantum advantage conjectures and quantum cryptography, it is natural to wonder how far we can push this connection.
For example, how generically can one claim that the existence of quantum advantage implies the existence of quantum commitments?

While we do not know how to build quantum cryptography in Microcrypt generically from the assumption that $\mathsf{(Samp)BQP} \neq \mathsf{(Samp)BPP}$, we do in fact obtain one-way puzzles if we assume the existence of distributions that cannot be {\em pseudo-deterministically} sampled\footnote{By pseudo-deterministic, we mean that the sampler when run multiple times on the same input, with high probability outputs the same result on all executions.
}.

We show that distributions that admit a quantum sampler but do not admit a {\em pseudo-deterministic} quantum sampler imply the existence of quantum commitments. 

In more detail, assume that there exists a distribution $\cD$ that can be efficiently quantumly sampled, such that every {\em pseudo-deterministic} efficient quantum sampler fails to sample from any distribution $\cD'$ with $\SD(\cD,\cD') \leq \epsilon$ (where $\SD$ denotes statistical distance), and $\epsilon = \frac{1}{p(n)}$ for some fixed polynomial $p(\cdot)$). We show that the existence of any such distribution implies the existence of one-way puzzles. 
\begin{theorem}\label{thm:psd}
[Informal]
    If there exists a quantumly sampleable distribution that does not admit a pseudo-deterministic sampler, then one-way puzzles and quantum commitments exist.
\end{theorem}

Thus, if one-way puzzles do not exist, then every quantumly sampleable distribution can essentially be treated as being deterministically sampleable. 
We believe that this outlook may in the future yield simple ways to build cryptography from hardness assumptions associated with quantum sampleable distributions.

Our next set of results relates the hardness of state synthesis with one-way puzzles and commitments. This also builds on the conceptual connection between the hardness of approximating probabilities and quantum cryptography.

\subsubsection {One-way Puzzles and the Hardness of State Synthesis}
\label{sec:owpss}
Another fundamental ``quantum'' search problem is what we will call a {\em state puzzle} -- this is like a one-way puzzle except that the hard-to-find `key' is not classical, but an arbitrary quantum state. A state puzzle consists of a quantum polynomial time algorithm $\cG$ that samples a (secret) quantum state $\ket{\psi_s}$ along with a (public) string $s$ -- here w.l.o.g. we may assume that for every $s$, $\ket{\psi_s}$ is pure. The security guarantee is that given the ``puzzle'' $s$, it is infeasible for non-uniform QPT machines to synthesize any state that noticeably overlaps with $\ket{\psi_s}$ in expectation.
An extremely natural question, that we address in this work, is whether quantum search problems such as state puzzles also imply one-way puzzles and quantum commitments.

The complexity of synthesizing known
quantum states has been studied in several works~\cite{aaronsonsynth,irani,Rosenthal24}, and it is known that a state puzzle can be inverted by a $\mathsf{BQP}$ machine with access to a $\#\mathsf{P}$ oracle~\cite{aaronsonsynth,irani}.
Thus, the existence of state puzzles implies that 
$\mathsf{BQP}^{\#\mathsf{P}} \not\subseteq \mathsf{(io)BQP}/\mathsf{qpoly}$.

By our previous result (Informal Theorem \ref{thm1}), we know that Assumption \ref{con1} implies the existence of quantum commitments, as long as $\mathsf{BQP}^{\#\mathsf{P}} \not\subseteq \mathsf{(io)BQP}/\mathsf{qpoly}$\footnote{While we state Theorem \ref{thm1} as assuming that $\mathsf{P}^{\#\mathsf{P}} \not\subseteq \mathsf{(io)BQP}/\mathsf{qpoly}$, we note that this is infact equivalent to $\mathsf{BQP}^{\#\mathsf{P}} \not\subseteq \mathsf{(io)BQP}/\mathsf{qpoly}$.}. 
From the discussion in the previous paragraph, we can replace  $\mathsf{BQP}^{\#\mathsf{P}} \not\subseteq \mathsf{(io)BQP}/\mathsf{qpoly}$ with the existence of state puzzles, in the previous sentence.
Thus, if Assumption \ref{con1} holds, then state puzzles imply one-way puzzles.  
However, Assumption \ref{con1} is about average-case hardness of approximating probabilities, and state puzzles capture the average-case hardness of synthesizing states. So can state puzzles imply one-way puzzles and quantum commitments unconditionally? We show that this is indeed the case, and that Assumption \ref{con1} is not needed for this implication.

\begin{theorem}[Informal]
\label{inf3}
    State puzzles imply quantum bit commitments.
\end{theorem}
As an intermediate step, we again build one-way puzzles unconditionally from state puzzles. 
To obtain one-way puzzles, we introduce novel techniques that reduce phase estimation to calculating probabilities of outcomes of a distribution generated using the state itself.
This in fact proves the following equivalence:
\begin{theorem}
[Informal]
\label{inf4}
 The existence of state puzzles is equivalent to the existence of one-way puzzles.
\end{theorem}

\noindent {\bf Public-Key Quantum Money.}
Quantum money aims to use states as banknotes, leveraging the no-cloning principle to prevent counterfeiting.
In a simplified model (often called a ``public-key quantum money mini-scheme''~\cite{AarChr}), a sampler outputs a banknote $\ket{\psi_s}$ together with a classical serial number $s$ that can be efficiently obtained from $\ket{\psi_s}$ without disturbing it. Furthermore, it is computationally intractable to generate ``clones'' of $\ket{\psi_s}$.
It is easy to see that a mini-scheme is also a state puzzle, since if synthesizing $\ket{\psi_s}$ were easy given $s$, then one could clone a banknote $\ket{\psi_s}$ efficiently by first computing $s$ and synthesizing $\ket{\psi_s}$ from $s$.
This observation combined with Theorems \ref{inf3} and \ref{inf4} yields the following corollary.
\begin{corollary}
\label{cor:one-one}
    Quantum money mini-schemes imply one-way puzzles and quantum bit commitments.
\end{corollary}
We stress that in general state puzzles appear to be a weaker primitive than quantum money -- unlike money, they (1) do not require unclonability, only the hardness of generating $\ket{\psi_s}$ given $s$, and (2) do not require $\ket{\psi_s}$ to be efficiently verifiable with respect to $s$.

\paragraph{Amplifying State Puzzles.}
It is also natural to consider a `weak' variant of a state puzzle (analogous to weak one-way functions), where 
it is computationally intractable to synthesize any state that overwhelmingly overlaps with $\ket{\psi_s}$ in expectation.
Our implication from state puzzles to one-way puzzles holds even when starting with a weak state puzzle, which combined with Theorem \ref{inf4} yields the following amplification theorem for state puzzles. 
\begin{theorem}
    Weak state puzzles are equivalent to (strong) one-way puzzles and state puzzles.
\end{theorem}

\subsection{Perspective}

Finally, we reflect on the implications of these results in the broader context of understanding hardness in quantum cryptography.\\

\noindent {\bf Microcrypt is ``Real''.}
As already discussed above, this work provides the strongest evidence so far for the existence of Microcrypt: there are real, unrelativized constructions of quantum cryptosystems based on well-studied mathematical assumptions that do not imply the existence of one-way functions (as otherwise, quantum advantage claims break down). Under mild complexity assumptions ($\mathsf{P}^{\#\mathsf{P}} \not\subseteq \mathsf{ioBQP}^{\mathsf{NP}}$), the resulting constructions remain secure even against an adversary that can access an $\mathsf{NP}$ oracle. 
This also rules out {\em non-black-box constructions} of one-way functions from one-way puzzles or quantum commitments (with a black-box reduction), assuming any of the probability approximation conjectures. For such a reduction would be able to use an $\mathsf{NP}$ oracle --  that inverts any one-way function -- to also solve a $\#\mathsf{P}$-hard problem, which cannot happen unless $\mathsf{P}^{\#\mathsf{P}} \subseteq \mathsf{ioBQP}^{\mathsf{NP}}$.
Previous oracle results inherently only ruled out {\em black-box/relativizing constructions} (with black-box reductions) of one-way functions from quantum commitments and puzzles.

We also note that the mathematical assumptions/conjectures underlying our primitives are not new and have previously been extensively studied in the completely different context of quantum advantage. It is only their application to the realm of cryptography that is new. \\

\noindent {\bf A Sharper Understanding of Microcrypt.}
Since disproving the probability approximation conjectures will have far-ranging consequences in quantum advantage, let us assume that at least one of these conjectures (originally made in the context of quantum advantage) is true. Then the existence of one-way puzzles is equivalent to $\mathsf{P}^{\#\mathsf{P}} \not\subseteq \mathsf{(io)BQP}/\mathsf{poly}$.

This gives us a sharper perspective on the hardness of primitives in Microcrypt. 
For one, any primitives that become insecure in the presence of a $\#\mathsf{P}$ oracle (e.g. state puzzles, but also any new primitives we come up with in the coming years) will imply one-way puzzles and therefore also, commitments.

This equivalence also enables new insights into separations. For instance, one-way puzzles (and commitments) should not imply quantum cryptographic primitives that break in the presence of a $\mathsf{QMA}$ oracle, assuming $\mathsf{P}^{\#\mathsf{P}} \not\subseteq \mathsf{(io)BQP}^{\mathsf{QMA}}$ (see e.g.~\cite{Vyalyi} for some evidence in support of this assumption). Concretely, a public-key quantum money (PKQM) mini-scheme implies one-way puzzles/commitments (by Corollary \ref{cor:one-one}) but becomes insecure given a search-$\mathsf{QMA}$ oracle, and therefore likely should not be implied by one-way puzzles/commitments.
Besides PKQM, other natural primitives that break in the presence of a search-$\mathsf{QMA}$ oracle include public key encryption (PKE) with classical public keys (and quantum secret keys/ciphertexts),  digital signatures (DS) with classical verification keys (and quantum secret keys/signatures). Under any of the hardness of approximation conjectures, $\mathsf{P}^{\#\mathsf{P}} \not \subseteq \mathsf{(io)}\mathsf{BQP}^{\mathsf{search}\text{-}\mathsf{QMA}}$ will imply a separation between these forms of PKQM/PKE/DS and one-way puzzles/commitments.\\

\noindent {\bf On Minimal Assumptions for Quantum Cryptography.}
How hard is it to break quantum cryptography? This work strengthens evidence that it is at least $\mathsf{P}^{\#\mathsf{P}}$-hard to break one-way puzzles (and therefore also, quantum commitments). 

While the bound of $\mathsf{P}^{\#\mathsf{P}}$ is tight for one-way puzzles\footnote{One-way puzzles can be broken by a $\mathsf{P}^{\#\mathsf{P}}$ machine~\cite{CGGHLP}.}, it is not known to be tight for quantum commitments. In fact, recent work~\cite{LMW} demonstrated (relative to a random oracle) the existence of quantum commitments that remain secure against all efficient adversaries that make only a single query to an arbitrary Boolean oracle. 
They also conjecture that their single query lower bound extends to polynomially many queries, and therefore that there is {\em no} classical oracle $\cO$ such that $\mathsf{BQP}^{\cO}$ will break a commitment.
Proving or disproving this conjecture remains an open problem -- in the small chance that $\mathsf{BQP}^{\#\mathsf{P}}$ does successfully break commitments, then our results would say that commitments also imply one-way puzzles (assuming the probability approximation conjectures, although we suspect that the conjectures may not be necessary just like the case of state puzzles). 
If their conjecture is true, then there is also the fascinating possibility that (computationally secure) quantum commitments can be constructed {\em unconditionally}, i.e. without the need for unproven assumptions. However, so far, unconditionally secure commitments are only known in models where participants have access to (inefficiently prepared) quantum auxiliary input~\cite{Qian24,MNY24}.\\

\noindent{\bf Beyond Cryptography.}
This work shows that BosonSampling and Random Circuit Sampling are not as ``useless'' as they are often made out to be. Our constructions use the ability to efficiently sample from these distributions in a crucial way, although we also need some additional quantum processing to finally obtain useful cryptography.
Our work also uncovers a connection between sampling-based quantum advantage and the complexity of decoding Hawking radiation. In more detail, Assumption \ref{con1} together with the mild assumption that $\mathsf{P}^\mathsf{\#\mathsf{P}} \not\subseteq \mathsf{BQP}/\mathsf{qpoly}$ implies the existence of hard-to-decode Hawking radiation.
This follows by combining this work with a theorem of~\cite{C:Brakerski23} demonstrating an equivalence between the existence of quantum commitments and the hardness of black-hole radiation decoding. 

\subsection{Open Problems}
We will now examine some remaining technical obstacles to gaining a complete understanding of hardness in quantum cryptography.
\begin{enumerate}
\item{\bf Concrete Instantiations of Other Microcrypt Primitives.}
An important open challenge is to base other quantum primitives like pseudorandom states and unitaries, digital signatures, quantum money, etc. on concrete, ideally well-studied, mathematical assumptions weaker than one-way functions. 
For example: can we prove, under quantum advantage conjectures, that the output of random circuits are pseudorandom states? Note that such a claim can only hold for specific circuit architectures: for instance, BosonSampling outputs are trivially distinguishable from random~\cite{AA14}. 
Other primitives such as quantum money mini-schemes will likely require different mathematical assumptions that may not be as hard as $\mathsf{P}^{\#\mathsf{P}}$, but also do not necessarily imply the existence of one-way functions.
\item {\bf Connections between Quantum Advantage and Cryptography.} 
Ours is not the first work to use tools developed in the context of quantum advantage to obtain evidence for quantum cryptography in Microcrypt. Previously,~\cite{KQST} built on the oracle separation of $\mathsf{BQP}$ and $\mathsf{PH}$~\cite{Raz-Tal} to demonstrate, {\em relative to an oracle}, that quantum commitments exist even when $\mathsf{P} = \mathsf{NP}$. This work develops similar connections without relying on oracles.
How far does this relationship between quantum advantage and cryptography in Microcrypt extend? For example, it may not be outrageous to imagine that the classical hardness of factoring implies a one-way puzzle, although we do not (yet) know how to prove this.

\item {\bf Fully Quantum Search Problems.} 
We now know how to build commitments from various classical-quantum search problems: one-way state generators~\cite{C:MorYam22,My22} -- where the challenge is a quantum state and the solution is a classical key -- imply commitments~\cite{KT24,BJ24}; and so do state puzzles, where the challenge is classical and the solution is a quantum state. These types of search problems are meaningful because they are often easily implied by other cryptographic primitives and protocols, and are much easier to cryptanalyze than decision problems. For example as demonstrated by this work, one-way puzzles form a useful link between conjectures in quantum advantage and the existence of quantum commitments.
A useful next step towards understanding the relationships between cryptographic primitives
is to study computational search problems where both the challenge and the solution are quantum states. 
\end{enumerate}

\section{Related Work}
\setulcolor{black}
We discuss related independent work below, in decreasing order of technical overlap.

Qian, Raizes and Zhandry~\cite{MRQZ} independently obtain commitments from state puzzles via completely different techniques\footnote{QRZ verbally announced their results, to certain private audiences, before us. 
However, we only learned that they proved an implication from state puzzles to commitments towards the completion of our manuscript. Techniques in both works were developed independently and the manuscripts were publicly posted concurrently.}.
Their techniques generalize to also imply commitments from a variant of state puzzles where the classical challenge is replaced by a clonable basis state\footnote{We suspect that our techniques
can also be extended (e.g., using efficient shadow tomography~\cite{HKP20}) to
yield one-way puzzles from such primitives, but we leave a formal exploration to future work.}. 
    On the other hand, our techniques yield a stronger theorem statement about state puzzles: we show that \ul{even weakly hard} state puzzles imply \ul{one-way puzzles, strongly hard state puzzles} and commitments. Using measurements to obtain a ``classical handle'' on quantum states is a technique that is unique to our work, and this allows us to obtain an implication to one-way puzzles as well as an amplification theorem for state puzzles. We also present candidates for state puzzles assuming $\#\mathsf{P}$-hardness and conjectures from quantum advantage. Other parts of our paper, including building cryptography from $\#\mathsf{P}$ hardness and well-studied conjectures, have no technical overlap with their work.

Two other concurrent and independent works connect quantum cryptography with meta-complexity~\cite{Eli-concurrent,morimae-another-concurrent}. Their techniques also establish an equivalence between the hardness of approximating probabilities with the existence of one-way puzzles (~\cite{morimae-another-concurrent} prove only one direction of the equivalence). The proof of our first result also establishes such an equivalence,
and this is the only part of our paper that overlaps with~\cite{Eli-concurrent,morimae-another-concurrent}. 
We note that in order to build puzzles from quantum advantage conjectures, we require (and obtain) a stronger implication: we build puzzles from the hardness of approximating probabilities upto $(1 \pm 1/\mathsf{poly})$ multiplicative approximation factors for arbitrary polynomials $\mathsf{poly}$,  whereas techniques in~\cite{Eli-concurrent,morimae-another-concurrent} appear to be limited to constant factor approximations only.
Additionally, unlike the assumptions we use, there is no independent evidence that the meta-complexity assumptions studied in~\cite{Eli-concurrent,morimae-another-concurrent} are separated from the existence of one-way functions.

Two additional works~\cite{mori-other,John-et-al} are thematically related, but do not have technical overlap with our work. The first,~\cite{mori-other}, proves an equivalence between interactive (inefficiently verifiable/sampling-based) quantum advantage and (only) {\em classically-secure} one-way puzzles. 
Classically secure puzzles are a new primitive introduced in~\cite{mori-other} and are {\em not} known to imply classical or quantum cryptographic primitives (such as commitments). Therefore,~\cite{mori-other} do not build useful cryptography from the existence of quantum advantage, or from conjectures in quantum advantage.
Additionally in this work, we observe that if one-way puzzles exist but one-way functions do not, this immediately implies advantage (see Appendix \ref{app:owpowf}). Combining this with prior work~\cite{MYadv2} obtaining interactive advantage from one-way functions, makes it easy to see why puzzles imply advantage.
The other work~\cite{John-et-al} introduces {\em new} hardness assumptions that have not been previously scrutinized, and uses these to build pseudoentangled states, quantum trapdoor functions and pseudorandom unitaries (on the other hand our goal is to rely on well-founded assumptions). Unlike our work, it is not clear whether their assumptions are separated from the existence of one-way functions.

In summary, our primary conceptual contribution -- building quantum cryptography from well studied assumptions weaker than one-way functions, and our primary technical contri\-bu\-tion -- converting state puzzles to a classical search problem by carefully replacing a $\#\mathsf{P}$ oracle, do not overlap with any other work.

\section{Technical Overview}
We now provide an overview of our techniques. 

Recall that our dream goal is to build quantum cryptography from the hardness of $\#\mathsf{P}$, for which computing permanents is a complete problem. An immediate obstacle to building cryptography from the hardness of computing permanents is that it appears difficult to efficiently sample random matrices together with their permanents. If such sampling were (quantumly) efficiently possible, we would be done: we would set our one-way puzzle to be the matrix, and the corresponding solution to its permanent. 

In the absence of such samplers, we turn to the rich literature on sampling based quantum advantage (e.g.,~\cite{shepherd2009temporally,BJS,AA11,bremneraverage,fujiicommuting,boixoetal,BFNV19,kondofocs,boulandnoise,krovi22,movassagh,Zlokapa2023}). This line of work obtains quantum advantage from $\#\mathsf{P}$ hardness by building on the following observations:
\begin{enumerate} 
\item There exist quantumly sampleable distributions $\cX$ such that $\Pr_{\cX}[x]$  -- i.e., the probability assigned by $\cX$ to any string $x$ -- equals the square of the permanent of a corresponding unitary matrix. 
Moreover, permanents are known to be $\#\mathsf{P}$-hard to compute on average, and even $\#\mathsf{P}$-hard to {\em approximate}, upto inverse polynomial relative error, in the worst case. 
\item If there existed a {\em classical sampler} that was able to sample {\em exactly} from $\cX$, then $\mathsf{P}^{\#\mathsf{P}} = \mathsf{BPP}^{\mathsf{NP}}$ which is highly implausible because it would collapse the polynomial heirarchy (by Toda's theorem). 

Why would the existence of an exact classical sampler imply that $\mathsf{P}^{\#\mathsf{P}} = \mathsf{BPP}^{\mathsf{NP}}$? This is because given any deterministic sampler $\cO$ for $\cX$, universal hashing makes it possible to approximate $\Pr_{\cX}[x]$ for every $x$ to within a multiplicative constant in $\mathsf{BPP}^{\mathsf{NP}^{\cO}}$~\cite{stockmeyer}. Since there is at least one $x$ for which approximating $\Pr_{\cX}[x]$ is $\#\mathsf{P}$-hard, this puts $\mathsf{P}^{\#\mathsf{P}} \subseteq \mathsf{BPP}^{\mathsf{NP}^{\cO}}$.
\end{enumerate}
Ruling out the possibility of classically {\em approximately} sampling from $\cX$ is not as straightforward. Suppose there exists a classical sampler that samples from a distribution $\cY$ such that $\SD(\cX,\cY) \leq \epsilon$ for some small constant $\epsilon$. The arguments above break down because this sampler may lead to large errors in approximating $\Pr_{\cX}[x]$ for certain $x$, and it is no longer possible to rely only on the {\em worst case} hardness of approximating $\Pr_{\cX}[x]$. What is done instead is that permanents are conjectured to be $\#\mathsf{P}$-hard to approximate even in the average case (see, e.g.,~\cite{AA11}). This conjecture, combined with a type of rerandomization or {\em hiding} property 
of the sampler implies that probabilities $\Pr_{\cX}[x]$ are $\#\mathsf{P}$-hard to {\em approximate on average for uniform choice of $x$}. 

Furthermore, $\Pr_{\cX}[x]$ are assumed to {\em anti-concentrate}, i.e. they must not be too small too often -- suppose that an overwhelming fraction of strings $x$ had $\Pr_{\cX}[x] < \frac{1}{2^{2n}}$, then a classical sampler could sample from a statistically close distribution while still assigning incorrect probabilities (say  $\frac{1}{2^{2n}}$) to an overwhelming fraction of the $x$ values.

In particular, the following assumption cleanly captures the  hardness implied by a variety of known sampling-based quantum advantage conjectures for RCS, BosonSampling, IQP, DQC, etc.

\newtheorem*{T1}{Assumption~\ref{con1}}

\newtheorem*{T2}{Assumption~\ref{asmptn1}}

\begin{T1}
[$\#\mathsf{P}$-Hardness of Approximating Probabilities]
There is a family of (uniform) efficiently sampleable distributions $\cC = \{\cC_n\}_{n\in\bbN}$ over quantum circuits $C$ where each $C$ has $n$-qubit outputs (and potentially additional junk qubits), such that there exist polynomials $p(\cdot)$ and $\gamma(\cdot)$ satisfying:
    \begin{enumerate}
        \item \textbf{Anticoncentration. }  For all large enough $n\in\bbN$$$\Prr_{\substack{C \leftarrow \cC_n\\x\leftarrow\bin^n}}\left[\Pr_C[x] \geq \frac{1}{(p(n)\cdot2^n)}\right] \geq \frac{1}{\gamma(n)}$$
        \item \textbf{Hardness of Approximating Probabilities. } For any oracle $\cO$ satisfying that for all large enough $n \in \bbN$, 
        \[
            \Prr_{\substack{C \leftarrow \cC_n\\x\leftarrow\bin^n}}\left[ \left|\cO(C, x) - \Pr_{C}[x]\right| \leq \frac{\Pr_{C}[x]}{p(n)} \right] 
            \geq \frac{1}{\gamma(n)}-\frac{1}{p(n)}
        \]
    \end{enumerate}
we have that $\mathsf{P}^{\#\mathsf{P}} \subseteq \mathsf{BPP}^{\cO}$. 
Here, $\Pr_{C}[x]$ denotes the probability of obtaining outcome $x$ when the output register of $C\ket{0}$ is measured in the standard basis.
\end{T1}

Assumption \ref{con1} yields a new route to obtaining one-way puzzles, as we describe next. For conceptual reasons, we will begin by reformulating the one-way puzzle primitive in terms of the hardness of post-selected sampling.

\paragraph{The Hardness of Post-Selected Sampling implies One-way Puzzles.}
For efficiently sampleable distribution $\cX$, consider the task of {\em post-selected} sampling from $\cX$. Namely, a challenger samples $x^* \leftarrow \cX$, $i \in [n]$, and outputs $x^*$ truncated to its first $i-1$ bits: $(x^*_{[1\ldots i-1]})$. The post-selected sampling task is to sample from $x \leftarrow \cX$ conditioned on their first $(i-1)$ bits of $x$ being $x^*_{[1\ldots i-1]}$.

Observe that the hardness of post-selected sampling immediately implies the existence of a (distributional) one-way puzzle:
the puzzle sampler samples $x = x_{[1\ldots n]} \leftarrow \cX$, $i \leftarrow [n]$ and outputs $\mathsf{puz} = i, x_{[1\ldots i-1]})$ and $\sol = (x_{[i\ldots n]})$.
Sampling from the distribution over  $\sol$ corresponding to a puzzle $\mathsf{puz}$ is exactly the task of post-selected sampling.
We will make use of this implication (together with the fact that distributional one-way puzzles imply one-way puzzles \cite{CGG24}) in the upcoming subsections.

Finally, we note that the hardness of post-selected sampling is related to the hardness of {\em universal extrapolation}~\cite{FOCS:ImpLev90}, and the hardness of approximating probabilities is related to the hardness of {\em universal approximation}~\cite{OstWig,FOCS:ImpLev90}. The existence of universal extrapolators and universal approximators for classically sampleable distributions is known to be equivalent to the existence of one-way functions~\cite{FOCS:ImpLev90}. This work will implicitly show and exploit analogous equivalences between one-way puzzles and the hardness of universal approximation/extrapolation for quantumly sampled distributions.

\subsection{One-way Puzzles from the Hardness of Approximating Probabilities} 
Our first key insight is that an {\em exact} post selected sampler makes it easy to compute probabilities of strings, upto inverse polynomial multiplcative error. 
In more detail, given an {\em exact} post-selected sampler $\sS_{\cX}$ for $\cX$, there is a polynomial-time machine $\sR$ parameterized by a polynomial $p(\cdot)$ -- that with oracle access to $\sS_{\cX}$ and on input a string $v \in \{0,1\}^n$ -- outputs an approximation of $\Pr_{\cX}[v]$ that is accurate upto inverse polynomial relative error.\\

\underline{$\sR^{\sS_\cX}(v):$}
\begin{itemize}
\item Parse $v$ as a sequence of bits $v[1] v[2] \ldots v[n]$.
\item Set $\mathsf{prefix} = \bot$, and set $\mathsf{pr} = 1$.
\item For $i \in [n]$, do the following:
\begin{itemize}
    \item Run $\sS$ on input $\mathsf{prefix}$ $p(n)$ times, and let $\mathsf{pr}_{i}$ denote the fraction of times that the first bit of the output is $v[i]$. 
    \item Set $\mathsf{pr} = \mathsf{pr} \cdot \mathsf{pr}_{i}$, and $\mathsf{prefix} = \mathsf{prefix}||v[i]$.
\end{itemize}
\item Output $\mathsf{pr}$.
\end{itemize}
By Chernoff bounds, as long as for each prefix (denoted by $\mathsf{prefix}_{v}$) of $v$, $\Pr_{\cX}[\mathsf{prefix}_v] > \frac{1}{\mathsf{q}(n)} \cdot \frac{1}{2^{|\mathsf{prefix}_v|}}$ for some fixed polynomial $q(\cdot) < p(\cdot)$, the reduction $\sR$ above outputs an approximation to $Pr_{\cX}[v]$ with small inverse polynomial relative error.
On the other hand, if $\Pr_{\cX}[v] \ll \frac{1}{\mathsf{q}(n)} \cdot \frac{1}{2^n}$ the reduction can fail, so we do not get a good approximation of probabilities in the worst case, and are only able to contradict {\em average-case hardness} of approximating $\Pr_{\cX}[x]$ (Assumption \ref{asmptn1}). In the actual analysis, we crucially use the fact that challenges $v$ that are sent to $\sR$, are sampled according to the distribution $\cX$, and therefore strings $v$ for which $\Pr_{\cX}[v] \ll \frac{1}{\mathsf{q}(n)} \cdot \frac{1}{2^n}$ are sampled with relatively low probability.

Note that the analysis described so far applies if we had a {\em perfect} post-selected sampler. But in the definition of a (distributional) one-way puzzle, not only do we want the hardness of sampling {\em exactly} from the target distribution, we also need it to be hard for adversaries to sample from a distribution $\cD'$ that is inverse-polynomially close (in statistical distance) to the target distribution.
This gives an adversary/post-selected sampler the flexibility to introduce arbitrary errors in sampling, while maintaining overall low statistical distance from the target distribution.
However, note that this latter requirement enforces that the adversary can only introduce high relative errors on values $v$ (and their prefixes) for which $\Pr_{\cX}[v]$ is low.
We use this observation together with a more sophisticated analysis to show that the reduction $\sR$ described above will still output a low relative error approximation to $\Pr_{\cX}[x]$ on average, as long as $\Pr_{\cX}[x]$ is not too small (i.e., $\Pr_{\cX}[v] \geq \frac{1}{\mathsf{q}(n)} \cdot \frac{1}{2^n}$).

This, combined with Assumption \ref{con1} yields a quantum polynomial time machine that solves $\#\mathsf{P}$-hard problems, contradicting the assumption that $\mathsf{P}^{\#\mathsf{P}} \not\subseteq \mathsf{(io)BQP}/\mathsf{qpoly}$.
This completes a sketch of our proof that Assumption \ref{con1} implies one-way puzzles.
Our actual proof first further abstracts out the properties we need from Assumption \ref{con1} along with $\mathsf{P}^{\#\mathsf{P}} \not\subseteq \mathsf{(io)BQP}/\mathsf{qpoly}$ into a different assumption (described below). The analysis above is then applied to prove that Assumption \ref{asmptn1} implies one-way puzzles.

\begin{T2}
 [Native Approximation Hardness]
    There exists a family of (uniform) efficiently sampleable distributions $\cD = \{\cD_n\}_{n\in\bbN}$ over classical strings 
    such that
    there exists a polynomial $p(\cdot)$ such that 
    for all QPT $\cA = \{\cA_\secpar\}_{\secpar \in \bbN}$, every (non-uniform, quantum) advice ensemble $\ket{\tau} = \{\ket{\tau_n}\}_{n \in \mathbb{N}}$, and large enough $n \in \bbN$, 
\[
    \Prr_{x\leftarrow \cD_n}\left[ \left( \left|\cA(\ket{\tau},x) - \Pr_{\cD_n}[x]\right| > \frac{\Pr_{\cD_n}[x]}{p(n)} \right) \right] 
    \geq \frac{1}{p(n)}
\]
\end{T2}

We point out one important technical issue that arises from the mismatch between the complexity-theoretic style of Assumption \ref{con1} and the cryptographic style of Assumption \ref{asmptn1}.
We would like to use a $\mathsf{BQP}/\mathsf{qpoly}$ adversary that contradicts Assumption \ref{asmptn1} to show that $\mathsf{P}^{\#\mathsf{P}} \subseteq \mathsf{(io)BQP}/\mathsf{qpoly}$, as long as Assumption \ref{con1} holds.
To contradict Assumption \ref{asmptn1}, the $\mathsf{BQP}/\mathsf{qpoly}$ adversary only needs to succeed in approximating probabilities on infinitely many $n \in \bbN$. On the other hand, Assumption \ref{con1} converts any adversary that approximates probabilities {\em for every large enough $n \in \bbN$} into an oracle that can solve $\#\mathsf{P}$-hard problems for large enough $n$. It is at first unclear why these two statements can be put together to obtain the implication we want, i.e., $\mathsf{P}^{\#\mathsf{P}} \subseteq \mathsf{(io)BQP}/\mathsf{qpoly}$. But on opening things up, we observe that Assumption 1 guarantees a black-box reduction that on input length $n$, queries an approximator adversary on polynomially many input lengths, each polynomially related to $n$. By modifying our distribution for Assumption \ref{asmptn1} to output samples for each of these input lengths, we ensure that the approximator adversary answers all of the reductions queries correctly, infinitely often. This allows us to conclude that $\mathsf{P}^{\#\mathsf{P}} \subseteq \mathsf{(io)BQP}/\mathsf{qpoly}$, as desired.

\paragraph{An Equivalence between Puzzles and the Hardness of Approximating Probabilities.} 
We also prove a reverse implication, i.e., the existence of one-way puzzles implies that Assumption \ref{asmptn1} holds. 
To prove this, we would like to use one-way puzzles to define a distribution $\cD$ such that we can invert the puzzle given a probability estimator for strings in the support of $\cD$.

Defining the distribution to be the same as the output of the one-way puzzle sampler does not work. This is because even given an estimator that perfectly computes probabilities $100\%$ of the time, it is unclear how to find a key $k$ corresponding to a puzzle $s$ by using an oracle that on input any $(s||k)$ outputs $\Pr_{\mathsf{Samp}}[(s||k)]$.

Instead, following~\cite{CGGHLP}, we will aim to perform a binary search for $k$, given $s$.
Indeed as a first attempt, our distribution $\cD$ will be defined as follows: run the puzzle sampler $\mathsf{Samp}$ to obtain puzzle and key $(s,k)$, sample $i \leftarrow [n]$ and then output $(i,s,k_{1\ldots i})$ where $k_{1\ldots i}$ denotes a truncation of $k$ to the first $i$ bits. 

Now given $s$, it may at first appear easy to search for consecutive bits of $k$ using a probability estimator for $\cD$: first, run the estimator on $(1,s,0)$ to obtain $\mathsf{pr}_0$ and $(1, s, 1)$ to obtain $\mathsf{pr}_1$. Pick bit $b$ for which $\mathsf{pr}_{b} \geq \mathsf{pr}_{1-b}$, and set the first bit of $k$, i.e. $k_1$ to $b$. Next, run the estimator on $(1,s,k_10)$ to obtain $\mathsf{pr}_{k_1,0}$ and $(1, s, k_11)$ to obtain $\mathsf{pr}_{k_1,1}$. Pick bit $b'$ for which $\mathsf{pr}_{k_1,b'} \geq \mathsf{pr}_{k_1,1-b'}$, set the second bit of $k$ to $b'$, and continue the process to iteratively find a key $k$. 
Indeed, this works if the estimator always returns correct probabilities, even on strings that are not in support of the distribution.

However, our assumption only requires the adversarial probability approximator to return (approximately) correct probabilities on strings that have non-zero probability mass. The construction and analysis above breaks down given such an approximator: in particular, we began by running the estimator on $(1, s, 0)$ and $(1, s, 1)$ to see if keys for $s$ had a higher probability of beginning with $0$ or with $1$. Note that if every actual preimage key of $s$ had first bit $0$, the point $(1, s, 1)$ would be assigned zero probability mass in the distribution, meaning the adversary may return arbitrary values on $(1, s, 1)$ to confuse our search algorithm, without any penalty. The same problem can arise even if we have extremely low, but non-zero probability mass on certain points.
We resolve this by modifying the distribution, as we describe next.

We will run the puzzle sampler $\mathsf{Samp}$ to obtain $(s, k)$, sample $i \leftarrow [n]$, and sample bit $c \leftarrow \{0,1\}$.
If $c = 0$, set $\widetilde{k} = k_{1 \ldots i}$ and if $c=1$ set $\widetilde{k} = k_{1 \ldots {i-1}}||\beta$ for $\beta \leftarrow \{0,1\}$.
Output $(i,s,\widetilde{k})$. Intuitively adding some probability mass to both $0$ and $1$ on the last bit, we force an adversary to pay a penalty in statistical distance whenever they send the binary search algorithm described above down an incorrect path.
The analysis requires some additional care to account for various types of errors, but we are able to conclude that any probability approximator for the above distribution implies an inverter for the one-way puzzle.

We refer the reader to Section \ref{sec:dist} for additional details and a formal proof of the equivalence.

\subsection{One-way Puzzles from the Hardness of Pseudo-Deterministic Sampling} Next, we discuss why the existence of a distribution $\cX$ that is quantumly efficiently sampleable, but is not efficiently {\em pseudo-deterministically} sampleable implies the existence of one-way puzzles. 

An $\epsilon$-pseudo-deterministic sampler is a QPT machine $\cQ$ along with (non-uniform, quantum) advice ensemble $\ket{\tau} = \{\ket{\tau_n}\}_{n \in \mathbb{N}}$, that satisfies the following property for all large enough $n \in\bbN$: for at least $1 - \epsilon(n)$ fraction of random strings $r \in \{0,1\}^n$,  
    $$\exists y \text{ s.t. }\Pr\left[\cQ(\ket{\tau};r) = y \right] = 1- \negl(n)$$

A distribution $\cX$ is $\epsilon$-{\em pseudo-deterministically} sampleable with $\epsilon(\cdot)$ error if there exists an efficient quantum $\epsilon$-pseudo-deterministic sampler that outputs distribution $\cD$ such that $\SD(\cX, \cD) \leq \epsilon(n)$. 

We prove that if one-way puzzles do not exist, then for every polynomial $q(\cdot)$, every quantumly sampleable distribution is also $\frac{1}{q(n)}$-pseudo-deterministicaly sampleable.

Our key insight here is that if post-selected sampling is easy, then every distribution can be pseudo-deterministically sampled by using the post-selected sampler to approximate probabilities. We describe a reduction, parameterized by a polynomial $p(\cdot)$ that with access to  post-selected sampler $\cS_\cX$ for $\cX$, samples pseudo-deterministically from $\cX$.\\

\underline{$\sR^{\sS_\cX}(r):$}
\begin{itemize}
\item Parse $r$ as a sequence of blocks of randomness $r[1] r[2] \ldots r[n]$, each block of size $n$ bits. 
\item Set $\mathsf{prefix} = \bot$.
\item For $i \in [1,n]$, do the following:
\begin{itemize}
    \item Run $\sS$ on input $\mathsf{prefix}$ $p(n)$ times, and let $\mathsf{pr}_{i}$ denote the fraction of times that the first bit of the output is $0$.
    \item If $r[i] \leq \mathsf{pr}_i \cdot 2^{n}$, set $x[i] = 0$. Otherwise set $x[i] = 1$.
    \item Set $\mathsf{prefix} = \mathsf{prefix}||x[i]$.
\end{itemize}
\item Output $\mathsf{prefix}$.
\end{itemize}
This insight can be turned into a formal proof that for every polynomial $q(\cdot)$, there is a polynomial $p(\cdot)$ such that $\cR$ parameterized with $p(\cdot)$ samples $\frac{1}{q(n)}$-pseudo-deterministically from $\cX$.
Using analysis that is similar to the previous section, we prove that the next-bit probabilities computed by $\cR$ are approximately correct (on average). This lets us show that except for some bad choices of randomness (which are close to actual probability thresholds), the output of $\cR$ is (almost) deterministic. Moreover, the distribution output by $\cR$ has inverse polynomial statistical distance from $\cX$ as long as $\cS$ is a perfect post-selected sampler. 

As before, the non-existence of one-way puzzles only guarantees that the adversary can sample from a distribution that is (arbitrarily) inverse-polynomially close to the post-selected distribution. With some more care, we are able to show that the reduction $\cR$ described above, even given access to such an adversary, will pseudodeterministically sample from $\cX$. This completes an overview of our technique, and we encourage the reader to see Section \ref{sec:pseudo} for a complete proof.

\subsection{One-Way Puzzles from Hard State Synthesis Problems}
Finally, we use the conceptual equivalence between the existence of one-way puzzles and the hardness of approximating probabilities to demonstrate an equivalence between one-way puzzles and a natural notion of hard state synthesis problems, which we call state puzzles. 

A state puzzle consists of a QPT sampler $\cG$ that outputs pairs $(s, \ket{\psi_s})$ such that given $s$, it is quantum computationally infeasible to output a state that overlaps noticeably with $\ket{\psi_s}$.\\

\noindent{\bf State Puzzles with Real, Positive Amplitudes.}
Consider any state $\ket{\psi} = \sum_x \alpha_x \ket{x}$ where all amplitudes $\alpha_x$ are real and positive. Measuring this state results in a distribution over computational basis terms, i.e. a distribution $D_\psi$ where $\Pr_{D_\psi}[x] = |\alpha_x|^2$.
Intuitively, the hardness of computing probabilities in $D_\psi$ should be related to the hardness of synthesizing $\ket{\psi}$. 
We use this conceptual connection to obtain an equivalence between state puzzles and one-way puzzles, as follows.

Let us begin by recalling a procedure due Aaronson~\cite{aaronsonsynth} that calls a classical ($\mathsf{PP}$) oracle to synthesize a state. Let $m(\cdot)$ be a large enough polynomial, and $\cO$ be an oracle that on input $(x,i)$ outputs the value $|\alpha_{x_{1\ldots i}\|1}|^2/(|\alpha_{x_{1\ldots i}\|0}|^2+|\alpha_{x_{1\ldots i}\|1}|^2)$, where $x_{1\ldots i}\|b$ denotes $x$ truncated to its first $i$ bits then concatenated with $b$, and for any $\ell \leq [n], t\in\bin^\ell$, $|\alpha_{t}|^2 = \sum_{z\in \{0,1\}^{n-\ell}} |\alpha_{t||z}|^2$. 
Intuitively, the oracle outputs the probability that a string sampled from $D_\psi$ will have $1$ as its $(i+1)^{th}$ bit given that the first $i$ bits are $x_{1\ldots i}$. Call this probability $p_{1|x_{1\ldots i}}$.

\begin{itemize}
    \item Set $i = 0$. Initialize register $\mathsf{A}$ to $\ket{0^n,i}$ and initialize an auxiliary register $\mathsf{B}$ to $\ket{0^{m(n)}}$.
    \item While $i \leq n$,
    \begin{itemize}
        \item Query the oracle $\cO$ on the contents of the $\mathsf{A}$ register  and CNOT the result on the $\mathsf{B}$ register. Denote the the resulting state by $$\sum_{x \in \{0,1\}^{i-1}} \beta_x \ket{x,0^{n-i+1},i}_{\mathsf{A}} \ket{p_{1|x_{1\ldots i}}}_{\mathsf{B}}.$$
        \item Apply an efficient unitary to the previous state to obtain state 
        \[ \sum_{x \in \{0,1\}^{i}} \beta_x \ket{x} ( \beta_{x0}\ket{0} + \beta_{x1} \ket{1})  \ket{0^{n-i},i}_{\mathsf{A}}\ket{p_{1|x_{1\ldots i}}}_{\mathsf{B}}.
       \] 
       where
      $\beta_{x0} = \sqrt{1-p_{1|x_{1\ldots i}}}$ and 
      $\beta_{x1} = \sqrt{p_{1|x_{1\ldots i}}}$.
    \item Use a related call to uncompute the auxiliary information and obtain 
        $$\sum_{x \in \{0,1\}^{i+1}} \beta_x \ket{x,0^{n-i},i}_{\mathsf{A}}\ket{0^{m(n)}}_{\mathsf{B}}$$ 
        \item Set $i = i +1$, also update the last part of $\mathsf{A}$ to $\ket{i+1}$.
     \end{itemize}
\end{itemize}
It is straightforward to observe that this process results in a state close to the desired state, upto precision errors in the probabilities.
Our goal will be to simulate this procedure with access to an adversary breaking an appropriately defined distributional one-way puzzle.

The one-way puzzle sampler runs the state puzzle sampler $\cG$ to obtain $(s, \ket{\psi_s})$. 
It then measures $\ket{\psi_s}$ in the standard basis to obtain string $x$,  and samples $i \leftarrow [0,n-1]$. Finally, it outputs $(s, i, x_{1 \ldots i})$ as the puzzle, and $x_{i+1}$ as the key.

Assume there exists a perfect distributional inverter for this one-way puzzle. On input $s$, the reduction queries the one-way puzzle inverter iteratively for $i \in [n]$, starting with a state initialized to $\ket{0^n, 0} \otimes \ket{0^{m(n)}}$. At each step, the reduction queries the one-way puzzle inverter to obtain (coherently) for each basis string $x_{1 \ldots i}$, multiple samples of the next bit $x_{i+1}$ distributed according to the target distribution. These samples are then used to obtain an estimate $p_{x_i}$ of the probability $|\alpha_{x_i||0}|^2/(|\alpha_{x_i||0}|^2+\alpha_{x_i||1}|^2)$, which is then used to build the state iteratively for each $i$
by applying the same unitary as the one in Aaronson's algorithm described above. 

Here, the one-way puzzle inverter may entangle its output on every query with arbitrary junk states on an auxiliary register; and we need to be able to uncompute these junk states. We cannot apply the standard trick of CNOT-ing our probability estimates on a fresh register and uncomputing, since the CNOT will end up disturbing the adversary's state and uncomputing may not remove junk.
However, since the probability estimate $p_{1|x_{1 \ldots i}}$ is guaranteed to be close to the actual probability $|\alpha_{x_{1\ldots i}||0}|^2/(|\alpha_{x_{1 \ldots i}||0}|^2+\alpha_{x_{1 \ldots i}||1}|^2)$, we are able to use this estimate to compute each step of the synthesis algorithm (i.e., use the probability estimate to insert appropriate amplitudes on $\ket{x_{1\ldots i}0}$ and $\ket{x_{1\ldots i}1}$ coherently for each $x_{1\ldots i}$) and remove junk at the end of each step. 

This process applied iteratively for $i \in [0,n-1]$ yields a state whose trace distance from $\ket{\psi}$ is at most $1/q(n)$ for arbitrary polynomial $q(\cdot)$, as long as the one-way puzzle inverter is $1/p(n)$-close to the target distribution for some polynomial $p(\cdot)$ related to $q(\cdot)$.\\

\noindent {\bf Recovering Phase Information.}
The technique described above is limited to states with real, positive amplitudes. We need to work harder when the states to be synthesized carry non-trivial phase information. In particular, the one-way puzzle cannot be based only on measuring the state $\ket{\psi_s}$ in the standard basis.
Instead, our one-way puzzle will be obtained by randomly choosing to either measure the state in the standard basis as before, or measuring phase information. 

Given any state puzzle sampler $\cG$, the one-way puzzle sampler does the following.

\begin{itemize}
    \item Obtain $(s, \ket{\zeta_s}) \leftarrow \cG(1^n)$.
    \item Sample a random 2-design $C$ and let $\ket{\psi_s} = C(\ket{\zeta_s})$.
    \item Sample $c \leftarrow \{0,1\}$.
    \item If $c = 0$, then as before, measure $\ket{\psi_s}$ in the standard basis to obtain string $x$. Sample $i \leftarrow [0,n-1]$. Output $(s,C,c,i,x_{1\ldots i})$ as the puzzle and $x_{i+1}$ as the key, where $x_i$ denotes the first $i$ bits of $x$ and $x[i+1]$ denotes the $(i+1)^{th}$ bit of $x$.
    \item If $c = 1$, 
    choose a two-to-one function $f$ defined by a random shift $\Delta$, i.e. $f(x) = f(x \oplus \Delta)$, then apply $f$ to the register containing $\ket{\psi_s}$ and measure the output to obtain a residual state of the form $$\left(\cos (\theta/2) \ket{x_0} + \sin (\theta/2) e^{-i\phi}\ket{x_1}\right)_{\mathsf{A}} \otimes \ket{f, f(x_0)}_\mathsf{B},$$
    for some $x_0, x_1 = x_0 \oplus \Delta$, and some $\theta \in [0,\pi), \phi \in [0,2\pi)$.
    Next sample bit $d \leftarrow \{0,1\}$, and
    \begin{itemize}
        \item If $d = 0$, measure the $\mathsf{A}$ register in basis $(\ket{x_0} + \ket{x_1}, \ket{x_0} - \ket{x_1})$ to obtain bit $z$. 
        \item If $d = 1$, measure the $\mathsf{A}$ register in basis $(\ket{x_0} + i \ket{x_1}, \ket{x_0} - i \ket{x_1})$ to obtain bit $z$. 
    \end{itemize}
    Output $(s,C,c,x_0,x_1,d)$ as the puzzle and $z$ as the key.
\end{itemize}

Denote by $\ket{\widetilde{\psi}_s}$ the state that corresponds to  $\ket{\psi_s}$ with its phase information removed. That is, 
$$\ket{\widetilde{\psi}_s} = \sum_x a_x \ket{x}$$ and 
$$\ket{\psi_s} = \sum_x \alpha_x \ket{x}$$ 
where every $a_x$ is real and positive and
$\alpha_x = a_x e^{i \phi_x}$ for
$\phi_x \in [0,2\pi)$. 

As already described above, a perfect distributional puzzle inverter for the corresponding input can be queried on $c=0$ to synthesize a state close to $\ket{\widetilde{\psi}_s}$. We discuss how the puzzle inverter queried on $c=1$ can be used to find and insert phases $e^{i \phi_x}$ on every basis term $\ket{x}$ in $\ket{\widetilde{\psi}_s}$.

Recall that when $c=1$, the puzzle is generated by applying a two-to-one function and measuring its output, which results in state $$\ket{\psi}_{f,x_0} :=(\cos (\theta/2) \ket{x_0} + \sin (\theta/2) e^{-i\phi}\ket{x_1})_{\mathsf{A}} \otimes f, f(x_0)_\mathsf{B}.$$
Measuring $\ket{\psi}_{f,x_0}$ in basis 
$$\ket{x_0} + \ket{x_1}, \ket{x_0} - \ket{x_1}$$
results in outcome $\ket{x_0} + \ket{x_1}$ with probability $(1+\sin \theta \cos \phi)/{2}$; and in basis 
$$\ket{x_0} + i\ket{x_1}, \ket{x_0} - i\ket{x_1}$$
results in $\ket{x_0} + i\ket{x_1}$ with probability $(1+\sin \theta \sin \phi)/{2}$.

The puzzle itself is $(s, C, c, x_0, x_1, d)$ where $d \leftarrow \{0,1\}$. If $d=0$, the key is the outcome of measuring $\ket{\psi}_{f,x_0}$ in 
the first basis,
which is $0$ with probability $(1+\sin \theta \cos \phi)/{2}$ and $1$ otherwise. If $d=1$, the key is the outcome of measuring $\ket{\psi}_{f,x_0}$ in 
the second basis, which is $0$ with probability $(1+\sin \theta \sin \phi)/{2}$ and $1$ otherwise.

The unifying technique in this work is to use a one-way puzzle inverter to approximate probabilities, and that is exactly what we will do at this point. For fixed $f$ and $x_0$, we will use the one-way puzzle inverter to estimate the probability values $(1+\sin \theta \cos \phi)/{2}$ and $(1+\sin \theta \sin \phi)/{2}$ corresponding to the state $\ket{\psi}_{f, x_0}$; which will then help us approximate the relative phase $e^{i\phi}$ between the basis terms $\ket{x_0}$ and $\ket{x_1}$ in $\ket{\psi}$. Here,  $x_1 = x_0 \oplus \Delta$ for $\Delta$ indicating the shift chosen by function $f$.
This gives us a way to use the puzzle inverter to compute the relative phase between pairs of terms in $\ket{\psi_s}$. Next, we would like to  ``insert'' the resulting phases in the state $\ket{\widetilde{\psi}_s}$ to recover our state $\ket{\psi_s}$.

For this, let us first imagine sampling and fixing a uniformly random anchor $y \in \{0,1\}^n$.
We will then compute the phase on every standard basis term $\ket{x}$ in $\ket{\psi_s}$, relative to $\ket{y}$. Namely, we coherently estimate the relative phases $\phi_{x_0,y}$ for all $\ket{\psi}_{x_0, f}$ -- where $f$ applies the shift $\Delta = x_0 \oplus y$. 
This yields the state
$$\ket{\widetilde{\psi}_s} = \sum_x \alpha_x \ket{x} \ket{\widetilde{\phi}_{xy}}$$
which can be transformed into 
$$\ket{\widetilde{\psi}_s} = \sum_x \alpha_x e^{i\widetilde{\phi}_{xy}}\ket{x}$$ by applying an efficient unitary that applies the phase and then uncomputes $\widetilde{\phi}_{xy}$.
As long as the estimates $\widetilde{\phi}_{xy}$ were approximately correct, this state is close (upto global phase) to the state $\ket{\widetilde{\psi}_s}$.\\

\noindent {\bf On Errors in Phase Estimation.} 
Note that the choice of anchor $y$ in the process above affects errors: for instance, if $\langle y | \psi_s \rangle = 0$, the puzzle inverter is allowed to return arbitrary values that may be completely uncorrelated with actual phases in $\ket{\psi_s}$. So if we picked an anchor $y$ for which $\langle y | \psi_s \rangle = 0$, we could end up synthesizing a state that is orthogonal to $\ket{\psi_s}$. The same problem persists even for anchors $y$ for which $\langle y | \psi_s \rangle$ is extremely small, since the puzzle inverter is not perfect, and is allowed to sample from a distribution that has inverse polynomial statistical distance from the target distribution.
To synthesize a state close to $\ket{\psi_s}$ in the presence of these errors, we sample our anchor $y$ by measuring the state $\ket{\widetilde{\psi}_s}$ in the standard basis, which outputs $y$ proportionally to the  probability mass of $\ket{y}$ in $\ket{\psi_s}$.

The accuracy of the phase estimate $\widetilde{\phi}_{xy}$ depends on the relative probability mass on $\ket{x}$ and $\ket{y}$ in the state. 
In particular, to meaningfully recover an estimate of the  phase $\widetilde{\phi}_{xy}$ with small relative error, we require the ratio $\alpha_x/\alpha_y$ to be at most $p(n)$ for some polynomial $p(\cdot)$. This is where the Clifford operator helps: since $\ket{\psi_s}$ was obtained by applying a random Clifford operator to $\ket{\zeta_s}$, we can use properties of $2$-designs along with a careful step-wise Chebyshev bound to argue that the total probability mass on basis terms $\ket{y}$ for which $|\alpha_y|^2 > p(n) \cdot 2^n$ for some polynomial $p(\cdot)$, is less than $1/q(n)$ for a related polynomial $q(\cdot)$.
This allows us to condition our analysis on obtaining  anchors $y$ for which $|\alpha_y|^2 \leq p(n) \cdot 2^n$.
With some additional care in the analysis, we are able to bound the error in reconstructing the state $\ket{\psi_s}$ as a function of the error allowed to the one-way puzzle inverter. In particular, we can reconstruct $\ket{\psi_s}$ to arbitrary small inverse polynomial error ${1}/{p_1(n)}$ as long as the one-way puzzle inverter samples from a distribution that is at most ${1}/{p_2(n)}$-far from the target distribution, for a related polynomial $p_2(\cdot)$.

This gives us an equivalence between weak state puzzles and distributional one-way puzzles.
Noting that distributional one-way puzzles are equivalent to (strong) one-way puzzles, this implies an equivalence between weak state puzzles and one-way puzzles.
In the following subsection, we describe why one-way puzzles imply state puzzles (with pure state secrets), thereby yielding an amplification theorem for state puzzles.\\

\noindent{\bf One-way Puzzles Imply (Strong) State Puzzles.}
By definition, one-way puzzles imply a form of state puzzles where the state $\ket{\psi_s}$ is replaced with a mixture over classical keys of a one-way puzzle. By purifying this mixture, we obtain a strong state puzzle with pure states. Combined with the results described above, this shows that weak state puzzles are equivalent to strong state puzzles.
We refer the reader to Section \ref{sec:owpsp} for a formal statement of the equivalence, and proofs of results related to state puzzles.

\section{Preliminaries}
In this section, we discuss some notation and preliminary information, including definitions, that will be useful in the rest of the exposition.

\subsection{Notation and Conventions}

We write $\negl(\cdot)$ to denote any \emph{negligible} function, which is a function $f$ such that for every constant $c \in \mathbb{N}$ there exists $N \in \mathbb{N}$ such that for all $n > N$, $f(n) < n^{-c}$.
We will use $\mathsf{SD}(A,B)$ to denote the statistical distance between (classical) distributions $A$ and $B$.\\

\noindent{\bf Quantum conventions.} A register $\sX$ is a named Hilbert space $\bbC^{2^n}$. A pure state on register $\sX$ is a unit vector $\ket{\psi} \in \bbC^{2^n}$, and we say that $\ket{\psi}$ consists of $n$ qubits. A mixed state on register $\sX$ is described by a density matrix $\rho \in \bbC^{2^n \times 2^n}$, which is a positive semi-definite Hermitian operator with trace 1. 

A \emph{quantum operation} $F$ is a completely-positive trace-preserving (CPTP) map from a register $\sX$ to a register $\sY$, which in general may have different dimensions. That is, on input a density matrix $\rho$, the operation $F$ produces $F(\rho) = \tau$ a mixed state on register $\sY$.
A \emph{unitary} $U: \sX \to \sX$ is a special case of a quantum operation that satisfies $U^\dagger U = U U^\dagger = \bbI^{\sX}$, where $\bbI^{\sX}$ is the identity matrix on register $\sX$. A \emph{projector} $\Pi$ is a Hermitian operator such that $\Pi^2 = \Pi$, and a \emph{projective measurement} is a collection of projectors $\{\Pi_i\}_i$ such that $\sum_i \Pi_i = \bbI$.

We say a quantum circuit $C$ outputs strings in $\bin^n$ if $C$ acts on $\ket{0}$ to produce an $n$-qubit output register (potentially along with a junk register). The output of the circuit is the outcome of measuring the output register of $C\ket{0}$ in the computational basis.

\begin{theorem}[Additive Chernoff Bound]\label{thm:chernoff-additive} For every $i \in [n]$, let $X_i$ be an independent Bernoulli random variable that takes value $1$ with probability $p$. 
Let $X:= \sum_i X_i/n$. Then for $\delta>0$:
\[
\Pr[|X - p| \geq \delta/\sqrt{n}] \leq 2e^{-\delta^2}
\]
\end{theorem}
\begin{definition}[Unitary 2-design (from \cite{haar-intro})] \label{def:unitary-2-design}
Let $D$ be a distribution over $n$-qubit unitaries. $D$ is a unitary 2-design if and only if for all $O \in \cL((\bbC^2)^{\otimes 2})$
\[
    \underset{U \leftarrow D}{\bbE}\left[U^{\otimes 2} O U^{\dag\otimes 2}\right] =  \underset{U \leftarrow \mu_H}{\bbE}\left[U^{\otimes 2} O U^{\dag\otimes 2}\right] 
\]  where $\mu_H$ is the Haar measure.
\end{definition}

We also use the following theorem showing that the trace distance of pure states is upper bounded by their Euclidean distance.
\begin{theorem}
\label{thm:trace-dist-and-euclidean-dist}
    Let $\ket{\psi}$ and $\ket{\phi}$ be two pure states such that $|\ket{\psi} - \ket{\phi}| \leq \epsilon$. Then 
    \[
    \TD\left(\ketbra{\psi}, \ketbra{\phi}\right) \leq \epsilon
    \] 
\end{theorem}
\begin{proof}
    We have that
    \begin{align*}
        \epsilon^2 &\geq |\ket{\psi} - \ket{\phi}|^2 \\
        &= (\bra{\psi} - \bra{\phi}) (\ket{\psi} - \ket{\phi}) \\
        &= 2 - (\langle \phi | \psi \rangle + \langle \phi | \psi \rangle) \\
        &= 2 - 2 \text{Re} (\langle \phi | \psi \rangle) \\
        &\geq 2 - 2 |\langle \phi | \psi \rangle|,
    \end{align*}
    which can be rearranged to 
    \begin{align*}
        &|\langle \phi | \psi \rangle| \geq 1 - \frac{\epsilon^2}{2}.
    \end{align*}
    By the identity for trace distance of pure states,
    \begin{align*}
        \TD( \ket{\psi} \bra{\psi}, \ket{\phi} \bra{\phi}) &= \sqrt{1 - |\langle \psi | \phi \rangle|^2} \\
        &= \sqrt{1 - \left(1 - \frac{\epsilon^2}{2} \right)^2} \\
        &= \sqrt{\epsilon^2 - \frac{\epsilon^4}{4}} \\
        &\leq \epsilon.
    \end{align*}
    This completes the proof.
\end{proof}

\subsection{Quantum Cryptographic Primitives}
\begin{definition}[One-way Puzzles]
\label{def:owp}
A one-way puzzle is a pair of sampling and verification algorithms $(\Gen,\mathsf{Ver})$
with the following syntax. 
\begin{itemize}
\item $\Gen(1^n) \rightarrow (s,k)$, is a QPT algorithm that outputs a pair of classical strings $(s,k)$.
We refer to $s$ as the puzzle and $k$ as its key. Without loss of generality we may assume that $k\in\bin^n$.
\item $\mathsf{Ver}(s,k) \rightarrow \top$ or $\bot$,
is a Boolean function that maps every pair of classical strings $(k,s)$ to either $\top$ or $\bot$.
\end{itemize}
These satisfy the following properties.
\begin{itemize}
\item {\bf Correctness.} Outputs of the sampler pass verification with overwhelming probability, i.e., 
$$\Prr_{(s,k) \leftarrow \Gen(1^n)} [\Ver(s,k) = \top] = 1 - \negl(n)$$ 
\item {\bf Security.}
Given $s$, it is (quantum) computationally infeasible to find $k$ satisfying $\Ver(s,k) = \top$, i.e., for every quantum polynomial-sized adversary $\cA$ and every quantum advice state $\ket{\tau} = \{\ket{\tau_n}\}_{n \in \mathbb{N}}$,
 $$\Prr_{(s,k) \leftarrow \Gen(1^n)}[\Ver(s,\mathcal{A}(\ket{\tau},s)) = \top] = \negl(n)$$
\end{itemize}
\end{definition}

\begin{definition}[$\varepsilon$-Distributional One-way Puzzles]
\label{def:dist-owp}
    For $\varepsilon: \bbN \rightarrow \bbR$, a $\varepsilon$-distributional one-way puzzle is defined by a quantum polynomial-time generator $\Gen(1^n)$ that outputs a pair of classical strings $(s,k)$ such that 
    for every quantum polynomial-time adversary $\cA$, every (non-uniform, quantum) advice ensemble $\ket{\tau} = \{ \ket{\tau_n}\}_{n \in \mathbb{N}}$, for large enough $n \in \mathbb{N}$, 
     \[ \mathsf{SD} \left(
    \{s,k\}\, \{s,\cA(\ket{\tau},s)\} \right)
    \geq \varepsilon(n)\] where $(s,k)\leftarrow \Gen(1^n)$.
\end{definition}
We will sometimes simply refer to distributional one-way puzzles. This is taken to mean $1/p(n)$-distributional one-way puzzles for some non-zero polynomial $p$. 

The following theorem shows that distributional one-way puzzles can be amplified to (standard) one-way puzzles.
\begin{theorem}[Theorem 33 from \cite{CGG24}, rephrased]\label{thm:owp-amplification}
If there exists a polynomial $p(\cdot)$ for which $1/p(n)$-distributional one-way puzzles exist, then one-way puzzles exist.
\end{theorem}

\begin{definition}[$\varepsilon$-Pseudo-deterministic Hard Distributions] \label{def:pd-hardness}
    An algorithm $\cB$ is $\varepsilon$-pseudo-deterministic for $\varepsilon: \bbN \rightarrow \bbR$ if 
    \[
         \Prr_{r}[\exists y \text{ s.t. }\Pr\left[\cB(r) \neq y \right] \leq \negl(n)] > 1-\varepsilon(n)
    \] where $r$ is a uniformly random string.
    
    A family of efficiently sampleable distributions $D = \{D_n\}$ is $\varepsilon$-pseudo-deterministic hard if for all quantum polynomial-time $\varepsilon$-pseudodeterministic adversaries $\cA$ that take (non-uniform, quantum) advice ensemble $\ket{\tau} = \{\ket{\tau_n}\}_{n \in \bbN}$, for all
    large enough $n \in \mathbb{N}$
    \[
        \SD(\cA(\ket{\tau};r), D_n) > \varepsilon(n)
    \] where $r$ is a uniformly random string.                                                                                                                                                                                                                                                                                                                                                                               
\end{definition}

\section{Hardness of Approximating Probabilities implies One-way Puzzles}
\label{sec:dist}
In this section we define notions of distributions with hard to approximate probabilities and prove their equivalence with one-way puzzles. 
\subsection{Definitions}
We begin by presenting a hardness assumption that has been extensively studied in the literature on sampling based quantum advantage~\cite{AA11,bremneraverage,fujiicommuting,boixoetal,BFNV19,kondofocs,krovi22,movassagh,qsurvey}). Here, the adversary is given uniformly sampled $x\leftarrow\bin^n$ together with a quantum circuit $C$ on $n$ qubits (with possibly $m = m(n)$ ancillas). The adversary's task is to estimate $\Pr_C[x]$, i.e. the probability that the output state of $C\ket{0^{n+m}}$ results in outcome $x$ when measured in the standard basis.
The adversary is required to output a low \textit{relative error} approximation to $\Pr_C[x]$, i.e. for some polynomial $p(\cdot)$ and for sufficiently many $x$, the adversary must output a value $y$ such that \[\left(1 + \frac{1}{p(n)}\right)\Pr_C[x] \leq y \leq \left(1 - \frac{1}{p(n)}\right) \Pr_C[x]\]

The outputs of circuit $C$ are required to satisfy an additional property called anticoncentration, which says that a noticeable fraction of strings $x$ have probability mass in $C$ above a particular threshold of $1/(p(n)\cdot 2^n)$.

\begin{definition}
[Uniform Approximation Hardness]
\label{def:type-2}
A family of (uniform) efficiently sampleable distributions $\cC = \{\cC_n\}_{n\in\bbN}$ over quantum circuits $C$ that output classical strings in $\bin^n$ has uniform approximation hardness if it has the following properties. There exist polynomials $p(\cdot)$ and $\gamma(\cdot)$ such that:
    \begin{enumerate}
        \item \textbf{Anticoncentration. }  For all large enough $n\in\bbN$$$\Prr_{\substack{C \leftarrow \cC_n\\x\leftarrow\bin^n}}[\Pr_C[x] \geq 1/(p(n)\cdot2^n)] \geq 1/\gamma(n)$$
        \item \textbf{Approximate Hardness of Sampling. } For any oracle $\cO$ satisfying that for all large enough $n \in \bbN$, 
        \[
            \Prr_{\substack{C \leftarrow \cC_n\\x\leftarrow\bin^n}}\left[ \left|\cO(C, x) - \Pr_{C}[x]\right| \leq \frac{\Pr_{C}[x]}{p(n)} \right] 
            \geq 1/\gamma(n)-\frac{1}{p(n)}
        \]
    \end{enumerate}
we have that $\mathsf{P}^{\#\mathsf{P}} \subseteq \mathsf{BPP}^{\cO}$. 
Here, $\Pr_{C}[x]$ denotes that probability the $C$ outputs $x$.
\end{definition}
While Definition \ref{def:type-2} captures the hardness conjectures proposed in the quantum advantage literature, it is slightly inconvenient for cryptographic applications. 
There are a few reasons for this: the biggest one is that only adversaries that succeed for \textit{all} large enough $n$ will help solve $\#\mathsf{P}$. The standard adversarial model for cryptography allows for adversaries that succeed on infinitely many $n$. 
Second, hardness is defined for strings sampled uniformly, whereas for our applications, it will be more suitable to have hardness defined for strings sampled from the output of the circuit itself.
Finally, we only require hardness against all polynomial-sized quantum machines, whereas the definition above requires $\#\mathsf{P}$ hardness of the approximation task.

We will therefore define and build from Definition \ref{def:type-2} (along with the assumption that $\mathsf{P}^{\#\mathsf{P}} \not\subseteq \mathsf{ioBQP}/\mathsf{qpoly}$) a related but more ``crypto-friendly'' primitive that will simplify the implication to puzzles (Theorem \ref{thm:type2-implies-type1}).

\begin{definition}[Native Approximation Hardness]\label{def:type-1} 
A family of (uniform) efficiently sampleable distributions $\cD = \{\cD_n\}_{n\in\bbN}$ over classical strings 
has native approximation hardness if there exists a polynomial $p$ such that for all QPT $\cA = \{\cA_\secpar\}_{\secpar \in \bbN}$, every (non-uniform, quantum) advice ensemble $\ket{\tau} = \{\ket{\tau_n}\}_{n \in \mathbb{N}}$, and large enough $n \in \bbN$, 
\[
    \Prr_{x\leftarrow \cD_n}\left[\left|\cA(\ket{\tau}, x) - \Pr_{\cD_n}[x]\right| \leq \frac{\Pr_{\cD_n}[x]}{p(n)}\right] \leq 1 - \frac{1}{p(n)}
\]
\end{definition}

Next we will show (Theorem \ref{thm:type1-implies-puzzles}) that the existence of distribution families satisfying Definition \ref{def:type-1} implies the existence of one-way puzzles (Definition \ref{def:owp}). Finally we will show the reverse implication (Theorem \ref{thm:puzzles-imply-type1}), namely that the existence of one-way puzzles implies the existence of distribution families satisfying Definition \ref{def:type-1}.

\subsection{Uniform Approximation Hardness Implies Native Approximation Hardness}
\begin{theorem}
\label{thm:type2-implies-type1}
     If $\mathsf{P^{\#P}/qpoly} \not\subseteq \mathsf{ioBQP/qpoly}$ and there exists a family of distributions $\cC = \{\cC_n\}_{n\in\bbN}$ that satisfies Definition $\ref{def:type-2}$, then there exists a family of distributions $\cD = \{\cD_n\}_{n\in\bbN}$ that satisfies Definition $\ref{def:type-1}$.
\end{theorem}
\begin{proof}
    Let $L$ be a $\mathsf{PP}$-complete language. $\cC$ satisfies Definition $\ref{def:type-2}$, so there must exist  exist polynomials $q$ and $\gamma$ along with oracle PPT $R^{(\cdot)}_L$ such that: 
    \begin{enumerate}
        \item For all large enough $n\in\bbN$
        \[
            \Prr_{\substack{C\leftarrow \cC_n\\x\leftarrow \bin^n}}[\Pr_C[x] \geq 1/(q(n)\cdot 2^n)] \geq 1/\gamma(n)
        \]
        \item Let $\cO$ be any oracle such that for all large enough $n\in\bbN$
      \[
            \Prr_{\substack{C \leftarrow \cC_n\\x\leftarrow\bin^n}}\left[ \left|\cO(C,x) - \Pr_{C}[x]\right| \leq \frac{\Pr_{C}[x]}{q(n)} \right] 
            \geq \frac{1}{\gamma(n)}-\frac{1}{q(n)}
        \]
        then for all $n\in\bbN$ and for all $x\in\bin^n$, $\Prr[R^\cO_L(x) = L(x)] \geq 1-\negl(n)$. We call the set of all such oracles $\bbO$.
    \end{enumerate}
    Let ${t_R}$ be a polynomial such that $R^{(\cdot)}_L$ runs in time ${t_R}(n)$.
    
    Since $L$ is $\mathsf{PP}$-complete, $\mathsf{BQP}^L/\mathsf{qpoly} = \mathsf{BQP}^\mathsf{PP}/\mathsf{qpoly} = \mathsf{BQP^{\#P}/qpoly} $. Let $L' \in \mathsf{BQP}^L/\mathsf{qpoly}$ such that $L' \notin \mathsf{ioBQP/qpoly}$. Then there exists an oracle QPT $\cM^{(\cdot)} = \{\cM^{(\cdot)}_\secpar\}_{\secpar \in \bbN}$ and (non-uniform, quantum) advice ensemble $\ket{\sigma} = \{\ket{\sigma_n}\}_{n \in \mathbb{N}}$ such that for all $n \in \bbN$, for all $x\in \bin^n$, 
    \[
        \Prr[\cM^L(\ket{\sigma}, x) = L'(x)] \geq 1-\negl(n)
    \]
    We will drop the advice from the notation since it is always implicitly provided to $\cM$.  Additionally, let ${t_\cM}$ be a polynomial such that $|\cM_n|\leq {t_\cM}(n)$. 

Our overall proof strategy will be to show that (given an adversary that estimates probabilities for a carefully constructed distribution $\cD$) we can replace the oracle to $L$ in $\cM^L$ with an efficient quantum algorithm (with advice). This is accomplished by using $R_L^{(\cdot)}$ with an appropriately constructed oracle to decide $L$ instead. We wish to define $\cD_n$ in such a way that any adversary that estimates the probabililities of $y\leftarrow \cD_n$ will also allow us to estimate $\Pr_C[x]$ for $C\leftarrow \cC_n$ and $x\leftarrow \bin^n$. We may initially try setting $\cD_n$ to be the induced distribution on $(C,x)$. The adversary will in this case provide an estimate of $\Pr_{\cC_n}[C]\cdot\Pr_C[x]$. To compute $\Pr_C[x]$ we therefore also need to compute $\Pr_{\cC_n}[C]$. This is accomplished by having $\cD_n$ consist of two modes, one in which the output is $(C,x)$ and one in which the output is $C$ alone. This allows us to compute estimates for both $\Pr_{\cC_n}[C]\cdot\Pr_C[x]$ and $\Pr_{\cC_n}[C]$ and therefore estimate $\Pr_C[x]$. 

The above strategy is still insufficient for instantiatiating an oracle for $R^{(\cdot)}_L$. This is because an adversary that breaks the security of $\cD_n$ may only do so for infinitely many values of $n$. The reduction $R^{(\cdot)}_L$, on the other hand, requires an oracle that breaks security for every large enough value of $n$. While it is tempting to think that instantiating $R^{(\cdot)}_L$ with an infinitely-often oracle would lead to a reduction that also succeeds infinitely often, this is not necessarily the case. 
$R^{(\cdot)}_L$ may query its oracle on a variety of input sizes, and only succeed if the oracle performs well on all of them. We must therefore be able to answer queries for all input sizes up to $t_R(m)$ when running $R^{(\cdot)}_L$ on inputs of size $m$. Since $R^{(\cdot)}_L$ will ultimately be queried by $\cM_n(x)$ for $x\in\bin^n$, $m$ can be as large as $t_\cM(n)$. 
We therefore modify $\cD_n$ in such a way that we can use an adversary breaking its security to answer queries of all sizes upto $t_R(t_\cM(n))$.

Let ${\thres}$ be any polynomial such that ${\thres}(n) > {t_R}({t_\cM}(n))$ and let $p$ be any polynomial such that $p(n) > nq^4({\thres}(n))\cdot ({\thres}(n))^{3}$.
We define $\cD_n$ as follows:
\begin{itemize}
    \item Sample $\ell \leftarrow [{\thres}(n)]$
    \item Sample $C \leftarrow \cC_\ell$
    \item Sample $x \leftarrow C$
    \item Sample $b_{\text{mode}} \leftarrow \bin$
    \item If $b_\text{mode} = 0$:
    \begin{itemize}
        \item Output $(b_\text{mode}, C, 0^\ell)$
    \end{itemize}
    \item Else:
    \begin{itemize}
        \item Output $(b_\text{mode}, C, x)$
    \end{itemize}
\end{itemize}

We will prove that $\cD$ satisfies Definition \ref{def:type-1}. Suppose for the sake of contradiction that this is not the case. Then there exists a QPT $\cA = \{\cA_\secpar\}_{\secpar \in \bbN}$ and (non-uniform, quantum) advice ensemble $\ket{\tau} = \{\ket{\tau_n}\}_{n \in \mathbb{N}}$ such that for infinitely many $n \in \bbN$, 
   \[
    \Prr_{y\leftarrow \cD_n}\left[\left|\cA(\ket{\tau}, y) - \Pr_{\cD_n}[y]\right| \leq \frac{\Pr_{\cD_n}[y]}{p(n)}\right] > 1 - \frac{1}{p(n)}
\]
Fix any such adversary $\cA$, any such advice ensemble $\ket{\tau}$, and any such large enough $n\in \bbN$. We will drop the advice from the notation since it is always implicitly provided to the adversary. We will now use $\cA$ to build a QPT  $\cB$ that calls $\cA$ internally such that for all $x\in\bin^n$, with high probability, $R^{\cB}(x) = L(x)$. Specifically, we define $\cB$ as the algorithm that takes input $(C,x)$ and returns $\cA(1,C,x)/\cA(0,C,0^{|x|})$.

First we show that we can replace the oracle $\cO$ with $\cB$ for queries of size less than $\thres(n)$. We may view the above inequality as bounding the probability that $\cA(y)$ is too far from $\Pr_{\cD_n}[y]$, i.e.
\[
    \Prr_{y\leftarrow \cD_n}\left[\left|\cA( y) - \Pr_{\cD_n}[y]\right|  > \frac{\Pr_{\cD_n}[y]}{p(n)}\right] < \frac{1}{p(n)}
\]
We can split this error probability into cases indexed by the values of $b_\text{mode}$ and $\ell$. First consider when $b_\text{mode}=0$. For all $\ell\in [{\thres}(n)]$ 
\begin{align*}
    \Prr[b_\text{mode}=0] \cdot \Prr[\ell]\cdot\Prr_{C\leftarrow \cC_\ell}\left[\left|\cA(0, C, 0^\ell) - \Pr_{\cD_n}[0, C, 0^\ell]\right|  > \frac{\Pr_{\cD_n}[0, C, 0^\ell]}{p(n)}\right] < \frac{1}{p(n)}
\end{align*}
which can be simplified to 
\begin{align}
\label{eq:estimating-circuit-prob-1}
\forall \ell\in [{\thres}(n)],     \Prr_{C\leftarrow \cC_\ell}\left[\left|2{\thres}(n)\cdot\cA(0, C, 0^\ell) - \Pr_{\cC_\ell}[C]\right|  > \frac{\Pr_{\cC_\ell}[C]}{p(n)}\right] < \frac{2{\thres}(n)}{p(n)}
\end{align}
Similarly, consider when $b_\text{mode}=1$. For all $\ell\in [{\thres}(n)]$
\begin{align*}
    \Prr[b_\text{mode}=1]\cdot\Prr[\ell] \cdot{\Prr_{\substack{C\leftarrow \cC_\ell\\x\leftarrow C}}}\left[\left|\cA(1, C, x) - \Pr_{\cD_n
    }[1, C, x]\right|  > \frac{\Pr_{\cD_n}[1, C, x]}{p(n)}\right] < \frac{1}{p(n)}
\end{align*}
which can be simplified to 
\begin{align}
\label{eq:estimating-circuit-prob-2}
   \forall \ell \in [{\thres}(n)], \Prr_{\substack{C\leftarrow \cC_\ell\\x\leftarrow C}}\left[\left|2{\thres}(n)\cdot \frac{\cA(1, C, x)}{\Pr_{\cC_\ell}[C]} - \Pr_{C}[x]\right|  > \frac{\Pr_{C}[x]}{p(n)}\right] < \frac{2{\thres}(n)}{p(n)}
\end{align}
Let $\delta:= \sqrt{\frac{2{\thres}(n)}{p(n)}}$. We will now show for every large enough $\ell$ the existence of a ``good'' set of $(C,x)$ that is sampled with high enough probability and for which $\cB(C,x)$ gives a good estimate of $\Pr_C[x]$ with high probability. For all $\ell\in[\thres(n)]$ let $\bbC_\ell$ be defined as follows.
\begin{align*}
    \bbC_\ell:=\left\{C \text{ s.t } \Prr\left[\left|2{\thres}(n)\cdot\cA(0, C, 0^\ell) - \Pr_{\cC_\ell}[C]\right|  > \frac{\Pr_{\cC_\ell}[C]}{p(n)}\right] > \delta\right\}
\end{align*}
Intuitively, $\bbC_\ell$ is the set of $C$ output by $\cC_\ell$ such that with all but $\delta$ probability, $2{\thres}(n)\cdot\cA(0, C, 0^\ell)$ is a good estimate for $\Pr_{\cC_\ell}[C]$.
By a Markov argument on \eqref{eq:estimating-circuit-prob-1}, for all $\ell\in[\thres(n)]$
\[
\Prr_{C\leftarrow \cC_\ell}[C\in \bbC_\ell] \leq \delta
\]
Similarly, for all $\ell\in[\thres(n)]$ let $\bbG_\ell$ be defined as follows.
\begin{align*}
    \bbG_\ell:=\left\{(C,x) \text{ s.t }  \Prr\left[\left|2{\thres}(n)\cdot \frac{\cA(1, C, x)}{\Pr_{\cC_\ell}[C]} - \Pr_{C}[x]\right|  > \frac{\Pr_{C}[x]}{p(n)}\right]  > \delta\right\}
\end{align*}
Intuitively, $\bbG_\ell$ is the set of $(C,x)$ where $C$ is output by $\cC_\ell$ and $x$ is an $\ell$-bit string such that with all but $\delta$ probability, $2{\thres}(n)\cdot \frac{\cA(1, C, x)}{\Pr_{\cC_\ell}[C]}$ is a good estimate for $\Pr_{C}[x]$. By a Markov argument on \eqref{eq:estimating-circuit-prob-2}, for all $\ell\in[\thres(n)]$
\[
\Prr_{\substack{C\leftarrow \cC_\ell\\x\leftarrow C}}[(C,x)\in \bbG_\ell] \leq \delta
\]

\begin{claim}
\label{clm:ioreduction}
    For all large enough $\ell\in[{\thres}(n)]$, let $\bbB_\ell := \left\{(C,x) \text{ s.t } C\notin\bbC_\ell \wedge (C,x) \notin \bbG_\ell\right\}$
    \begin{itemize}
        \item For all $(C,x) \in \bbB_\ell$, $\Prr\left[\left|\cB(C,x) - \Pr_{C}(x)\right| \leq \frac{3\Pr_{C}(x)}{p(n)}\right] \geq 1-2\delta$
        \item $\Pr_{\substack{C\leftarrow \cC_\ell\\x\leftarrow \bin^\ell}}[(C,x) \in \bbB_\ell] \geq 1/\gamma(\ell) - 2\delta\cdot q(\ell)$
    \end{itemize}
\end{claim}
\begin{proof}
 For the first part of the claim, note that by the definitions of $\bbC_\ell$, for all $(C,x)\in\bbB_\ell$, with probability atleast $(1-\delta)$
\[
\left|2{\thres}(n)\cdot\cA(0, C, 0^\ell) - \Pr_{\cC_\ell}[C]\right|  > \frac{\Pr_{\cC_\ell}[C]}{p(n)}
\]
which can be rearranged as
\[
   \frac{\cA(0, C, 0^\ell)}{(1-1/p(n))} < \frac{\Pr_{\cC_\ell}[C]}{2{\thres}(n)} < \frac{\cA(0, C, 0^\ell)}{(1+1/p(n))}
\]
Additionally, by the definitions of $\bbG_\ell$, for all $(C,x)\in\bbB_\ell$, with probability atleast $(1-\delta)$
\[\left|2{\thres}(n)\cdot \frac{\cA(1, C, x)}{\Pr_{\cC_\ell}[C]} - \Pr_{C}[x]\right|  > \frac{\Pr_{C}[x]}{p(n)}\]
which can similarly be rewritten as
\[
    \Pr_{C}[x] \cdot (1-1/p(n)) < 2{\thres}(n)\cdot \frac{\cA(1, C, x)}{\Pr_{\cC_\ell}[C]} < \Pr_{C}[x] \cdot (1+1/p(n))
\]
Both events occur simultanaeously with probability atleast $1-2\delta$, in which case we may apply the bounds for $\frac{\Pr_{\cC_\ell}[C]}{2{\thres}(n)}$ to the previous inequality.
\[
\Pr_{C}[x] \cdot (1-1/p(n))^2 < \frac{\cA(1, C, x)}{\cA(0, C, 0^\ell)} < \Pr_{C}[x] \cdot (1+1/p(n))^2
\]
which for large enough $n$ gives
\[
\left|\Pr_{C}[x] - \frac{\cA(1, C, x)}{\cA(0, C, 0^\ell)}\right| < \frac{3\Pr_{C}[x]}{p(n)} 
\]
which concludes the proof of part 1.

For the second part of the claim, first note that since $\Prr_{\substack{C\leftarrow \cC_\ell\\x\leftarrow C}}[(C,x)\in \bbG_\ell] \leq \delta$ and $\Prr_{C\leftarrow \cC_\ell}[C\in \bbC_\ell] \leq \delta$
\[
\Prr_{C\leftarrow \cC_\ell}[C\in \overline{\bbB}_\ell] \leq 2\delta
\] where $\overline{\bbB}_\ell$ represents the complement of $\bbB_\ell$.
Let $\bbA_\ell := \left\{(C,x) : \Pr_C[x] \geq 1/(q(\ell)\cdot 2^\ell)\right\}$.
By the anticoncentration property we know that for large enough $\ell \in [{\thres}(n)]$
\[
    \Prr_{\substack{C\leftarrow \cC_\ell\\x\leftarrow \bin^\ell}}[(C,x) \in \bbA] \geq 1/\gamma(\ell)
\]
We aim to bound $\Prr_{\substack{C\leftarrow \cC_\ell\\x\leftarrow \bin^\ell}}[(C,x) \in \bbA_\ell \cap \overline{\bbB}_\ell]$ as follows
\begin{align*}
   \Prr_{\substack{C\leftarrow \cC_\ell\\x\leftarrow \bin^\ell}}[(C,x) \in \bbA_\ell \cap \overline{\bbB}_\ell]&= \sum_{(C,x)\in\bbA_\ell \cap \overline{\bbB}_\ell} \Pr_{\cC_\ell}[C] \cdot 1/2^\ell\\
    &\leq \sum_{(C,x)\in\bbA_\ell \cap \overline{\bbB}_\ell} \Pr_{\cC_\ell}[C] \cdot q(\ell)\cdot\Pr_C[x]\\
    &= {q(\ell)} \cdot \Prr_{\substack{C\leftarrow \cC_\ell\\x\leftarrow C}}[(C,x) \in \bbA_\ell \cap \overline{\bbB}_\ell]\\
     &\leq{q(\ell)} \cdot \Prr_{\substack{C\leftarrow \cC_\ell\\x\leftarrow C}}[(C,x) \in \overline{\bbB}_\ell]\\
    &\leq 2\delta\cdot q(\ell)
\end{align*}
where the first step follows from the definition of $\bbA$ and the last step follows from Claim \ref{clm:ioreduction}. We can now bound $\Prr_{\substack{C\leftarrow \cC_\ell\\x\leftarrow \bin^\ell}}[(C,x) \in  \bbB_\ell]$ as follows
\begin{align*}
\Prr_{\substack{C\leftarrow \cC_\ell\\x\leftarrow \bin^\ell}}[(C,x) \in  \bbB_\ell]  &\geq \Prr_{\substack{C\leftarrow \cC_\ell\\x\leftarrow \bin^\ell}}[(C,x) \in \bbA_\ell \cap \bbB_\ell] \\
   &=   \Prr_{\substack{C\leftarrow \cC_\ell\\x\leftarrow \bin^\ell}}[(C,x) \in \bbA_\ell] - \Prr_{\substack{C\leftarrow \cC_\ell\\x\leftarrow \bin^\ell}}[(C,x) \in \bbA_\ell \cap \overline{\bbB}_\ell]\\
   &\geq 1/\gamma(\ell) - 2\delta\cdot q(\ell)
\end{align*}
which concludes the proof.
\end{proof}

\begin{claim}
    For all $m$ such that ${t_R}(m) \leq {\thres}(n)$, for all $x\in \bin^m$
    \[
        \Prr[R_L^{\cB}(x) = L(x)] \geq 1-\negl(m)-2\delta\cdot {\thres}(n)
    \]
\end{claim} 
\begin{proof}
   
    Since $R_L^{(\cdot)}(x)$ runs in time atmost ${t_R}(|x|)$, the size of the longest query to made is atmost ${t_R}(m) \leq {\thres}(n)$. Additionally, the number of queries made is also atmost ${t_R}(m) \leq {\thres}(n)$. By Claim \ref{clm:ioreduction} and by noting that $p(n) > nq^4({\thres}(n))\cdot ({\thres}(n))^{3}$ and $\delta= \sqrt{\frac{2{\thres}(n)}{p(n)}}$, for every $\ell \leq {\thres}(n)$, there exists a set $\bbB_{\ell}$
     \begin{itemize}
        \item For all $(C,x) \in \bbB_{\ell}$, $\Prr\left[\left|\cB(C,x) - \Pr_{C}(x)\right| \leq \frac{\Pr_{C}(x)}{q(\thres(n))}\right] \geq 1-2\delta$
        \item $\Prr_{\substack{C\leftarrow {\cC_\ell}\\x\leftarrow \bin^{\ell}}}[(C,x) \in \bbB_{\ell}] \geq 1/\gamma({\ell}) - 1/q(\thres(n))$
    \end{itemize}
Let $\bbO'$ be the set of oracles such that for every $\ell \leq {\thres}(n)$, for all $(C,x) \in \bbB_{\ell}$,
$$\left|\cO(C,x) - \Pr_{C}(x)\right| \leq \frac{\Pr_{C}(x)}{q(\ell)}$$
Since $\ell \leq t(n)$ it is easy to see that $\bbO' \subseteq \bbO$.
$R^{(\cdot)}_L(x)$ succeeds with probability $1-\negl(m)$ for every $\cO \in \bbO'$, so it succeeds with probability atleast $1-\negl(m)$ given any random distribution over oracles in $\bbO'$. We will now show that with high probability, $\cB$ is such a distribution over oracles. 

Assume WLOG that $R^{\cB}_L(x)$  queries any string atmost once during its execution. We only need to consider queries of length atmost ${\thres}(n)$ since no query is longer than ${\thres}(n)$. If for every query $(C,x) \in \bbB_\ell$ made to $\cB$ where $x\in\bin^\ell$, it was the case that $\left|\cB(C,x) - \Pr_{C}(x)\right| \leq \frac{\Pr_{C}(x)}{q(n)}$ then the outputs of $\cB$ are distributed according to some distribution over $\bbO'$. For each query this occurs independently with probability atleast $1-2\delta$ and there are at most ${\thres}(n)$ queries, so this event occurs with probability atleast $1-2\delta\cdot {\thres}(n)$. Therefore, $\Pr[R_L^\cB(x) = L(x)] \geq 1-\negl(m) - 2\delta\cdot {\thres}(n)$.\footnote{A slightly different version of Definition \ref{def:type-2} requires the oracle to approximate probabilities upto arbitrary relative error $\epsilon$ (given $1^{1/\epsilon}$ as input) and with probability  $1/\gamma({\ell}) - \delta$ (given $1^{1/\delta}$ as input) over the randomness of the input. In this case, note that since the size of the largest query is atmost $\thres(n)$, $\epsilon$ and $\delta$ are atleast $1/\thres(n)$. $\cB$ can therefore still be used to answer such oracle queries by setting $p(n)$ to be large enough to achieve relative error $1/\thres(n)$ with probability $1/\gamma({\ell}) - 1/\thres(n)$ over the randomness of the input, and Theorem \ref{thm:type2-implies-type1} will be unaffected by the change in the definition.}
\end{proof}
Noting that $p(n) > nq^4({\thres}(n))\cdot ({\thres}(n))^{3}$ and $n$ is large enough, we get $2\delta\cdot {\thres}(n) \leq 1/3$. Therefore, by repeating in parallel and taking the majority outcome we obtain an oracle PPT $S^{(\cdot)}$ where for all $m$ such that ${t_R}(m) \leq {\thres}(n)$, for all $y\in \bin^m$, $\Pr[S^\cB(y) = L(y)] \geq 1-\negl(m)$. Finally, we will use $\cM^{\cS^\cB}$ to decide $L'$ on strings of length $n$. For $x\in\bin^n$, $M^{\cS^\cB}_n(x)$ can only query strings of length at most ${t_\cM}(n)$. Since ${t_R}({t_\cM}(n)) \leq {\thres}(n)$, this means that $\cS^\cB$ will correctly answer the queries with all but negligible probability. As a result, for all $x\in\bin^n$, $\Pr[\cM_n^{\cS^\cB}(x) = L'(x)] \geq 1-\negl(n)$.

Finally, since $\cS^\cB$ is a polynomial size quantum circuit with quantum advice,  $\cM^{\cS^\cB}$ may also be expressed as a polynomial size quantum circuit with quantum advice that decides $L'$ for infinitely many input lengths. This contradicts the assumption that $L' \notin \mathsf{ioBQP/qpoly}$ and concludes the proof of the theorem.
\end{proof}

\subsection{Native Approximation Hardness Implies One-Way Puzzles}

\begin{theorem} 
\label{thm:type1-implies-puzzles}The existence of families of distributions that satisfy Definition $\ref{def:type-1}$ implies the existence of one-way puzzles (Definition \ref{def:owp})
\end{theorem}
\begin{proof}
Let $\cD = \{ \cD_\secpar\}_{\secpar \in \bbN}$ be a family of distributions that satisfies Definition \ref{def:type-1}.  Therefore there exists a polynomial $q$ such that for all QPT $\cA = \{\cA_\secpar\}_{\secpar \in \bbN}$, every (non-uniform, quantum) advice ensemble $\ket{\tau} = \{\ket{\tau_n}\}_{n \in \mathbb{N}}$, and large enough $n \in \bbN$, 
\begin{align*}
\label{eq:type1-instantiation}
    \Prr_{x\leftarrow \cD_n}\left[\left|\cA(\ket{\tau},x) - \Pr_{\cD_n}[x]\right| \leq \frac{\Pr_{\cD_n}[x]}{q(n)} \right] \leq 1 - \frac{1}{q(n)} 
\end{align*}
We may also assume w.l.o.g. that the outputs of $\cD_n$ are of $n$ bits. Define $\Gen(1^n)$ as follows:
    \begin{itemize}
        \item Sample $i \leftarrow [0,n-1]$
        \item Sample $x \leftarrow \cD_n$
        \item Output puzzle $x_{1\ldots i}$ and key $x_{i+1}$
    \end{itemize}
Let $p(n)$ be any polynomial such that $p(n) > n^6q(n)^3$.
We will prove that $\Gen(1^n)$ is a $1/p(n)$-distributional one-way puzzle. Since Theorem \ref{thm:owp-amplification} shows that distributional one-way puzzles can be amplified to obtain (strong) one-way puzzles, this suffices to prove the theorem.

Suppose for the sake of contradiction that $\Gen(1^n)$ is \textit{not} a $1/p(n)$-distributional one-way puzzle. Then there exists a QPT $\cA = \{\cA_\secpar\}_{\secpar \in \bbN}$ and (non-uniform, quantum) advice ensemble $\ket{\tau} = \{\ket{\tau_n}\}_{n \in \mathbb{N}}$ such that for infinitely many $n \in \bbN$, 
   \[
 \SD\left(\{x_{1\ldots i}, x_{i+1}\}, \{x_{1\ldots i}, \cA(\ket{\tau}, x_{1\ldots i})\}\right) \leq 1/p(n)
\] where $x_{1\ldots i}, x_{i+1} \leftarrow \Gen(1^n)$.
Fix any such adversary $\cA$, any such advice ensemble $\ket{\tau}$, and any such large enough $n\in \bbN$. We will drop the advice from the notation since it is always implicitly provided to the adversary. 

First we define some useful terms. For $i \in [0,n-1]$, $z\in\bin^i$, $z' \in \bin^{ n - i}$ and $b \in \bin$:
\begin{itemize}
    \item Define $p_z:= \Pr_{x\leftarrow \cD_n}[x_{1\ldots i} = z]$, i.e. the probability that the first $i$ bits of $x$ sampled from $\cD_n$ match $z$.
    \item Define $p_{z'|z}:= \Prr_{x\leftarrow \cD_n}\left[x_{i+1 \ldots n} = z'\ \middle|\  x_{1\ldots i} = z\right]$ i.e. the probability that the last $n-i$ bits of $x$ sampled from $\cD_n$ match $z'$ conditioned on the first $i$ bits of $x$  matching $z$.
    \item Similarly define $p_{b|z}:= \Prr_{x\leftarrow \cD_n}\left[x_{i+1} = b\ \middle|\  x_{1\ldots i} = z\right]$ i.e. the probability that the $i+1$-th bit of $x$ sampled from $\cD_n$ is $b$ conditioned on the first $i$ bits of $x$  matching $z$.
    \item Define $\tp_{b|z} := \Pr[\cA(z) = b]$
\end{itemize}
We will use the adversary to build $\cA'$ that contradicts the security of $\cC$. We will first build an estimator $\cE_b$ that takes input $z$ and estimates $\tp_{b|z}$.

For $i \in [0,n-1]$, $z \in \bin^{i}, b\in \bin$ define the algorithm $\cE_b(z)$ as:
\begin{itemize}
    \item For $j=1$ to $16n^7q(n)^4$:
    \begin{itemize}
        \item $X_j \leftarrow \bbI\{A(z) = b\}$
    \end{itemize}
    \item Return $\sum_j X_j/16n^7q(n)^4$
\end{itemize}
where $\bbI$ is an indicator function. We now define the algorithm $\cA'$ that takes input $x$ and computes an approximation to $p_x$. Define $\cA'(x)$ as
\begin{itemize}
    \item Return $\Pi_{j\in[0,n-1]} \cE_{x_{j+1}}(x_{1\ldots j})$
\end{itemize}
Observe that $p_x = \prod_{j=0}^{n-1} p_{x_{j+1}|x_{1\ldots j}}$. On a high level, to approximate $p_x$ it therefore suffices to approximate $p_{x_{j+1}|x_{1\ldots j}}$ for every $j$ and multiply the approximations. We cannot directly estimate $p_{x_{j+1}|x_{1\ldots j}}$ however $\cA'$ uses $\cE_{x_{j+1}}(x_{1\ldots j})$  to estimate $\tp_{x_{j+1}|x_{1\ldots j}}$ instead. We will use the fact that $\cA$ distributionally inverts $\Gen$ to argue that this is sufficient to (on average) obtain an approximation to $p_x$.

First, we show that $\cE_b(z)$ is a good approximation of $\tp_{b|z}$ with high probability.
\begin{claim}\label{clm:chernoff-prob-est} For all $i \in [0,n-1]$, $z \in \bin^{i}, b\in \bin$
    \[
    \Pr\left[\left|\cE_b(z) - \tp_{b|z}\right|\leq \frac{1}{4n^3q(n)^2}\right] \geq 1- 2e^{-n}
    \]
\end{claim}
\begin{proof}
    Follows from setting $\delta = \sqrt{n}$ in the additive Chernoff bound (Theorem \ref{thm:chernoff-additive}).
\end{proof}
The estimate we obtain has an inverse polynomial additive error. If the value being estimated is too small, then this leads to a large relative error in the estimate. We define the set $\bbB$ as containing all $x$ such that $p_{x_{j+1}|x_{1\ldots j}}$ is atleast $1/n^2q(n)$ for every index $j$.  Formally
\[
\bbB:= \left\{x\in \bin^n \text{ s.t.} \forall j \in [0,n-1], p_{x_{j+1}|x_{1\ldots j}} \geq 1/n^2q(n)\right\}
\]
 Intuitively, $\bbB$ contains strings for which the additive error induced by Claim \ref{clm:chernoff-prob-est} only leads to a small relative error. Next we show that with high probability $x$ sampled from $\cD_n$ is in $\bbB$.
\begin{claim}
    \label{clm:bbB-prob-est}
    \[
    \Prr_{x\leftarrow \cD_n}[x\in\bbB] \geq 1-2/nq(n)
    \]
\end{claim}
\begin{proof}
    For each index $j$ we define the set $\bbS_j$ as the set of strings $x$ such that $p_{x_{j+1}|x_{1\ldots j}}$ is less than $1/n^2q(n)$. Intuitively, if a string $x$ is not in $\bbB$, it must be in $\bbS_j$ for some $j$. Therefore, we can prove the claim by bounding the probability of sampling a string in $\bbS_j$ for every index $j$. Formally, for all $j \in [0,n-1]$
    \[
    \bbS_j:= \left\{x \text{ s.t. } p_{x_{j+1}|x_{1\ldots j}} < 1/n^2q(n)\right\}
    \]
    Now, since $p_x = \prod_{j=0}^{n-1} p_{x_{j+1}|x_{1\ldots j}}$, the probability of sampling $x$ depends on $p_{x_{j+1}|x_{1\ldots j}}$ for every index $j$. If $x\in\bbS_j$ then $p_{x_{j+1}|x_{1\ldots j}}$ is small, which allows us to bound the probability of sampling such an $x$, i.e.
    \begin{align*} 
    \Pr_{x\leftarrow {\cD_n}}\left[ x \in \bbS_j\right] &= \sum_{x\in\bbS_j} \Pr_{\cD_n}[x]\\
    &= \sum_{\substack{x\in\bbS_j}} p_{x_{1\ldots j}} \cdot p_{x_{j+1}|x_{1\ldots j+1}} \cdot p_{x_{j+2 \ldots n}|x_{1\ldots j+1}}\\
    &< \sum_{\substack{x\in\bbS_j}} p_{x_{1\ldots j}} \cdot 1/n^2q(n) \cdot p_{x_{j+2 \ldots n}|x_{1\ldots j+1}}\\
    &\leq \sum_{\substack{x\in\bin^n}} p_{x_{1\ldots j}} \cdot 1/n^2q(n) \cdot p_{x_{j+2 \ldots n}|x_{1\ldots j+1}}\\
    &= 2/n^2q(n)
    \end{align*}
    By a union bound, this implies
    \[
    \Pr_{x\leftarrow \cD_n}\left[ \exists j \text{ s.t. } x \in \bbS_j\right] \leq 2/nq(n)
    \] 
    Since for all $x\notin \bbB$, $\exists j \text{ s.t. } x\in \bbS_j$
    \[
    \Pr_{x\leftarrow {\cD_n}}\left[ x \in \bbB\right] \geq 1- 2/nq(n)
    \] which concludes the proof of the claim.
\end{proof}
Next we define the set of strings $x$ such that for all $j$ a good estimate of $\tp_{x_{j+1}|x_{1\ldots j}}$ is also a good estimate of $p_{x_{j+1}|x_{1\ldots j}}$. Define the set $\bbD$ as follows.
\[
\bbD:= \left\{x\in \bin^n \text{ s.t.} \forall j \in [0,n-1], \left|\tp_{1|x_{1\ldots j}}-p_{1|x_{1\ldots j}}\right|\leq 1/4n^3q(n)^2\right\}
\]
We now use the fact that $\cA$ is a distributional inverter to show  that with high probability $x$ sampled from $\cD_n$ is in $\bbD$.
\begin{claim}
    \label{clm:bbD-prob-est} 
    \[
    \Prr_{x\leftarrow {\cD_n}} [x\in\bbD] \geq 1 - 4n^5q(n)^2/p(n)
    \]
\end{claim}
\begin{proof}
    Recall that
     \[
 \SD\left(\{x_{1\ldots i}, x_{i+1}\}, \{x_{1\ldots i}, \cA(x_{1\ldots i})\}\right) \leq 1/p(n)
\] where $(x_{1\ldots i}, x_{i+1}) \leftarrow \Gen(1^n)$. Since $i$ is chosen uniformly, we may split the statistical distance into terms for each value of $i$, i.e. for $x\leftarrow \cD_n$
\begin{gather*}
 \sum_{j\in[0,n-1]}\Pr_{i\leftarrow[0,n-1]}[i=j]\cdot \SD\left(\{x_{1\ldots j}, x_{j+1}\}, \{x_{1\ldots j}, \cA(x_{1\ldots j})\}\right) \leq 1/p(n)\\
 \implies \forall j\in[0,n-1], \SD\left(\{x_{1\ldots j}, x_{j+1}\}, \{x_{1\ldots j}, \cA(x_{1\ldots j})\}\right) \leq n/p(n)
\end{gather*}
Expanding the statistical distance term, we may rewrite the expression as
\[
\forall j, \bbE_{x\leftarrow {\cD_n}}\left[\left|\tp_{1|x_{1\ldots j}}-p_{1|x_{1\ldots j}}\right|\right] \leq n/p(n)
\] 
By a Markov argument, 
\[
\forall j, \Prr_{x\leftarrow {\cD_n}}\left[\left|\tp_{1|x_{1\ldots j}}-p_{1|x_{1\ldots j}}\right| \geq \frac{1}{4n^3q(n)^2}\right] \leq 4n^4q(n)^2/p(n)
\] 
By a union bound
\[
\Prr_{x\leftarrow {\cD_n}}\left[\forall j, \left|\tp_{1|x_{1\ldots j}}-p_{1|x_{1\ldots j}}\right| \geq \frac{1}{4n^3q(n)^2}\right] \leq 4n^5q(n)^2/p(n)
\] 
which by the definition of $\bbD$ implies
\[
    \Prr_{x\leftarrow {\cD_n}} [x\in\bbD] \geq 1 - 4n^5q(n)^2/p(n)
    \]
    concluding the proof of the claim.
\end{proof}
To complete the proof we will first show that for all $x \in\bbB \cap \bbD$ , $\cA'(x)$ is a good approximation to $p_x$ with high probability. Then we show that with high probability $x \in \bbB \cap \bbD$ when $x$ is sampled from $\cD_n$. To show the former, we will need the following lemma. The lemma shows that for some real values $\{a_i, b_i\}_{i\in [0,n-1]}$, if $b_i$ is a good (i.e. low relative error) estimate of $a_i$ for all $i$ then $\prod_i b_i$ is a good estimate of $\prod_i a_i$.
\begin{lemma}
\label{lem:relative-error-prob-est}
    For all $i\in[0,n-1]$, let $a_i, b_i, \delta$ be values such that 
    \begin{itemize}
        \item $0<a_i\leq 1$
        \item $0\leq\delta < 1/n$
        \item $\frac{\left|a_i -b_i\right|}{a_i} \leq \delta$
    \end{itemize} 
    Then
    \[
    \frac{\left|\Pi_i a_i - \Pi_i b_i\right|}{\Pi_i a_i} \leq 2n\delta
    \]
\end{lemma}
\begin{proof}
    Let $a := \Pi_i a_i$ and $b:= \Pi_i b_i$. Then 
    \[
    b = \Pi_i b_i = \Pi_i a_i\cdot\frac{b_i}{a_i} = a\cdot\Pi_i \left(1 + \frac{a_i - b_i}{a_i}\right)
    \]
    Since $\left|\frac{a_i - b_i}{a_i}\right| \leq \delta \leq 1/n$ and $a>0$
    \begin{gather*} 
    a\left(1-\delta\right)^n \leq b \leq a(1+\delta)^n\\
    \implies a\left(1-n\delta\right) \leq b \leq a/(1-n\delta)
    \end{gather*}
    Therefore,
    \begin{align*}
        \left|1 -\frac{b}{a}\right| &\leq \max\left(n\delta,\frac{n\delta}{1-n\delta}\right)\\
        &\leq 2n\delta
    \end{align*}
    which concludes the proof of the lemma.
\end{proof}
Now we can show that for all $x \in\bbB \cap \bbD$ , with high probability $\cA'(x)$ is a good approximation of $p_x$.
\begin{claim}
\label{clm:final-error-prob-est}
    $\forall x \in \bbB\cap \bbD$
    \[\Prr\left[\left|\cA'(x) - \Pr_{\cD_n}[x]\right|\leq \frac{\Pr_{\cD_n}[x]}{q(n)}\right] \geq 1-2ne^{-n}
    \]
\end{claim}
\begin{proof}
    We start by noting that for all $j$, $\cE_{x_{j+1}}(x_{1\ldots j})$ is close $\tp_{x_{j+1}|x_{1\ldots j}}$ with high probability. Formally, by Claim \ref{clm:chernoff-prob-est}, for all $j\in[0,n-1]$
    \[
    \Pr\left[\left|\cE_{x_{j+1}}(x_{1\ldots j}) - \tp_{x_{j+1}|x_{1\ldots j}}]\right|\leq \frac{1}{4n^3q(n)^2}\right] \geq 1- 2e^{-n}
    \]
    By a union bound
     \[
    \Pr\left[\forall j\in[0,n-1],\left| \cE_{x_{j+1}}(x_{1\ldots j}) - \tp_{x_{j+1}|x_{1\ldots j}}\right|\leq \frac{1}{4n^3q(n)^2}\right] \geq 1- 2ne^{-n}
    \]
    Since $x\in \bbD$, $\tp_{x_{j+1}|x_{1\ldots j}}$ is close to $p_{x_{j+1}|x_{1\ldots j}}$. Formally, $\forall j\in[0,n-1]$
    \[
    \left|  p_{x_{j+1}|x_{1\ldots j}} - \tp_{x_{j+1}|x_{1\ldots j}}\right| \leq \frac{1}{4n^3q(n)^2}
    \]
    By the triangle inequality, this shows that $\cE_{x_{j+1}}(x_{1\ldots j})$ is close $p_{x_{j+1}|x_{1\ldots j}}$ with high probability.
    \[
    \Pr\left[\forall j\in[0,n-1],\left| \cE_{x_{j+1}}(x_{1\ldots j}) - p_{x_{j+1}|x_{1\ldots j}}\right|\leq \frac{1}{2n^3q(n)^2}\right] \geq 1- 2ne^{-n}
    \]
    This gives us a bound on the additive error. To obtain a bound on relative error, we note that for all $x\in \bbB$, $\forall j\in[0,n-1],  p_{x_{j+1}|x_{1\ldots j}} \geq 1/n^2q(n)$. Therefore we can divide by $p_{x_{j+1}|x_{1\ldots j}}$
    \begin{equation*}
    \Pr\left[\forall j\in[0,n-1],\frac{\left| \cE_{x_{j+1}}(x_{1\ldots j}) - p_{x_{j+1}|x_{1\ldots j}}\right|}{p_{x_{j+1}|x_{1\ldots j}}}\leq \frac{1}{2nq(n)}\right] \geq 1- 2ne^{-n}
    \end{equation*}
    Finally we use Lemma \ref{lem:relative-error-prob-est} to show that the product of good estimates for $p_{x_{j+1}|x_{1\ldots j}}$ is a good estimate for $p_x$. For all $j \in [0,n-1]$, let $a_j := p_{x_{j+1}|x_{1\ldots j}}$ and $b_j:=\cE_{x_{j+1}}(x_{1\ldots j})$. Let $\delta:= 1/2nq(n)$. Note that $\Pi_j a_j = \Pi_j p_{x_{j+1}|x_{1\ldots j}} = p_x$ and $\Pi_j b_j = \cA'(x)$. Then applying Lemma \ref{lem:relative-error-prob-est} to the above
\[
    \Pr\left[\frac{\left|\cA'(x) - p_x\right|}{p_x} \leq 1/q(n)\right] \geq 1- 2ne^{-n}    
\]
which after rearranging gives
\[
    \Pr\left[\left|\cA'(x) - p_x\right| \leq p_x/q(n)\right] \geq 1- 2ne^{-n}    
\]
concluding the proof of the claim.
\end{proof}
Finally, combining Claim $\ref{clm:bbB-prob-est}$ and Claim $\ref{clm:bbD-prob-est}$ we can show that 
\[
\Prr_{x\leftarrow {\cD_n}}\left[x\in \bbB \cap \bbD\right] \geq 1-4n^5q(n)^2/p(n) -2/nq(n)
\]
By Claim \ref{clm:final-error-prob-est}
\begin{align*}
\Prr_{x\leftarrow {\cD_n}}\left[\left|\cA'(x) - \Pr_{\cD_n}[x]\right|\leq \frac{\Pr_{\cD_n}[x]}{q(n)}\right] &\geq \Prr_{x\leftarrow {\cD_n}}[x\in\bbB\cap\bbD] \cdot \left(1-2e^{-n}\right)\\
&\geq \left(1-4n^5q(n)^2/p(n) -2/nq(n)\right)\left(1-2e^{-n}\right)
\end{align*}
Since $p(n)>n^6q(n)^3$ and $n$ is large enough,
\begin{align*}
\Prr_{x\leftarrow {\cD_n}}\left[\left|\cA'(x) - \Pr_{\cD_n}[x]\right|\leq \frac{\Pr_{\cD_n}[x]}{p(n)}\right] &\geq \left(1-4n^5q(n)^2/p(n) -2/nq(n)\right)\left(1-2e^{-n}\right)\\
&> \left(1-4/nq(n) -2/nq(n)\right)\left(1-2e^{-n}\right)\\
&>(1-1/q(n))
\end{align*}
which contradicts Definition $\ref{def:type-1}$, concluding the proof of the theorem.
\end{proof}

\subsection{One-Way Puzzles Imply Native Approximation Hardness}
\begin{theorem}
    \label{thm:puzzles-imply-type1}
     The existence of distibutional one-way puzzles (Definition \ref{def:dist-owp}) implies the existence of families of distributions that satisfy Definition $\ref{def:type-1}$.
\end{theorem}
Let $\Gen(1^n)$ be a $1/q(n)$-distributional one-way puzzle for some polynomial $q$ that samples $n$ bit puzzles and $n$ bit keys. Let $\cD = \{\cD_n\}_{n\in\bbN}$ be a family of distributions where $\cD_n$ is defined as follows:
\begin{enumerate}
    \item $(s,k) \sample \Gen(1^n)$
    \item $i \leftarrow [0,n-1]$
    \item $x \leftarrow s \| k_{1\dots i}$ 
    \item $b_\text{mode} \leftarrow \{0, 1\}$
    \item If $b_\text{mode} = 0$, then $\beta \leftarrow \{0, 1\}$. If $b_\text{mode} = 1$, then $\beta \leftarrow k_{i + 1}$
    \item Output $(x, \beta)$
\end{enumerate}

Let $p$  be a polynomial such that $p(n) > 5nq(n)$. We will prove that $\cD$ satisfies Definition $\ref{def:type-1}$, i.e.
for all QPT $\cA = \{\cA_\secpar\}_{\secpar \in \bbN}$, every (non-uniform, quantum) advice ensemble $\ket{\tau} = \{\ket{\tau_n}\}_{n \in \mathbb{N}}$, and large enough $n \in \bbN$, 
\[
    \Prr_{x\leftarrow \cD_n}\left[\left|\cA(\ket{\tau}, x) - \Pr_{\cD_n}[x]\right| \leq \frac{\Pr_{\cD_n}[x]}{p(n)}\right] \leq 1 - \frac{1}{p(n)}
\] Note that this suffices to prove the theorem.

     Suppose for the sake of contradiction that $\cD$ \textit{does not} satisfy Definition $\ref{def:type-1}$. Therefore there exists a QPT $\cA = \{\cA_\secpar\}_{\secpar \in \bbN}$ and (non-uniform, quantum) advice ensemble $\ket{\tau} = \{\ket{\tau_n}\}_{n \in \mathbb{N}}$ such that for infinitely many $n \in \bbN$, 
      \begin{align*}
        & \underset{x \leftarrow \cD_n}{\Prr}\left[ \left|\cA(\ket{\tau}, x) - \Pr_{\cD_n}[x]\right| \leq \frac{\Pr_{\cD_n[x]}}{p(n)} \right] > 1 - 1/{p(n)}.
    \end{align*} 
Fix any such adversary $\cA$, any such advice ensemble $\ket{\tau}$, and any such large enough $n\in \bbN$. We will drop the advice from the notation since it is always implicitly provided to the adversary.  
    We will show that we can use $\cA$ to contradict one-wayness of $\Gen$. 
    
    First we define some useful terms. 
        For all $j \in [0,n-1]$, $s \in \{0, 1\}^n$, $z \in \{0, 1\}^j$, and $b \in \bin$ we define:
\begin{itemize}
        \item Define $p_{s} := {\Prr_{s',k \leftarrow \Gen(1^n)}}[s = s']$ i.e. the probability that $\Gen(1^n)$ samples puzzle $s$.
        \item Define $p_{b|sz}:= \Prr_{s',k \leftarrow \Gen(1^n)}\left[k_{i+1} = b\ \middle|\  s' = s \wedge k_{1\ldots i} = z\right]$ i.e. the probability that the $i+1$-th bit of key $k$ sampled by $\Gen(1^n)$ equals $b$ conditioned on $\Gen(1^n)$ sampling puzzle $s$ and the first $i$ bits of $k$  matching $z$.
\end{itemize}
For all $s \in \{0, 1\}^n$, $j \in [0,n-1]$, $k \in \{0, 1\}^n$, define ${\cS}(s, k_{1\dots j}) :=$
    \begin{enumerate}
        \item $\ta_1 \leftarrow \cA(s\|k_{1\dots j}\|1)$
        \item $\ta_0 \leftarrow \cA(s\|k_{1\dots j}\|0)$
        \item $\pi := \frac{3\ta_1 - \ta_0}{2(\ta_1 + \ta_0)}$ \\
        If $\pi > 1$, set $\pi \leftarrow 1$ \\
        If $\pi < 0$, set $\pi \leftarrow 0$
        \item Sample $k_{j+1} \leftarrow \text{Bern}(\pi)$, return $k_{j+1}$.
    \end{enumerate}
    Intuitively, ${\cS}(s, k_{1\dots j})$ aims to sample from the distribution on $k_{j + 1}$ induced by sampling from $\Gen(1^n)$ conditioned on $s$ and $k_{1\dots j}$, i.e. to sample $k_{j+1}$ with probability $p_{k_{j+1}|sk_{1\ldots j}}$. Define $\tp_{b|sz} := \Pr[\cS(s,z) = b]$ .\\

    \noindent For all $s \in \{0, 1\}^n$, define ${\cA'}(s) :=$
    \begin{enumerate}
        \item For $j = 0$ to $n - 1$: 
        \begin{itemize}
            \item $k_{j + 1} \leftarrow {\cS}(s,k_{1\dots j})$ \hspace{2em}(i.e. if $j = 0$, $k_1 \leftarrow {\cS}(s)$)
        \end{itemize}
        \item Return $k$.
    \end{enumerate}
    Intuitively, $\cA'(s)$ aims to use the sampler $\cS$ to sample bit-by-bit from the distribution induced on keys $k$ sampled by $\Gen(1^n)$ conditioned on sampling puzzle $s$. More formally, we will prove that
    \begin{align*}
    &\SD(\{s,k\}, \{s, {\cA'}(s)\}) \leq \frac{1}{q(n)},
    \end{align*}
    where $(s, k )\leftarrow \Gen(1^n)$.

Note that $\Pr_{\Gen(1^n)}[(s,k)]$ may be expressed in terms of the probability of sampling $s$ and the probability of sampling $k_{j+1}$ conditioned on having sampled $s$ and $k_{1\ldots j} $.
\[
    \Pr_{\Gen(1^n)}[(s,k)] = p_s \cdot \prod_{i=1}^{n-1} p_{k_{i+1}|sk_{1\ldots i}}
\]
The probability distribution $\{s, {\cA'}(s)\}_{s,k\leftarrow\Gen(1^n)}$ may also be expressed similarly
\[
    \Prr_{(s',k') \leftarrow \Gen(1^n)}[s'=s] \cdot  \Pr_{\cA'(s)}[k] = p_s \cdot \prod_{i=1}^{n-1} \tp_{k_{i+1}|sk_{1\ldots i}}
\]
To argue that the two distributions are close, we will define a series of hybrid distributions interpolating between them. 
    For all $j \in [0,n-1]$, define distribution $\tD_j$ on $\bin^n \times \bin^n$ as follows. For all $s \in \bin^n$ and $k \in \bin^n$
    \begin{align*}
        \Pr_{\tD_j}[s, k] :=&\ p_s \cdot p_{k_1|s} \cdot p_{k_2|sk_1} \cdot p_{k_3| sk_{1,2}} \ldots p_{k_j|sk_{1\ldots j-1}} \cdot \tp_{k_{j+1}|sk_{1\ldots j}} \ldots \tp_{k_n|sk_{1\ldots n-1}}\\
        =&\ p_s \cdot \prod_{i = 0}^{j-1} p_{k_{i+1}|sk_{1\dots i}} \prod_{i = j}^{n-1} \tp_{k_{j+1}|sk_{1\dots j}}
    \end{align*}
    Note that $\tD_0 = \{s, {\cA'}(s)\}_{s,k\leftarrow\Gen(1^n)}$ and $\tD_n = \{s, k\}_{s,k\leftarrow\Gen(1^n)}$.
    Therefore, by the triangle inequality for statistical distance
    \begin{align}
        \label{eq:SD-sum}
        \SD(\{s,k\}, \{s, {\cA'}(s)\}) &= \SD(\tD_0, \tD_n)\nonumber\\
        &\leq \sum_{j=0}^{n-1} \SD(\tD_{j+1}, \tD_{j})
    \end{align}
    where $(s,k)\leftarrow\Gen(1^n)$. To upper bound $\SD(\{s,k\}, \{s, {\cA'}(s)\})$, it therefore suffices to upper bound $\SD(\tD_{j+1}, \tD_j)$. For any $j \in [0,n-1]$
    \begin{align}
    \label{eq:SD-in-puzzles-imply-type1}
    \SD(\tD_{j+1}, \tD_{j}) &= \frac{1}{2}\cdot\sum_{s,k} \left|\Pr_{\tD_{j+1}}[s,k]-\Pr_{\tD_{j}}[s,k]\right|\nonumber\\
        &=\frac{1}{2}\cdot\sum_{s,k} p_s\cdot\prod_{i = 0}^{j-1} p_{k_{i+1}|sk_{1\dots i}} \left|p_{k_{j+1}|sk_{1\dots j}} - \tp_{k_{j+1}|sk_{1\dots j}}\right| \prod_{i = j }^{n-1} \tp_{k_{i+1}|sk_{1\dots i}}\nonumber\\
        &=\frac{1}{2}\cdot\sum_{s,k_{1\ldots j+1}}p_s\cdot\prod_{i = 0}^{j-1} p_{k_{i+1}|sk_{1\dots i}} \left|p_{k_{j+1}|sk_{1\dots j}} - \tp_{k_{j+1}|sk_{1\dots j}}\right| \cdot \sum_{k_{j+2\ldots n}}\prod_{i = j }^{n-1} \tp_{k_{i+1}|sk_{1\dots i}}\nonumber\\
       &=\frac{1}{2}\cdot\sum_{s,k_{1\ldots j+1}}p_s\cdot\prod_{i = 0}^{j-1} p_{k_{i+1}|sk_{1\dots i}} \left|p_{k_{j+1}|sk_{1\dots j}} - \tp_{k_{j+1}|sk_{1\dots j}}\right|
       \end{align}
    The value of $\left|p_{k_{j+1}|sk_{1\dots j}} - \tp_{k_{j+1}|sk_{1\dots j}}\right|$ expresses how far the output distribution of $\cS(s, k_{1\dots j})$ is from the distribution of $k_{j+1}$ conditioned on $s, k_{1\dots j}$. In the next subclaim we show that this term is small if for all $b\in\bin$, $\cA(s, k_{1\dots j}, b)$ is close to $\Prr_{\cD_n}[s\|k_{1\dots j}\| b]$ with high probability.
    \begin{claim} For all $s\in\bin^n, k\in \bin^n, j\in[0,n-1]$, define $\tau_{s,k_{1\ldots j}}$ as follows.

    \[
        \tau_{s,k_{1\ldots j}} := \sum_{b\in\bin}\Pr\left[\left|\cA(s,k_{1\ldots j},b) - \Pr_{\cD_n}[s\|k_{1\ldots j}\|b]\right| > \frac{\Pr_{\cD_n}[s\|k_{1\ldots j}\|b]}{p(n)}\right]
    \]
        then
    \[
        \left|p_{k_{j+1}|sk_{1\dots j}} - \tp_{k_{j+1}|sk_{1\dots j}}\right| \leq 7/p(n) + \tau_{s,k_{1\ldots j}}
    \]
    \end{claim}
    \begin{proof}
    First we note that since by definition $p_{1|sk_{1\dots j}} + p_{0|sk_{1\dots j}} = 1$ and $\tp_{1|sk_{1\dots j}} + \tp_{0|sk_{1\dots j}} = 1$
        \[\left|p_{1|sk_{1\dots j}} - \tp_{1|sk_{1\dots j}}\right| = \left|p_{0|sk_{1\dots j}} - \tp_{0|sk_{1\dots j}}\right|\]
For $b\in \bin$, let $a_b:=\Pr_{\cD_n}[s\|k_{1\dots j}\|b]$. By the construction of $\cD_n$, for $b\in\bin$
\[
    a_b = \Pr[i=j]\cdot\Prr_{s',k' \leftarrow \Gen(1^n)}[s'=s \wedge k'_{1\dots j}=k_{1\dots j}]\cdot \left(\frac{1}{4} + \frac{p_{b|sk_{1\dots j}}}{2}\right)
\]
This allows us to express $p_{b|sk_{1\dots j}}$ in terms of $a_b$. More precisely we can say
\[
    p_{1|sk_{1\dots j}} =  \frac{3 a_1 - a_0}{2(a_0 + a_1)}
\]

Recall that $\cS(s, k_{1\ldots j})$ computes $\ta_1$ and $\ta_0$, and then calculates $\pi = \frac{3\ta_1 - \ta_0}{2\ta_1 + 2\ta_0}$. Finally it samples a bit according to Bern$(\pi)$. Therefore we can bound the distance between  $p_{1|sk_{1\dots j}}$ and $\tp_{1|sk_{1\dots j}}$ in the case where for all $b\in\bin$, $\ta_b$ is close to $a_b$.

\begin{subclaim}
    If for some sampled $\ta_0$ and $\ta_1$ the following two conditions hold:
    \begin{itemize}
        \item $|\ta_0 - a_0| \leq \frac{a_0}{p(n)}$
        \item $|\ta_1 - a_1| \leq \frac{a_1}{p(n)}$
    \end{itemize}
    Then $|\pi - p_{1|sk_{1\ldots j}}| \leq \frac{6}{p(n)}$.
\end{subclaim}

\begin{proof}
Let $t := a_1/a_0$ and $\ttt := \ta_1/\ta_0$. Since $a_0\geq 0$ and $a_1 \geq 0$ we can bound $\ttt$ using $t$ as follows. $\ta_0 \geq (1-1/p(n))\cdot a_0$ and $\ta_1 \leq (1+1/p(n))\cdot a_1$ so for large enough $n$
\begin{align*}
    \ttt &\leq \frac{1+1/p(n)}{1-1/p(n)}\cdot t\\
         &\leq (1+3/p(n)) \cdot t
\end{align*}
Similarly,  $\ta_0 \leq (1+1/p(n))\cdot a_0$ and $\ta_1 \geq (1-1/p(n))\cdot a_1$ so for large enough $n$
\begin{align*}
    \ttt &\geq \frac{1-1/p(n)}{1+1/p(n)}\cdot t\\
         &\geq (1-3/p(n)) \cdot t
\end{align*}
Now, $p_{1|sk_{1\dots j}} =  \frac{3 a_1 - a_0}{2(a_0 + a_1)} = \frac{3t - 1}{2(t + 1)}$ and $\pi =  \frac{3 \ta_1 - \ta_0}{2(\ta_0 + \ta_1)} = \frac{3\ttt - 1}{2(\ttt + 1)}$ so
\begin{align*}
    |p_{1|sk_{1\dots j}} - \pi| &= \left|\frac{3t - 1}{2(t + 1)} - \frac{3\ttt - 1}{2(\ttt + 1)}\right| \\
    &= \left|\frac{3}{2} - \frac{4}{2(t + 1)} - \frac{3}{2} + \frac{4}{2(\ttt + 1)}\right|\\
    &= \left|\frac{2}{(\ttt + 1)} - \frac{2}{(t + 1)}\right|\\
    &= \left|\frac{2(t -\ttt)}{(\ttt + 1)(t + 1)}\right|
\end{align*}
We have shown above that $|t -\ttt|\leq 3t/p(n)$
\begin{align*}
    |p_{1|sk_{1\dots j}} - \pi|&\leq\left|\frac{4t/p(n))}{(\ttt + 1)(t + 1)}\right|\\
    &=\frac{6}{p(n)}\cdot\left|\frac{t}{(\ttt + 1)(t + 1)}\right|\\
    &=\frac{6}{p(n)}\cdot\left|\frac{1}{(\ttt + 1)(1 + 1/t)}\right|\\
    &\leq \frac{6}{p(n)}
\end{align*}
\end{proof}
The subclaim shows that when $|\ta_0 - a_0| \leq \frac{a_0}{p(n)}$ and$|\ta_1 - a_1| \leq \frac{a_1}{p(n)}$ then $|\pi - p_{1|sz}| \leq \frac{6}{p(n)}$. Additionally note that $|\pi - p_{1|sz}|$ cannot exceed 1. We can therefore unconditionally bound $|\pi - p_{1|sz}|$ in terms of the probability of sampling such $\ta_0, \ta_1$. Define $\tau'$ as
\[
\tau' := \Pr\left[\left(\left|\ta_0 - a_0\right| > \frac{a_0}{p(n)}\right) \vee \left(\left|\ta_1 - a_1\right| > \frac{a_1}{p(n)}\right)\right]
\]
We can therefore express $ \left|p_{k_{j+1}|sk_{1\dots j}} - \tp_{k_{j+1}|sk_{1\dots j}}\right|$ as
\begin{align*}
     \left|p_{k_{j+1}|sk_{1\dots j}} - \tp_{k_{j+1}|sk_{1\dots j}}\right| &\leq (1-\tau') \cdot 6/p(n) + \tau'\cdot 1\\
            &=6/p(n) +(1-6/p(n))\tau'\\
            &\leq 6/p(n) + \tau'
\end{align*}
The statement of the claim follows from the observation that $\tau_{s,k_{1\ldots j}} \geq \tau'$
\end{proof}
We can use the claim to rewrite \eqref{eq:SD-in-puzzles-imply-type1} as
\begin{align*}
    \SD(\tD_{j+1}, \tD_{j}) &\leq\frac{1}{2}\cdot\sum_{s,k_{1\ldots j+1}}p_s\cdot\prod_{i = 0}^{j-1} p_{k_{i+1}|sk_{1\dots i}} \cdot \left(6/p(n) + \tau_{s,k_{1\ldots j}}\right)\\
    &= 3/p(n) + \frac{1}{2}\cdot\sum_{s,k_{1\ldots j+1}}p_s\cdot\prod_{i = 0}^{j-1} p_{k_{i+1}|sk_{1\dots i}} \cdot\tau_{s,k_{1\ldots j}}
\end{align*}
which can be plugged into \eqref{eq:SD-sum} to get
\begin{align}
    \label{eq:final-SD}
     \SD(\{s,k\}, \{s, {\cA'}(s)\})  &\leq 3n/p(n) + \frac{1}{2}\cdot\sum_{j,s,k_{1\ldots j+1}}p_s\cdot\prod_{i = 0}^{j-1} p_{k_{i+1}|sk_{1\dots i}} \cdot\tau_{s,k_{1\ldots j}}\nonumber\\
     &=3n/p(n) + \frac{1}{2}\sum_{j,s,k_{1\ldots j+1}}\cdot\Prr_{s',k' \leftarrow \Gen(1^n)}[s'=s \wedge k'_{1\dots j}=k_{1\dots j}]\cdot\tau_{s,k_{1\ldots j}}
\end{align}
Recall that by assumption, with high probability over $x\leftarrow \cD_n$, $\left|\cA(x) - \Pr_{\cD_n}[x]\right|$ is bounded, i.e.:
\begin{align*}
        & \underset{x \leftarrow \cD_n}{\Prr}\left[ \left|\cA(x) - \Pr_{\cD_n}[x]\right| > \frac{\Pr_{\cD_n[x]}}{p(n)} \right] < 1/{p(n)}
    \end{align*} 
which can be expressed as the following sum
\begin{align*}    \sum_x \Pr_{\cD_n}[x] \cdot \Prr\left[ \left|\cA(x) - \Pr_{\cD_n}[x]\right| > \frac{\Pr_{\cD_n[x]}}{p(n)} \right]<1/p(n) 
\end{align*}
and by the construction of $\cD_n$
\begin{align*}  
    \Pr_{\cD_n}[s\|k_{1\dots j}\|b] &= \Pr[i=j]\cdot\Prr_{s',k' \leftarrow \Gen(1^n)}[s'=s \wedge k'_{1\dots j}=k_{1\dots j}]\cdot \left(\frac{1}{4} + \frac{p_{b|sk_{1\dots j}}}{2}\right)\\
    &= \frac{1}{n}\cdot \Prr_{s',k' \leftarrow \Gen(1^n)}[s'=s \wedge k'_{1\dots j}=k_{1\dots j}]\cdot \left(\frac{1}{4} + \frac{p_{b|sk_{1\dots j}}}{2}\right)\\
    &\geq \frac{\Prr_{s',k' \leftarrow \Gen(1^n)}[s'=s \wedge k'_{1\dots j}=k_{1\dots j}]}{4n}
\end{align*}
so the above sum can be rewritten as
\begin{align*}    
1/p(n) >& \sum_{j, x =s\|k_{1\ldots j+1}\|b}\frac{\Prr_{s',k' \leftarrow \Gen(1^n)}[s'=s \wedge k'_{1\dots j}=k_{1\dots j}]}{4n}\cdot \Prr\left[ \left|\cA(x) - \Pr_{\cD_n}[x]\right| > \frac{\Pr_{\cD_n[x]}}{p(n)} \right]\\
=& \sum_{j,s,k_{1\ldots j+1}}\frac{\Prr_{s',k' \leftarrow \Gen(1^n)}[s'=s \wedge k'_{1\dots j}=k_{1\dots j}]}{4n}\cdot \tau_{s,k_{1\ldots j}}
\end{align*}
Plugging this back into \eqref{eq:final-SD} and recalling that $p(n) > 5nq(n)$
\begin{align*}
    \SD(\{s,k\}, \{s, {\cA'}(s)\})  &\leq 3n/p(n) + \frac{1}{2}\cdot4n/p(n)\\
    &=5n/p(n)\\
    &< 1/q(n)
\end{align*}
which contradicts the security of $\Gen$. 

\section{The Hardness of Pseudo-Deterministic Sampling implies One-Way Puzzles}
\label{sec:pseudo}
In this section we prove the following theorem.
\begin{theorem}
    If $1/q(n)$-pseudo-deterministic hard distributions (Definition \ref{def:pd-hardness}) exist for some non-zero polynomial $q$, then one-way puzzles exist.
\end{theorem}
\begin{proof}
    Let $D_n$ be a $1/q(n)$-pseudo-deterministic hard distribution on $n$ bits. We define a candidate puzzle $\Gen(1^n)$ as follows:
    \begin{itemize}
        \item Sample $i \leftarrow [0,n-1]$ 
        \item Sample $x \leftarrow D_n$
        \item Output puzzle $x_{1 \ldots i}$ and key $x_{i+1}$
    \end{itemize}
    Let $p(n)$ be a polynomial greater than $2nq(n)$. We prove that $\Gen(1^n)$ is a $1/p(n)$-distributional one-way puzzle. Since Theorem \ref{thm:owp-amplification} shows that distributional one-way puzzles can be amplified to obtain (strong) one-way puzzles, this suffices to prove the theorem.

    Assume for the sake of contradiction that $\Gen$ is not a $1/p(n)$-distributional one-way puzzle. Therefore, there exists a QPT $\cA = \{\cA_\secpar\}_{\secpar \in \bbN}$ and (non-uniform, quantum) advice ensemble $\ket{\tau} = \{\ket{\tau_n}\}_{n \in \mathbb{N}}$ such that for infinitely many $n \in \bbN$, 
    \begin{align}
        \label{eq:pd-adversary}
        \mathsf{SD} \left(
        \{x_{1\ldots i},x_{i+1}\}, \{x_{1\ldots i},\cA(\ket{\tau},x_{1\ldots i})\}\right)
        < \frac{1}{p(n)}
    \end{align}
  where $x_{1\ldots i}, x_{i+1} \leftarrow \Gen(1^n)$.
Fix any such adversary $\cA$, any such advice ensemble $\ket{\tau}$, and any such large enough $n\in \bbN$. We will drop the advice from the notation since it is always implicitly provided to the adversary. 

    We first define an estimator $E$ that takes input a string $z$ and performs the following:
    \begin{enumerate}
        \item For $j$ in $[25n^4q(n)]$:
        \begin{itemize}
            \item[] $X_j \leftarrow \cA(z)$
        \end{itemize}
        \item Return $\sum_j X_j / 25n^4q(n)$
    \end{enumerate}
    Intuitively, $E$ outputs an estimate of the probability that $\cA(z)$ outputs $1$.
    Define the algorithm $\cA'(z;R)$ that takes input $z$ and randomness $R\in[2^n]$ and performs the following:
    \begin{enumerate}
        \item $e \leftarrow E(z)$
        \item If $2^n\cdot e > R$ then return $1$ else return $0$
    \end{enumerate}
    Define the algorithm $\cB(1^n;R_1, R_2, \ldots R_n)$ that takes randomness $\{R_i\}_{i\in[n]}$ where $R_i \in [2^n]$ and performs the following:
\begin{enumerate}
    \item For $j = 0$ to $n - 1$: 
        \begin{itemize}
            \item[] $x_{j + 1} \leftarrow \cA'(x_{1\dots j}; R_{j+1})$ \hspace{2em}(i.e. if $j = 0$, $x_1 \leftarrow \cA'(\epsilon)$)
        \end{itemize}
        \item Return $x$.
\end{enumerate}
We will now show that $\cB$ contradicts the security of $\cD$.
\begin{claim}[Correctness]
    \label{clm:pd-correctness}
    \[
    \SD(\cB(1^n), D_n) < 1/q(n)
    \]
\end{claim}
\begin{proof}
We first show that the the output distribution of $\cA'$ is negligibly close to that of $\cA$.
\begin{subclaim}
    \label{subclm:pd-A-vs-A'}
    For any $z \in \bin^{<n}$, 
    \begin{align*}
        \left|\Pr[\cA'(z) = 1] - \Pr[\cA(z) = 1]\right| \leq 1/2^n
    \end{align*}
\end{subclaim}

\begin{proof}
    $\cA'$ uses $E$ to obtain a probability estimate $e$, and then uses the random coins $R$ to output $1$ if $2^n\cdot e > R$, else output $0$. Therefore, for any estimate $e$, the probability that $\cA'$ outputs $1$ equals the probability that $2^n\cdot e > R$. Since $R$ is uniformly sampled from $[2^n]$, this probability is $ \frac{\lfloor{2^n\cdot e}\rfloor}{2^n}$ which is at most $1/2^n$ far from $e$. For a fixed $z$, comparing the probability of $\cA'(z)$ returning $1$ and the expectation of $E(z)$ we therefore obtain
    \begin{align*}
        \left|\Pr[\cA'(z) = 1]  - \bbE[E(z)]\right|&= \left|\sum_e \Pr[E(z) = e]\cdot \frac{\lfloor{2^n\cdot e}\rfloor}{2^n}-\sum_e \Pr[E(z) = e]\cdot e\right|\\
        &= \left|\sum_e \Pr[E(z) = e]\cdot\left(\frac{\lfloor{2^n\cdot e}\rfloor}{2^n}- e\right)\right|\\
        &\leq \sum_e \Pr[E(z) = e]\cdot\left|\left(\frac{\lfloor{2^n\cdot e}\rfloor}{2^n}- e\right)\right|\\
        &\leq \sum_e \Pr[E(z) = e]\cdot\frac{1}{2^n}\\   
        &\leq \frac{1}{2^n}
    \end{align*}
    Additionally, since $E(z)$ simply returns the average of random variables $X_i$, each of which has expected value $\Prr_\cA[\cA(z)=1]$, the expected value of $E(z)$ is also $\Prr_\cA[\cA(z)=1]$. Therefore 
    \[
         \left|\Pr[\cA'(z) = 1]  - \Pr[\cA(z) = 1]\right| \leq 1/2^n
    \]

\end{proof}

A straightforward consequence of SubClaim \ref{subclm:pd-A-vs-A'} is that $\cA'$ is a valid adversary for $\Gen$. Formally,
\begin{subclaim}
    \label{subclm:pd-A'-is-valid-adv}
    \[
        \mathsf{SD} \left(
        \{x_{1\ldots i},x_{i+1}\}, \{x_{1\ldots i},\cA'(x_{1\ldots i})\}\right)
        < \frac{1}{p(n)} + 1/2^n.
    \]
    where $i \leftarrow [0,n-1]$ and $x \leftarrow D_n$.
\end{subclaim}
\begin{proof}
    The proof follows directly from SubClaim \ref{subclm:pd-A-vs-A'} and inequality \eqref{eq:pd-adversary}.
\end{proof}
We now define some helpful terms. For $i \in [0,n-1]$, $z \in \bin^{i}$
\begin{itemize}
    \item Define $p_z:= \Pr_{x\leftarrow D_n}[x_{1\ldots i} = z]$, i.e. the probability that the first $i$ bits of $x$ sampled from $D_n$ match $z$.
    \item Define $p_{b|z}:=\Pr_{x\leftarrow D_n}[x_{i+1} = b |x_{1\ldots i} = z]$ i.e. the probability that the $i+1$-th bit of $x$ sampled from $D_n$ is $b$ conditioned on the first $i$ bits of $x$  matching $z$.
    \item Define $\tp_{b|z}:= \Pr[\cA'(z) = b]$. 
\end{itemize}
    Also, for all $j \in [0,n]$, define distribution $\tD_j$ on $\bin^n$ as follows. For all $x \in \bin^n$:
    \begin{align*}
        \Pr_{\tD_j}[x] &:=  p_{x_1} \cdot p_{x_2|x_1} \cdot p_{x_3| x_{1,2}} \ldots p_{x_j|x_{1\ldots j-1}} \cdot \tp_{x_{j+1}|x_{1\ldots j}} \ldots \tp_{x_n|x_{1\ldots n-1}}\\
        &=\prod_{i = 0}^{j-1} p_{x_{i+1}|x_{1\dots i}} \prod_{i = j }^{n-1} \tp_{x_{i+1}|x_{1\dots i}}
    \end{align*}
    First, note that $\tD_n = D_n$.
    \begin{align*}
        \Pr_{\tD_n}[x] &=\prod_{i = 0 }^{n-1} p_{x_{i+1}|x_{1\dots i}}\\
        &=p_x\\
        &= \Pr_{D_n}[x]
    \end{align*}
    Also note that $\tD_0 = \cB(1^n)$.
    \begin{align*}
        \Pr_{\tD_0}[x] &=\prod_{i = 0 }^{n-1} \tp_{x_{i+1}|x_{1\dots i}}\\
        &= \Pr[B(1^n)=x]
    \end{align*}
    By the triangle inequality for statistical distance
    \begin{align*}
        \SD(\cB(1^n), D_n) &= \SD(\tD_0, \tD_n)\\
        &\leq \sum_{j=0}^{n-1} \SD(\tD_{j+1}, \tD_{j})
    \end{align*}
    The statistical distance between $\tD_{j+1}$ and $\tD_{j}$ can be expressed as 
    \begin{align*}
        \SD(\tD_{j+1}, \tD_{j}) &= \frac{1}{2}\cdot\sum_x \left|\Pr_{\tD_{j+1}}[x]-\Pr_{\tD_{j}}[x]\right|\\
        &=\frac{1}{2}\cdot\sum_x\prod_{i = 0}^{j-1} p_{x_{i+1}|x_{1\dots i}} \left|p_{x_{j+1}|x_{1\dots j}} - \tp_{x_{j+1}|x_{1\dots j}}\right| \prod_{i = j }^{n-1} \tp_{x_{i+1}|x_{1\dots i}}\\
        &=\frac{1}{2}\cdot\sum_{x_{1\ldots j+1}}\prod_{i = 0}^{j-1} p_{x_{i+1}|x_{1\dots i}} \left|p_{x_{j+1}|x_{1\dots j}} - \tp_{x_{j+1}|x_{1\dots j}}\right| \cdot \sum_{x_{j+2\ldots n}}\prod_{i = j }^{n-1} \tp_{x_{i+1}|x_{1\dots i}}\\
       &=\frac{1}{2}\cdot\sum_{x_{1\ldots j+1}}\prod_{i = 0}^{j-1} p_{x_{i+1}|x_{1\dots i}} \left|p_{x_{j+1}|x_{1\dots j}} - \tp_{x_{j+1}|x_{1\dots j}}\right|
    \end{align*}
    For any $x$, expanding $\prod_{i = 0}^{j-1} p_{x_{i+1}|x_{1\dots i}}$ shows that the expression equals $ p_{x_{1 \ldots j}}$, i.e. the probability of obtaining first $j$ bits $x_{1 \ldots j}$ when sampling from $D_n$. Therefore
     \begin{align*}
        \SD(\tD_{j+1}, \tD_{j}) &=\frac{1}{2}\cdot \sum_{x_{1\ldots j+1}} p_{x_{1\dots j}} \left|p_{x_{j+1}|x_{1\dots j}} - \tp_{x_{j+1}|x_{1\dots j}}\right|
    \end{align*}
    Now the right hand side of the equation equals $\SD(\{x_{1\ldots j}, x_{j+1}\}, \{x_{1\ldots j}, \cA'(x_{1\ldots j})\})$ when $x \leftarrow D_n$, which gives the following upper bound for $ \SD(\cB(1^n), D_n)$.
    \begin{align*}
        \SD(\cB(1^n), D_n) &\leq \sum_j \SD(\{x_{1\ldots j}, x_{j+1}\}, \{x_{1\ldots j}, \cA'(x_{1\ldots j})\})
    \end{align*}
    where $x \leftarrow D_n$. If we also consider $j \leftarrow [0,n-1]$ then
     \begin{align*}
        \SD(\cB(1^n), D_n) &\leq n \cdot \SD(\{x_{1\ldots j}, x_{j+1}\}, \{x_{1\ldots j}, \cA'(x_{1\ldots j})\})\\
        &< n \cdot (1/p(n) + 1/2^n)\\
        &< 2n/p(n)
    \end{align*}
    where the second step follows from SubClaim \ref{subclm:pd-A'-is-valid-adv} and the last step holds for large enough $n$. Since $p(n) \geq 2nq(n)$, this implies
    \[
    \SD(\cB(1^n), D_n) < 1/q(n)
    \]
    which concludes the proof of the claim.
\end{proof}
\begin{claim}[Pseudo-determinism]
\label{clm:pd-pd}
    \[
        \Pr_{R_1, \ldots, R_n}[\exists y \text{ s.t. }\Pr\left[\cB(1^n;R_1, \ldots, R_n) \neq y \right] \leq 1/2^n] > 1-1/q(n)
    \]\end{claim}
\begin{proof}
    $\cB(1^n;R_1, \ldots, R_n)$ consists of loop where in the $i$-th iteration $\cA'$ is run with randomness $R_i$. The input of $\cA'$ in the $i$-th iteration (apart from the random coins $R_i$) is completely determined by the output of previous iterations. It therefore suffices to show that for each iteration $i$, for most strings $R_i$, the output of $\cA'$ is pseudo-deterministic , i.e. for the $i$-th iteration there exists an output $y$ such that  $\cA'$ with random coins $R_i$ outputs $y$ with high probability. 
    \begin{subclaim}
        For any $i\in [0,n-1]$, $z\in \bin^i$, $R \in[2^n]$, we say that $\cA'(z;R)$ has determinism error atmost $\epsilon$ if there exists $y$ such that 
        \[
            \Pr[\cA'(z;R)=y] \geq 1-\epsilon
        \]
        Then for all $i\in [0,n-1]$, for all $z\in \bin^i$,
        \[
        \Pr_{R \leftarrow[2^n]}[\cA'(z;R) \text{ has determinism error atmost }(1/n2^n)] > 1-1/nq(n)
        \]
    \end{subclaim}
    \begin{proof}
        $\cA'(z;R)$ uses $E(z)$ to obtain a probability estimate $e$, and then uses the random coins $R$ to output $1$ if $2^n\cdot e > R$, else output $0$. Let $\pi$ be the probability that $A(z) = 1$. $E(z)$ is the average of $p(n)$ independent random variables that are $1$ with probability $\pi$ and are $0$ otherwise. Therefore by setting $\delta=n$ in the additive Chernoff bound (Theorem \ref{thm:chernoff-additive})
        \[
            \Prr_{E}\left[\left|\pi - E(z)\right| \geq 1/5nq(n)\right] \leq 2e^{-n^2}
        \]
        Suppose $|2^n\cdot \pi - R| \geq 2/5nq(n)$. Then with probability atleast $2e^{-n^2}$, $2^n\cdot E(z)$ and $2^n\cdot\pi$ are on the same side of $R$, i.e. if $ R < 2^n\cdot\pi$ then $R < 2^n\cdot E(z)$ with probability atleast $2e^{-n^2}$, while if $R > 2^n\cdot\pi$ then $R > 2^n\cdot E(z)$ with probability atleast $2e^{-n^2}$. Since the output of $\cA'(z;R)$ is entirely determined by whether or not $R < 2^n\cdot E(z)$, $\cA'(z;R)$ therefore has determinism error atmost $2e^{-n^2}$ which is less that $ (1/n2^n)$ for large enough $n$

        All that remains to be shown is that $|2^n\cdot \pi - R| \geq 2/5nq(n)$ holds with high enough probability. Since $R$ is uniformly sampled from $[2^n]$, the number of values of $R$ such that $|2^n\cdot \pi - R| < 2/5nq(n)$ is atmost $1 + 2^n \cdot 4/5nq(n)$. The probability that $R$ is not one of these values is therefore atleast $1 - 4/5nq(n) - 1/2^n$ which is greater than $1-1/nq(n)$ for large enough $n$.
    \end{proof}
    For each iteration $i$, with probability atleast $1-1/nq(n)$,  $\cA'(z_i;R_i)$ has determinism error atmost $1/(n2^n)$ (where $z_i$ is the input in the $i$-th iteration). Then 
    \[
        \Pr_{R_1, \ldots, R_n}[\forall i, \cA'(z_i;R_i) \text{ has determinism error at most }1/(n2^n)] > 1-n/nq(n) = 1/q(n)
    \]
    The determinism error of $\cB$ is atmost the sum of the determinism error of its iterations, therefore
    \[
        \Pr_{R_1, \ldots, R_n}[\cB(1^n;R_1, \ldots, R_n) \text{ has determinism error at most }1/2^n] > 1-1/q(n)
    \]
    which concludes the proof of the claim.
\end{proof}
Claim \ref{clm:pd-correctness} and Claim \ref{clm:pd-pd} show that $\cB$ contradicts the security of $\cD$ which concludes the proof of the theorem.
\end{proof}

\section{State Puzzles are Equivalent to One-Way Puzzles}
\label{sec:owpsp}
We define state puzzles, which capture the hardness of synthesizing a (secret) quantum state $\ket{\psi_s}$ corresponding to a (public) classical string $s$, and are implied by quantum money.

\begin{definition}[State Puzzles]
\label{def:state-puzzle}
    A state puzzle is defined by a quantum polynomial-time generator $\cG(1^n)$ that outputs a classical-quantum state $(s, \ket{\psi_s})$ such that given $s$, it is (quantum) computationally infeasible to output $\rho$ that overlaps noticeably with $\ket{\psi_s}$. 
    
    Formally, 
    for every quantum polynomial-time adversary $\cA$, every (non-uniform, quantum) advice ensemble $\ket{\tau} = \{\ket{\tau_n}\}_{n \in \bbN}$, for large enough $n \in \mathbb{N}$,
    \[
    \underset{\substack{(s,\ket{\psi_s}) \leftarrow \cG(1^n)\\\rho\leftarrow \cA(\ket{\tau},s)}}{\bbE} \Big[\tr(\ket{\psi_s}\bra{\psi_s} \rho)\Big] \leq \negl(n) 
    \]
\end{definition}

We also define a weaker version of state puzzles, where we require that the state output by $\cA$ must fail to project onto $\ketbra{\psi_s}$ with noticeable probability.

\begin{definition}[$\varepsilon$-Weak State Puzzles] \label{def:weak-state-puzzle}
    For $\varepsilon: \bbN \rightarrow \bbR$, a $\varepsilon$-weak state puzzle is defined by a quantum polynomial-time generator $\cG(1^n)$ that outputs a classical-quantum state $(s, \ket{\psi_s})$ such that given $s$, it is (quantum) computationally infeasible to output $\rho$ that almost completely overlaps with $\ket{\psi_s}$
    
    Formally, 
    for
    every quantum polynomial-time adversary $\cA$, every (non-uniform, quantum) advice ensemble $\ket{\tau} = \{\ket{\tau_n}\}_{n \in \bbN}$, for large enough $n \in \mathbb{N}$,
    \[
    \underset{\substack{(s,\ket{\psi_s}) \leftarrow \cG(1^n)\\\rho\leftarrow \cA(\ket{\tau},s)}}{\bbE} \Big[\tr(\ket{\psi_s}\bra{\psi_s} \rho)\Big] \leq 1 - \varepsilon(n)
    \]
\end{definition}
We will sometimes simply refer to weak state puzzles. This is taken to mean $1/p(n)$-weak state puzzles for some non-zero polynomial $p$. 

In this section we prove the equivalence of (weak and standard) state puzzles and (distributional and standard) one-way puzzles. We first prove that the existence of weak state puzzles implies the existence of one-way puzzles.
\begin{theorem}
\label{thm:state-puzzles-imply-owp}
    If $1/q(n)$-weak state puzzles (Definition \ref{def:weak-state-puzzle}) exist for some non-zero polynomial $q(\cdot)$ then $1/p(n)$-distributional one-way puzzles (Definition \ref{def:dist-owp}) exist for some non-zero polynomial $p(\cdot)$.
\end{theorem}
Since Theorem \ref{thm:owp-amplification} shows that distributional one-way puzzles can be amplified to obtain (strong) one-way puzzles, the following is a corollary of Theorem \ref{thm:state-puzzles-imply-owp}.
\begin{corollary}
    If $1/q(n)$-weak state puzzles (Definition \ref{def:weak-state-puzzle}) exist for some non-zero polynomial $q(\cdot)$ then one-way puzzles (Definition \ref{def:owp}) exist.
\end{corollary}
\begin{proof}(of Theorem \ref{thm:state-puzzles-imply-owp})
For some polynomial $q(\cdot)$, let $\cG$ be a $1/q(n)$-weak state puzzle (Definition \ref{def:weak-state-puzzle}) that outputs $s, \ket{\$_s}$, where $s\in \bin^n$ and $\ket{\$_s} \in \{\bbC^2\}^{\otimes n}$. Define the algorithm $\Gen(1^n)$ as follows:
\begin{enumerate}
    \item Sample $\ct \leftarrow \cC$ where $\cC$ is the Clifford group for $n$ qubits.
    \item Sample $\s, \ket{\$_{\s}} \leftarrow \cG(1^n)$
    \item Compute $\ket{\$_{\s,\ct}} := \ct\ket{\$_{\s}}$
    \item Sample $b_\text{mode} \leftarrow \bin$
    \item If $b_\text{mode} = 0:$
    \begin{enumerate}
        \item Sample $i \leftarrow [0,n-1]$
        \item Measure the first $i$ bits of $\ket{\$_{\s,\ct}}$ in the computational basis to obtain measurement output $x$ and residual state $\ket{\$_{x}}$. If $i = 0$ then $x$ is the empty string and $\ket{\$_{x}}:= \ket{\$_{\s,\ct}}$
        \item Measure the $i+1$-th bit of $\ket{\$_{\s,\ct}}$ (i.e. the first bit of $\ket{\$_{x}}$) to obtain measurement output $\beta$.
        \item Let $\pi = (\s, \ct, b_\text{mode}, i, x)$. Output puzzle $\pi$ and key $\beta$.
    \end{enumerate}
    \item If $b_\text{mode} = 1:$
    \begin{enumerate}
        \item Sample $r \leftarrow \bin^n \setminus \{0^n\}$
        \item For $z\in\bin^n$, define $f_r(z) := \mathsf{min}(z, z\oplus r)$. 
        \item Apply $\sum_z \ketbra{z}\otimes X^{f_r(z)}$ to $\ket{\$_{\s,\ct}}\ket{0}$
        \item Measure the second register in computational basis to obtain measurement outcome $x_0$, and set $x_1 = x_0 \oplus r$. Let the residual state on the first register be $\ket{\psi_\text{post}}$.
        \item Sample $b_\text{order} \leftarrow \bin$.
        \item $y_0 := x_{b_\text{order}}$ and $y_1:= x_{1-b_\text{order}}$
        \item Define $V_{y_0, y_1, b}$
        \footnote{We can implement $V_{y_0, y_1, 0}$ as follows. Let $U_0$ be a unitary that maps $\ket{0}$ to $\ket{\varphi_0}:= \frac{\ket{y_0}+\ket{y_1}}{\sqrt{2}}$ and let $U_1$ be a unitary that maps $\ket{0}$ to $\ket{\varphi_1}:=\frac{\ket{y_0}-\ket{y_1}}{\sqrt{2}}$. Given a state $\ket{y_b}\ket{0}$, apply $U_b$ to the second register to get $\ket{y_b}\ket{\varphi_b}$. Then apply $U_0^\dag$ to the second register. If $b=0$ this results in $\ket{y_0}\ket{0}$, else this results in $\ket{y_1}\ket{\varphi}$ where $\ket{\varphi} = U_0^\dag \ket{\varphi_1}$ is some state orthogonal to $\ket{0}$. Then coherently perform the operation that applies $X^{y_0}$ to the first register if the second register is $\ket{0}$ and applies $X^{y_1}$ to the first register if the second register is any other computational basis state. This results in $\ket{0}\ket{0}$ if $b=0$ and $\ket{0}\ket{\varphi}$ otherwise. Finally, apply $U_0$ to the second register to obtain $\ket{0}\ket{\varphi_b}$ and output the second register. $V_{y_0, y_1, 1}$ can be implemented similarly.}
        as a unitary that maps
        \begin{itemize}
            \item $\ket{y_0} \mapsto \frac{\ket{y_0} + i^b \ket{y_1}}{\sqrt{2}}$
            \item $\ket{y_1} \mapsto \frac{\ket{y_0} - i^b \ket{y_1}}{\sqrt{2}}$
        \end{itemize}
        \item Sample $b_\text{rot} \leftarrow \bin$.
        \item Apply $V_{y_0, y_1, b_\text{rot}}$ to $\ket{\psi_\text{post}}$ and measure in computational basis to obtain outcome $y$.
        \item If $y = y_0$ then $\beta\leftarrow 0$, else $\beta \leftarrow 1$.
        \item Let $\pi = (\s, \ct, b_\text{mode}, y_0, y_1, b_\text{rot})$. Output puzzle $\pi$ and key $\beta$.
    \end{enumerate}
    
\end{enumerate}
    Let $k$ be a constant greater than $6$ such that for large enough $n$, $n^k \geq q(n)^3$ and let $p(\cdot)$ be a polynomial such that $p(n) \geq n^{64k}$.
    We will prove that \Gen is a $1/p(n)$-distributional one-way puzzle. Note that this suffices to prove Theorem \ref{thm:state-puzzles-imply-owp}.

    Assume for the sake of contradiction that \Gen is \textit{not} a $1/p(n)$-distributional one-way puzzle. By the definition of $1/p(n)$-distributional one-way puzzle, there exists a QPT $\cA = \{A_\secpar\}_{\secpar \in \bbN}$ and (non-uniform, quantum) advice ensemble $\ket{\tau} = \{\ket{\tau_n}\}_{n \in \mathbb{N}}$ such that for infinitely many $n \in \bbN$, 
     \[
    \{\pi, \beta\}_{(\pi, \beta )\leftarrow \Gen(1^n)} \approx_{1/p(n)} \{\pi, \cA(\adv{},\pi)\}_{(\pi, \beta )\leftarrow \Gen(1^n)}
    \]
Fix any such adversary $\cA$, any such advice ensemble $\ket{\tau}$, and any such large enough $n\in \bbN$. 
 We will use this adversary to build a reduction that contradicts the security of $\cG$. For all $s \in \bin^n$ and  $c \in \cC_n$, define $\Delta_{s,c}$ as follows:
    \begin{itemize}
        \item Let $D_0$ be the distribution of $(\pi, \beta)$ when $\pi, \beta$ is sampled by $\Gen$ conditioned on $\s = s$ and $\ct = c$.
        \item Let $D_1$ be the distribution of $(\pi, A(\pi, \adv{}))$ when $\pi, \beta$ is sampled by $\Gen$ conditioned on $\s = s$ and $\ct = c$.
        \item  $\Delta_{s,c} := SD(D_0, D_1)$
    \end{itemize}
Intuitively, $\Delta_{s,c}$ is the adversary's error in sampling from the true distribution when $\s = s$ and $\ct = c$.
We will use the adversary to synthesize an approximation of $\ket{\$_{s,c}}$ for a random choice of $s,c$,  following the pattern of Aaronson's synthesis algorithm. The algorithm queries a PP oracle to obtain the values of probabilities and phases. The reduction cannot query a PP oracle, so we will replace the query responses with estimates obtained by querying the adversary. 

First, we perform the real amplitude step of Aaronson synthesis.

\begin{claim}[Real Amplitude Synthesis] \label{clm:amp-synth-local}
Let $s \in \bin^n$ and $c \in \cC_n$. Let $\adv{amp}:=\ket{0}\adv{}^{\otimes np(n)}$. Then there exists an efficient unitary $\tM_{s,c}$ such that
	\[
	\left| \tM_{s,c}\ket{0}\adv{amp} -  \ket{\$^*_{s,c}}\adv{amp}\right|\leq \sqrt{3n^3}/p(n)^{1/4} + 2\sqrt{n\Delta_{s,c}}
	\]
    Additionally, there is a uniform circuit family that takes $(s,c)$ as input and implements $\tM_{s,c}$.
\end{claim}
	\begin{proof} 
 Fix any $s$ and $c$. We first define some terms that will be useful for the proof.
\begin{itemize}
    \item Interpret $\ket{\$_{s,c}}$ as $\sum_{z \in \bin^n} a_z e^{-i\phi_z}\ket{z}$ where $a_z \geq 0$ and $\phi_z \in [0, 2\pi)$. 
    \item $\ket{\$^*_{s,c}} := \sum_{z \in \bin^n} a_z \ket{z}$. Intuitively, $\ket{\$_{s,c}^*}$ represents $\ket{\$_{s,c}}$ with the phase information removed, i.e., with real amplitudes.
    \item For all $i \in [1,n]$, for all $z \in \bin^i$, 
    \begin{itemize}
		\item $p_z := \Prr[x = z | b_{\text{mode}} =0\wedge i=i']$ represents the probability that the first $i$ bits equal $z$ when measuring $\ket{\$_{s,c}}$ in the computational basis.
    \end{itemize} 
    \item For all $i \in [0,n-1]$, for all $z \in \bin^i$, for all $b\in \bin$, 
    \begin{itemize}
    	\item $p_{b|z} := \Prr[\beta = b| x = z \wedge b_{\text{mode}} =0\wedge i=i']$
     represents the probability that the $i+1$th bit equals $b$ conditioned on the first $i$ bits equalling $z$ when measuring $\ket{\$_{s,c}}$ in the computational basis.
   	    	\item $\tp_{b|z} := \Prr[b = \cA(s,c,0,i,z, \adv{})]$
    	\end{itemize} 
\end{itemize}
 We first use $A$ to obtain a good estimate for $p_{1|z}$ by defining an estimator $\cE_i(z)$.
		For all $i \in [0,n-1]$, for all $z \in \bin^i$ define $\cE_i(z)$ that takes advice $\adv{}^{\otimes p(n)}$ as follows:
	\begin{itemize}
		\item For $j = 1$ to $p(n)$:
		\begin{itemize}
			\item $B_j \leftarrow A(s,c,0,i,z, \adv{})$
	\end{itemize}
	\item Return $\sum_j B_j/p(n)$
	\end{itemize}
    The rest of the construction of $\tM_{s,c}$ proceeds as in the real amplitude step of Aaronson synthesis, except with the oracle queries replaced with estimates given by $\cE_i(\cdot)$.

	Superposition queries to $\cE_i(\cdot)$ may produce entangled junk so we must later on uncompute to remove the junk. To do so we define $E_i$ to be the purification of $\cE_i$ that acts on input register $\sZ$, output register $\sV$, and advice register $\sV'$, i.e.,
 \[
 E_i\ket{z}_\sZ\ket{0}_\sV\adv{1}_{\sV'} = \ket{z}_\sZ \sum_v\sqrt{\Prr[v = \cE_i(z)]}\ket{v}_\sV\junk{v}_{\sV'}
 \]
 where $\adv{1} := \adv{}^{\otimes p(n)}\ket{0}$ and $\junk{v}$ is some normalized state.
 
 For all $v \geq 0$, define $\ket{\psi_v}:= \sqrt{v}\ket{1} + \sqrt{1-v}\ket{0}$ and let $P$ be a unitary that maps $\ket{v}_\sV\ket{0}_{\sZ'}$ to  $\ket{v}_\sV\ket{\psi_v}_{\sZ'}$\footnote{We may only be able to efficiently implement $P$ upto some exponentially small error, however, this small error will not affect our result so we will elide it for the sake of clarity.}. Let $\sX = \{\sX_i\}_{i\in[1,n]}, \sA =\{\sA_i\}_{i\in[1,n]}, \sA'=\{\sA'_i\}_{i\in[1,n]}$ be collections of registers where each $\sX_i$ is a single qubit. Define $\tM_i$ that acts on $\sX_{1\ldots i+1}, \sA_i,$ and $\sA'_i$ as follows:
\begin{itemize}
    \item Let $\sZ := \sX_{1\ldots i}$, $\sZ':=\sX_{i+1}, \sV := \sA_i, \sV' := \sA'_i$
    \item Apply $\left(E_i\right)^\dagger P E_i$ to $\sZ\sZ'\sV\sV$
\end{itemize}
Define $\tM_{s,c} := \tM_{n-1}\tM_{n-2}\ldots \tM_0$. It is easy to see that $\tM_{s,c}$ can be implemented efficiently given $(s,c)$. The rest of the proof of the claim is therefore dedicated to proving correctness of the construction.

We now show that the above construction satisfies Claim \ref{clm:amp-synth-local} by a series of subclaims. First we note that with high probability, $\cE_i(z)$ is close to $\tp_{1|z}$.
\begin{subclaim}
\label{subclm:good-est} For all $i \in [0,n-1]$, for all $z \in \bin^i$
\[ 
\Prr\left[\left|\cE_i(z) - \tp_{1|z}\right| \geq \frac{n}{\sqrt{p(n)}}\right] \leq 2e^{-2n^2}
\]
\end{subclaim}
\begin{proof}
    Follows from the definition of $\tp_{1|z}$ and setting $\delta$ to be $n$ in the additive Chernoff bound (Theorem \ref{thm:chernoff-additive}).
\end{proof}
Next we note that since $E_i, P,$ and $(E_i)^\dagger$ do not affect the computational basis state on the $\sZ$ register, for all $s,c$, for all $i \in [0,n-1]$, for all $z \in \bin^i$, there exists a state $\ket{\sigma_{z}}$ such that 
\[
    \left(E_i\right)^\dagger P E_i \ket{z}_\sZ\ket{0}_{\sZ'\sV}\adv{1}_{\sV'} = \ket{z}_\sZ\ket{\sigma_{z}}_{\sZ'\sV\sV'}
\]
The next subclaim shows that we can synthesize a state close to $\ket{\psi_{\tp_{1|z}}}$.
\begin{subclaim} \label{subclm:amp-step-error}For all $i \in [0,n-1]$, for all $z \in \bin^i$, if $n$ is large enough
    \[
    \left| \ket{\sigma_{z}}_{\sZ'\sV\sV'} - \ket{\psi_{\tp_{1|z}}}_{\sZ'}\ket{0}_{\sV}\adv{1}_{\sV'}\right| \leq \sqrt{3n}/p(n)^{1/4}
    \]
\end{subclaim}
\begin{proof}
First we apply $\left(E_i\right)^\dagger P E_i$ to both terms and note this does not change the absolute value of their difference.
    \begin{align}\label{eq:amp-step-error}
        \left| \ket{\sigma_{z}}_{\sZ'\sV\sV'} - \ket{\psi_{\tp_{1|z}}}_{\sZ'}\ket{0}_{\sV}\adv{1}_{\sV'}\right|^2 &= \left| \ket{z}_\sZ\ket{\sigma_{z}}_{\sZ'\sV\sV'} - \ket{z}_\sZ\ket{\psi_{\tp_{1|z}}}_{\sZ'}\ket{0}_{\sV}\adv{1}_{\sV'}\right|^2\nonumber\\
        &=\left| \left(E_i\right)^\dagger P E_i \ket{z}_\sZ\ket{0}_{\sZ'\sV}\adv{1}_{\sV'} - \ket{z}_\sZ\ket{\psi_{\tp_{1|z}}}_{\sZ'}\ket{0}_{\sV}\adv{1}_{\sV'}\right|^2\nonumber\\
         &=\left| P E_i \ket{z}_\sZ\ket{0}_{\sZ'\sV}\adv{1}_{\sV'} - E_i\ket{z}_\sZ\ket{\psi_{\tp_{1|z}}}_{\sZ'}\ket{0}_{\sV}\adv{1}_{\sV'}\right|^2
    \end{align}
    Next, we use the definition of $E_i$ and $P$ to expand each term.
    Expanding $E_i\ket{z}_\sZ\ket{\psi_{\tp_{1|z}}}_{\sZ'}\ket{0}_{\sV}\adv{1}_{\sV'}$
\begin{align}
\label{eq:amp-step-error-1}
    E_i\ket{z}_\sZ\ket{\psi_{\tp_{1|z}}}_{\sZ'}\ket{0}_{\sV}\adv{1}_{\sV'} &= \ket{z}_\sZ \ket{\psi_{\tp_{1|z}}}_{\sZ'}\sum_v\sqrt{\Prr[v = \cE_i(z)]}\ket{v}_\sV\junk{v}_{\sV'}
\end{align}
    Expanding $P E_i \ket{z}_\sZ\ket{0}_{\sZ'\sV}\adv{1}_{\sV'}$
    \begin{align}
    \label{eq:amp-step-error-2}
    P E_i\ket{z}_\sZ\ket{0}_{\sZ'}\ket{0}_{\sV}\adv{1}_{\sV'} &= P \ket{z}_\sZ \sum_v \sqrt{\Prr[v = \cE_i(z)]} \ket{0}_{\sZ'} \ket{v}_\sV\junk{v}_{\sV'}\nonumber\\
    &= \ket{z}_\sZ \sum_v \sqrt{\Prr[v = \cE_i(z)]}\ket{\psi_v}_{\sZ'} \ket{v}_\sV\junk{v}_{\sV'}
\end{align}
Plugging \eqref{eq:amp-step-error-1} and \eqref{eq:amp-step-error-2} into \eqref{eq:amp-step-error} gives
\begin{align}\label{eq:amp-step-error-expand}
    &\left| \ket{\sigma_{z}}_{\sZ'\sV\sV'} - \ket{\psi_{\tp_{1|z}}}_{\sZ'}\ket{0}_{\sV}\adv{1}_{\sV'}\right|^2 \nonumber\\
    =& \left|\ket{z}_\sZ \sum_v \sqrt{\Prr[v = \cE_i(z)]}\left(\ket{\psi_v}_{\sZ'} - \ket{\psi_{\tp_{1|z}}}_{\sZ'}\right)\ket{v}_\sV\junk{v}_{\sV'}\right|^2 \nonumber\\
    =& \left| \sum_v \sqrt{\Prr[v = \cE_i(z)]}\left(\ket{\psi_v} - \ket{\psi_{\tp_{1|z}}}\right)\ket{v} \right|^2\nonumber\\
    =& \sum_v \Prr[v = \cE_i(z)] \left|\ket{\psi_v} - \ket{\psi_{\tp_{1|z}}}\right|^2
\end{align}
 SubClaim \ref{subclm:good-est} shows that all but a negligible fraction of the probability mass in the output of $\cE_i$ is on $v$ that are $n/\sqrt{p(n)}$ close to $\tp_{1|z}$. We can therefore bound the contribution $v$ values that are not close, i.e. for $\bbV := \left\{v : \left|v - \tp_{1|z}\right| \leq \frac{n}{\sqrt{p(n)}}\right\}$, we see that terms not in $\bbV$ contribute negligible amounts to the sum.
 \begin{align}
 \label{eq:amp-step-error-negl}
     &\sum_v \Prr[v = \cE_i(z)] \left|\ket{\psi_v} - \ket{\psi_{\tp_{1|z}}}\right|^2\nonumber\\
     &= \sum_{v \in \bbV} \Prr[v = \cE_i(z)] \left|\ket{\psi_v} - \ket{\psi_{\tp_{1|z}}}\right|^2 + \sum_{v \notin \bbV} \Prr[v = \cE_i(z)] \left|\ket{\psi_v} - \ket{\psi_{\tp_{1|z}}}\right|^2\nonumber\\
     &\leq \sum_{v \in \bbV} \Prr[v = \cE_i(z)] \left|\ket{\psi_v} - \ket{\psi_{\tp_{1|z}}}\right|^2 + 4\sum_{v \notin \bbV} \Prr[v = \cE_i(z)]\nonumber\\
     &\leq \sum_{v \in \bbV} \Prr[v = \cE_i(z)] \left|\ket{\psi_v} - \ket{\psi_{\tp_{1|z}}}\right|^2 + 8e^{-2n^2}
 \end{align}
 where the fourth step uses SubClaim \ref{subclm:good-est} to show that $\sum_{v \notin \bbV} \Prr[v = \cE_i(z)]$. Now, for any $v\in \bbV$, since $v$ and  $\tp_{1|z}$ are close, we can bound the first term. By definition, $\ket{\psi_v} = \sqrt{v}\ket{1} + \sqrt{1-v}\ket{0}$ and $\ket{\psi_{\tp_{1|z}}} = \sqrt{\tp_{1|z}}\ket{1} + \sqrt{1-\tp_{1|z}}\ket{0}$. Therefore, for $v\notin \bbV$
 \begin{align*}
 \label{eq:amp-step-error-final}
     &\left|\ket{\psi_v} - \ket{\psi_{\tp_{1|z}}}\right|^2 \nonumber\\
     =& \left|\sqrt{v}\ket{1} + \sqrt{1-v}\ket{0} - \sqrt{\tp_{1|z}}\ket{1} - \sqrt{1-\tp_{1|z}}\ket{0}\right|^2\nonumber\\
     \leq& \left(\sqrt{v} - \sqrt{\tp_{1|z}}\right)^2+ \left( \sqrt{1-v} - \sqrt{1-\tp_{1|z}}\right)^2\nonumber\\
     \leq& \left|\left(\sqrt{v} - \sqrt{\tp_{1|z}}\right)\left(\sqrt{v} + \sqrt{\tp_{1|z}}\right)\right|+ \left|\left( \sqrt{1-v} - \sqrt{1-\tp_{1|z}}\right)\left( \sqrt{1-v} + \sqrt{1-\tp_{1|z}}\right)\right|\nonumber\\
     \leq& 2\left|v - \tp_{1|z}\right|\nonumber\\
     \leq& 2n/\sqrt{p(n)}
 \end{align*}
 where the last step follows directly from the definition of $\bbV$. Substituting this bound in \eqref{eq:amp-step-error-negl} and plugging the result into \eqref{eq:amp-step-error-expand} gives
 \begin{align*}
     \left| \ket{\sigma_{z}}_{\sZ'\sV\sV'} - \ket{\psi_{\tp_{1|z}}}_{\sZ'}\ket{0}_{\sV}\adv{1}_{\sV'}\right|^2 &\leq \sum_{v \in \bbV} \Prr[v = \cE_i(z)] \left|\ket{\psi_v} - \ket{\psi_{\tp_{1|z}}}\right|^2 + 8e^{-2n^2}\\
     &\leq \sum_{v \in \bbV} \Prr[v = \cE_i(z)] \cdot 2n/\sqrt{p(n)} + 8e^{-2n^2}\\
     &\leq 2n/\sqrt{p(n)} + 8e^{-2n^2} \leq 3n/\sqrt{p(n)}
 \end{align*}
 where the last step uses the fact that $n$ is large enough. This concludes the proof of SubClaim \ref{subclm:amp-step-error}.
 \end{proof}
 
 For all $i \in [0,n-1]$, define the unitary $M_i$ that for all $z \in \bin^i$ maps $\ket{z}_{\sX_{1\ldots i}}\ket{0}_{\sX_{i+1}}\ket{0}_{\sA_i}\adv{1}_{\sA_{i'}}$ to $\ket{z}_{\sX_{1\ldots i}}\ket{\psi_{p_{1|z}}}_{\sX_{i+1}}\ket{0}_{\sA_i}\adv{1}_{\sA_{i'}}$. These represent steps in the amplitude step of Aaronson synthesis. Define $\adv{2}_{\sA'} := \bigotimes_j \adv{1}_{\sA'_j}$. Note that when $z$ is the empty string, $p_{b|z} = p_b$. 
    \begin{subclaim}
    \label{subclm:amp-ideal} For all $i \in [0,n-1]$
    \[
    M_i M_{i-1} \ldots M_1 M_0 \ket{0}_{\sX_{1\ldots i+1}\sA}\adv{2}_{\sA'} = \sum_{z\in\bin^{i+1}} \sqrt{p_z}\ket{z}_{\sX_{1\ldots i+1}}\ket{0}_{\sA}\adv{2}_{\sA'}
    \]
        
    \end{subclaim}
    \begin{proof}
        We prove by induction on $i$. Consider the base case when $i=0$. By definition, $\ket{\psi_{\tp_{1|z}}} = \sqrt{\tp_{1|z}}\ket{1} + \sqrt{1-\tp_{1|z}}\ket{0}$, so
        \begin{align*}
        M_0 \ket{0}_{\sX_1}\ket{0}_{\sA}\adv{2}_{\sA'} &= \ket{\psi_{p_{1|z}}}_{\sX_1}\ket{0}_{\sA}\adv{2}_{\sA'}\\
        &= \left(\sqrt{p_0}\ket{0}_{\sX_1} + \sqrt{p_1}\ket{1}_{\sX_1}\right)\ket{0}_{\sA}\adv{2}_{\sA'}\\
        &= \sum_{z\in\bin} \sqrt{p_z}\ket{z}_{\sX_{1}}\ket{0}_{\sA}\adv{2}_{\sA'}
        \end{align*}
        Suppose the SubClaim holds for all $i< i'$. Then by applying the induction hypothesis
        \begin{align*}
        &M_{i'} M_{i'-1} \ldots M_1 M_0 \ket{0}_{\sX_{1\ldots i'+1}\sA}\adv{2}_{\sA'} \\&= M_{i'} \sum_{z\in\bin^{i'}} \sqrt{p_z}\ket{z}_{\sX_{1\ldots i}}\ket{0}_{\sX_{i'+1}}\ket{0}_{\sA}\adv{2}_{\sA'}\\
        &= \sum_{z\in\bin^{i'}}\sqrt{p_z}\ket{z}_{\sX_{1\ldots i'}}\left(\sqrt{p_{0|z}}\ket{0}_{\sX_{i'+1}} + \sqrt{p_{1|z}}\ket{1}_{\sX_{i'+1}}\right)\ket{0}_{\sA}\adv{2}_{\sA'}\\
        &= \sum_{z\in\bin^{i'+1}} \sqrt{p_z}\ket{z}_{\sX_{1\ldots i'+1}}\ket{z}_{\sX_{1\ldots i'+1}}\ket{0}_{\sA}\adv{2}_{\sA'}
        \end{align*}
    Which concludes the proof of SubClaim \ref{subclm:amp-ideal}.
    \end{proof}
    Next we show via hybrid argument that $M_i$ can be replaced by $\tM_i$ one by one for each $i$. 
    Let $$\Delta_{s,c,i} := \frac{1}{2}\cdot\sum_{z\in\bin^{i}} p_z\left(\left|p_{0|z} - \tp_{0|z}\right| + \left|p_{1|z} - \tp_{1|z}\right|\right)$$
    Note that by the definition of $\Delta_{s,c}$, the contribution to $\Delta_{s,c}$ in the case when $b_{\text{mode}} = 0$ may be written as the expectation over $i$ of $\Delta_{s,c,i}$. Therefore
    \begin{align*}
     \Delta_{s,c} &\geq \Pr[b_{\text{mode}} = 0]\cdot\sum_i \Pr[i]\cdot \Delta_{s,c,i} \\
     &\geq \sum_i \Delta_{s,c,i}/2n
    \end{align*}
    \begin{subclaim}
    \label{subclm:amp-hybrids} For all $i \in [0,n]$,  define 
    \[
    \ket{\cent_i} :=  \tM_{n-1} \tM_{n-2} \ldots \tM_i M_{i-1} \ldots M_0 \ket{0}_{\sX\sA}\adv{2}_{\sA'}
    \] Then for all $i \in [0,n-1]$ 
    \[\left|\ket{\cent_{i+1}} - \ket{\cent_{i}}\right| \leq \sqrt{3n}/p(n)^{1/4} + \sqrt{2\Delta_{s,c,i}}\]
        \end{subclaim}

\begin{proof} By expanding the definitions,
   \begin{align*}
       \left|\ket{\cent_{i+1}} - \ket{\cent_{i}}\right|&= \left|\tM_{n-1} \tM_{n-2} \ldots \tM_{i+1}\left(M_{i} - \tM_{i}\right) M_{i-1} \ldots M_0 \ket{0}_{\sX\sA}\adv{2}_{\sA'} \right|\\
       &=\left|\left(M_{i} - \tM_{i}\right) M_{i-1} \ldots M_0 \ket{0}_{\sX\sA}\adv{2}_{\sA'} \right|
   \end{align*}
   Applying SubClaim \ref{subclm:amp-ideal} to $M_{i-1} \ldots M_0 \ket{0}_{\sX\sA}\adv{2}_{\sA'}$ we get
   \begin{align*}
       \left|\ket{\cent_{i+1}} - \ket{\cent_{i}}\right|&= \left|\left(M_{i} - \tM_{i}\right) \sum_{z\in\bin^{i}} \sqrt{p_z}\ket{z}_{\sX_{1\ldots i}}\ket{0}_{\sX_{i+1\ldots n}\sA}\adv{2}_{\sA'} \right|
   \end{align*}
Expanding $M_{i} \ket{z}_{\sX_{1\ldots i}}\ket{0}_{\sX_{i+1\ldots n}\sA}\adv{2}_{\sA'}$
\[
M_{i} \ket{z}_{\sX_{1\ldots i}}\ket{0}_{\sX_{i+1\ldots n}\sA}\adv{2}_{\sA'} = \ket{z}_{\sX_{1\ldots i}}\ket{\psi_{{p_{1|z}}}}_{\sX_{i+1}}\ket{0}_{\sX_{i+2\ldots n}\sA}\adv{2}_{\sA'}
\]
Expanding $\tM_{i} \ket{z}_{\sX_{1\ldots i}}\ket{0}_{\sX_{i+1\ldots n}\sA}\adv{2}_{\sA'}$ 
\begin{align*}
&\tM_{i} \ket{z}_{\sX_{1\ldots i}}\ket{0}_{\sX_{i+1\ldots n}\sA}\adv{2}_{\sA'} \\
&=\tM_{i} \ket{z}_{\sX_{1\ldots i}}\ket{0}_{\sX_{i+1\ldots n}\sA}\left(\adv{1}_{\sA'_1}\ldots\adv{1}_{\sA'_n}\right)\\
&= \ket{z}_{\sX_{1\ldots i}}\ket{\sigma_{z}}_{\sX_{i+1}\sA_i\sA'_i}\ket{0}_{\sX_{i+2\ldots n}\sA_{1\ldots i-1, i+1 \ldots n}}\left(\adv{1}_{\sA'_1}\ldots\adv{1}_{\sA'_i-1}\adv{1}_{\sA'_i+1}\ldots\adv{1}_{\sA'_n}\right)
\end{align*}
When taking the norm of the difference, the registers $\sX_{i+2\ldots n},\sA_{1\ldots i-1, i+1 \ldots n},  \sA'_{1\ldots i-1, i+1 \ldots n}$ are identical in both states and may be ignored. Therefore by the triangle inequality
    \begin{align*}
       &\left|\ket{\cent_{i+1}} - \ket{\cent_{i}}\right|\\
       &= \left|\sum_{z\in\bin^{i}} \sqrt{p_z}\ket{z}\left(\ket{\psi_{{p_{1|z}}}}\ket{0}\adv{1} - \ket{\sigma_{z}}\right)\right|\\
       &=\left|\sum_{z\in\bin^{i}} \sqrt{p_z}\ket{z}\left(\ket{\psi_{{p_{1|z}}}}\ket{0}\adv{1} -\ket{\psi_{{\tp_{1|z}}}}\ket{0}\adv{1} + \ket{\psi_{{\tp_{1|z}}}}\ket{0}\adv{1} - \ket{\sigma_{z}}\right)\right|\\
       &\leq \left|\sum_{z\in\bin^{i}} \sqrt{p_z}\ket{z}\left(\ket{\psi_{{p_{1|z}}}}\ket{0}\adv{1} -\ket{\psi_{{\tp_{1|z}}}}\ket{0}\adv{1}\right)\right| +\left|\sum_{z\in\bin^{i}} \sqrt{p_z}\ket{z} \left(\ket{\psi_{{\tp_{1|z}}}}\ket{0}\adv{1} - \ket{\sigma_{z}}\right)\right|
   \end{align*}
   Consider the first term. Intuitively, $\tp_{1|z}$ is on average $\Delta_{s,c,i}$ far from $p_{1|z}$ (over randomness of $z$). We can therefore bound the first term as follows. 
   \begin{align*}
       &\left|\sum_{z\in\bin^{i}} \sqrt{p_z}\ket{z}\left(\ket{\psi_{{p_{1|z}}}}\ket{0}\adv{1} -\ket{\psi_{{\tp_{1|z}}}}\ket{0}\adv{1}\right)\right|\\
       &=\sqrt{\sum_{z\in\bin^{i}} p_z\left|\ket{\psi_{{p_{1|z}}}} -\ket{\psi_{{\tp_{1|z}}}}\right|^2}\\
       &=\sqrt{\sum_{z\in\bin^{i}} p_z\left|\left(\sqrt{p_{1|z}} - \sqrt{\tp_{1|z}}\right)\ket{1} + \left(\sqrt{p_{0|z}} - \sqrt{\tp_{0|z}}\right)\ket{0}\right|^2}\\
       &=\sqrt{\sum_{z\in\bin^{i}} p_z\left(\left(\sqrt{p_{1|z}} - \sqrt{\tp_{1|z}}\right)^2 + \left(\sqrt{p_{0|z}} - \sqrt{\tp_{0|z}}\right)^2\right)}\\
       &\leq\sqrt{\sum_{z\in\bin^{i}} p_z\left(\left|p_{1|z} - \tp_{1|z}\right| + \left|p_{0|z} - \tp_{0|z}\right|\right)}\\
       &\leq \sqrt{2\Delta_{s,c,i}}
   \end{align*} 
   Now, consider the second term. We have shown in SubClaim \ref{subclm:amp-step-error} that we synthesize a state close to $\ket{\psi_{{\tp_{1|z}}}}$, therefore
   \begin{align*}
       &\left|\sum_{z\in\bin^{i}} \sqrt{p_z}\ket{z} \left(\ket{\psi_{{\tp_{1|z}}}}\ket{0}\adv{1} - \ket{\sigma_{z}}\right)\right|\\
       &=\sqrt{\sum_{z\in \bin^i}p_z \left|\ket{\psi_{{\tp_{1|z}}}}\ket{0}\adv{1} - \ket{\sigma_{z}}\right|^2}\\
       &\leq\sqrt{\sum_{z\in \bin^i}p_z \cdot 3n/\sqrt{p(n)}}\\
       &=\sqrt{3n}/p(n)^{1/4}
   \end{align*}
   Where the third step follows from SubClaim \ref{subclm:amp-step-error}. Adding both error terms gives the final error and concludes the proof of SubClaim \ref{subclm:amp-hybrids}.
\end{proof}
Now, summing up the errors from SubClaim \ref{subclm:amp-hybrids}
\begin{align*}
    \left|\ket{\cent_n} - \ket{\cent_0}\right|\leq \sum_{i\in[0,n-1]} \left|\ket{\cent_{i+1}} - \ket{\cent_i}\right| \leq \left(n\cdot\sqrt{3n}/p(n)^{1/4} + \sum_i\sqrt{2\Delta_{s,c,i}}\right)
\end{align*}
Applying Jensen's inequality to the final term, followed by the definition of $\Delta_{s,c,i}$
\begin{align*}
    \left|\ket{\cent_n} - \ket{\cent_0}\right| \leq \left(n\cdot\sqrt{3n}/p(n)^{1/4} + \sqrt{2n \sum_i\Delta_{s,c,i}}\right) \leq \sqrt{3n^3}/p(n)^{1/4} + 2\sqrt{n\Delta_{s,c}}
\end{align*}
Finally, we note that by SubClaim \ref{subclm:amp-ideal}, 
\[\ket{\cent_n} = \sum_{z\in \bin^n} \sqrt{p_z}\ket{z}\ket{0}\adv{2} = \ket{\$^*_{s,c}}\ket{0}\adv{2}\]
and by definition of $\tM_{s,c}$ and $\ket{\cent_0}$
\[
\ket{\cent_0} = \tM_{s,c} \ket{0}\adv{2}
\]
Note that $\adv{amp}= \ket{0}\adv{2}$. Therefore
\[
    \left|\tM_{s,c} \ket{0}\adv{amp} - \ket{\$^*_{s,c}}\adv{amp}\right| \leq \sqrt{3n^3}/p(n)^{1/4} + 2\sqrt{n\Delta_{s,c}}
\]
and therefore
which concludes the proof of Claim \ref{clm:amp-synth-local}.\end{proof}
Claim \ref{clm:amp-synth-local} shows how to construct $\tM_{s,c}$ for each $(s,c)$. We now construct $\tM$ that takes $(s,c)$ as input and applies $\tM_{s,c}$.
\begin{claim}
    \label{clm:amp-synth}
 Let $\adv{amp}:=\ket{0}\adv{}^{\otimes np(n)}$. Then there exists an efficient unitary $\tM$ such that for all $s \in \bin^n$ and $c \in \cC_n$
	\[
	\left| \tM\ket{s,c}\ket{0}\adv{amp} -  \ket{s,c}\ket{\$^*_{s,c}}\adv{amp}\right|\leq \sqrt{3n^3}/p(n)^{1/4} + 2\sqrt{n\Delta_{s,c}}
	\]
\end{claim}
\begin{proof}
    Let $\tM$ be defined as $\sum_{s,c}\ketbra{s,c}\otimes \tM_{s,c}$. The statement then follows directly from Claim \ref{clm:amp-synth-local}.
\end{proof}
Claim \ref{clm:amp-synth} shows how the adversary may be used to synthesize a state close to $\ket{\$^*_{s,c}}$. The next step in Aaronson's synthesis is to coherently apply a phase to each basis vector $\ket{z}$ in $\ket{\$^*_{s,c}}$ to obtain $\ket{\$_{s,c}}$. We will show how to use the adversary to perform a similar task. Finally we will apply $\ct^\dagger$ to obtain an approximation to $\ket{\$_s}$.

\begin{claim}[State Synthesis] \label{clm:state-synth} Let $\adv{synth}:= \adv{}^{\otimes(2n+1)p(n)}$. There exists an efficient algorithm $\cB$ such that 
	\begin{align*}
    \underset{s, \ket{\$_s} \leftarrow \cG(1^n)}{\bbE}\left[\bra{\$_s}\cB(s, \adv{synth})\ket{\$_s}\right]& \geq 1 - 1/q(n)
    \end{align*}
\end{claim}
\begin{proof}
    We first construct an algorithm that takes input $z'$ and estimates the value of $\phi_{z'z} := \phi_{z'} - \phi_{z}$ for some fixed $z$, where $\phi_{z'}$ and $\phi_{z}$ area the arguments of the complex phases of $\ket{z}$ and $\ket{z'}$ respectively in $\ket{\$_{s,c}}$.
    For all $s,c$, for all $z, z' \in \bin^n$ define $\cU_z^{s,c}(z')$ that takes advice $\adv{}^{\otimes p(n)}$ as follows:
	\begin{itemize}
		\item For $j = 1$ to $np(n)^{1/4}$:
		\begin{itemize}
			\item $u_j \leftarrow A(s,c,1,z,z',0, \adv{})$
   			\item $v_j \leftarrow A(s,c,1,z,z',1, \adv{})$
	\end{itemize}
    \item $u' \leftarrow \sum_j u_j/np(n)^{1/4}$ and $u\leftarrow 2u'-1$
    \item $v' \leftarrow \sum_j v_j/np(n)^{1/4}$  and $v\leftarrow 2v'-1$
	\item Return $\arctanb(v,u)$
	\end{itemize}

The rest of the reduction proceeds as in the phase step of Aaronson's synthesis, except the oracle calls are replaced with estimates given by $\cU_{t'}^{s,c}(\cdot)$, where $t'$ is a pivot chosen by measuring an approximation of $\ket{\$^*_{s,c}}$.

Superposition queries to $\cU_z^{s,c}(\cdot)$ may produce entangled junk so we must later on uncompute to remove the junk.  
Let $U_z^{s,c}$ be the purification of $\cU_z^{s,c}$ that acts on input register $\sZ$, output register $\sV$, and advice register $\sV'$, i.e.,
 \[
 U_z^{s,c}\ket{z'}_\sZ\ket{0}_\sV\adv{1}_{\sV'} = \ket{z'}_\sZ \sum_v\sqrt{\Prr[v = \cU_z^{s,c}(z')]}\ket{v}_\sV\junk{v}_{\sV'}
 \]
 where $\adv{1} := \adv{}^{\otimes p(n)} \ket{0}$ and $\junk{v}$ is some normalized state.
 
 Let $P'$ \footnote{We may only be able to efficiently implement $P'$ upto some exponentially small error, however, this small error will not affect our result so we will elide it for the sake of clarity.} be a unitary that maps $\ket{v}_\sV$ to  $e^{-iv}\ket{v}_\sV$. Let $\adv{amp}$ and $\tM$ be as defined in Claim \ref{clm:amp-synth}. Define the algorithm $\cB'$ that takes input $(s,c)$ and  advice $\adv{amp}^{\otimes2}$ and $ \adv{1}$ as follows: 
 \begin{enumerate}
     \item Compute $\tM \ket{s,c}\ket{0^n}\adv{amp}$ and measure the second register to get $t'$.
     \item Compute $(U_{t'}^{s,c})^\dagger P'U_{t'}^{s,c}\left(\tM \ket{s,c}_{\sR}\ket{0^n}_{\sZ}\adv{amp}_\sA\right)\ket{0}_V\adv{1}_{V'}$
     \item Apply $c^\dagger$ to register $\sZ$
     \item Return $\sZ$
 \end{enumerate}
 Finally, let $\cB$ be the algorithm that takes input $s$, samples $c \leftarrow \cC_n$ and outputs $\cB'(s,c)$ using advice $\adv{amp}^{\otimes2}$ and $ \adv{1}$. Since $\adv{amp}=\ket{0}\adv{}^{\otimes np(n)}$ and $\adv{1} = \adv{}^{\otimes p(n)} \ket{0}$, the algorithm can be executed using advice $\adv{synth} = \adv{}^{\otimes(2n+1)p(n)}$.\\
 
\noindent \textbf{Analysis. } The reduction estimates $\phi_{z'z}$ by estimating the probabilities of certain binary outcome measurements. The candidate puzzle $\Gen$ is constructed so that any adversary that distributionally inverts the puzzle can be used to approximate the probabilities of the measurements with low error. However, small errors in the probability estimates can still result in large errors in the phase estimate. Essentially, $\cU_z^{s,c}(z')$ gives a good estimate of $\phi_{z'z}$ if 
 \begin{enumerate}
     \item[(a)] the adversary has low error when $y_0 = z$ and $y_1 = z'$, and 
     \item[(b)] the weights on $z$ and $z'$ in $\ket{\$_{s,c}}$ are not too far from each other.
 \end{enumerate}  Applying a unitary 2-design to the state flattens out the weights on the computational basis states, allowing us to argue that the weights on $z'$ and $z$ are close most of the time.
 
 Before we can analyse the above algorithm, we need to define some helpful sets. Recall $k$ is a constant such that $n^k \geq q(n)$ and $p(n) \geq n^{64k}$. For all $s$, we will define $\bbG_{s}$ as the set of $c$ where the probability that measuring $\ket{\$_{s,c}}$ in the computational basis results in a heavy $z$ is less than $1/n^k$, where a string $z$ is heavy if the probability mass of $z$ in $\ket{\$_{s,c}}$ is greater than $n^{3k}/2^n$. Intuitively, $\bbG_{s}$ is the set of $c$ such that $\ket{\$_{s,c}}$ has weight roughly evenly spread over computational basis states $z$, i.e., the total weight on heavy $z$ values is small. Formally, $\bbG_{s}:= \left\{c \text{ s.t. } \sum_{z: |\langle{z}|\$_{s,c}\rangle|^2 \leq n^{3k}/2^n} |\langle{z}|\$_{s,c}\rangle|^2 \leq 1/n^k \right\} $. We also define $\bbS$ as the set of $s,c$ such that the adversary has error less than $1/\sqrt{p}$ when $\s =s$ and $\ct = c$. Formally, $\bbS:= \left\{s, c \text{ s.t. } \Delta_{s,c} \leq 1/\sqrt{p}\right\} $.  
\begin{subclaim} 
\label{subclm:bbG}
For all $s$:
    \[
    \Prr_c[c\in\bbG_{s}]\geq 1 - 1/n^k
    \]
\end{subclaim}

\begin{proof}
    The proof follows directly from the following theorem, proved in Appendix \ref{appendix:2-design}.
    \begin{theorem}[Flatness of 2-designs]\label{thm:flatness-of-2-designs} Let $\cC$ be a unitary 2-design on $n$ qubits. Fix any $n$ qubit state $\ket{\psi}$. For any $C \in \Supp(\cC)$,  let $p_C(x):=|\bra{x}C\ket{\psi}|^2$ be the probability that measuring $C\ket{\psi}$ in the computational basis results in $x$. Then the following holds for all $k>6$ and sufficiently large $n$. Define $$\bbG := \left\{C \in \Supp(\cC) : \underset{x: p_C(x) \geq \frac{n^{3k}}{2^n}}{\sum} p_C(x) \leq 1/n^k \right\}$$ Then $$\Prr_{C\leftarrow\cC}[C\in\bbG]\geq 1 - 1/n^k$$
\end{theorem}
\end{proof}
\begin{subclaim}
\label{subclm:bbS}
    \[
    \Prr_{s,c}[s, c\in\bbS]\geq 1 - 1/p(n)^{1/2}
    \]
\end{subclaim}

\begin{proof}
   $\Delta_{s,c}$ is defined as the error when $\s = s$ and $\ct = c$. We can write the total error of $\cA$ (which is upper bounded by $1/p(n)$) as the expectation of $\Delta_{s,c}$, i.e.
    \begin{align*}
       \frac{1}{p(n)} &\geq \sum_{s, c} \Prr[\s = s] \cdot\Prr[\ct = c] \cdot\Delta_{s, c} \\
       &\geq \sum_{s, c \notin \bbS} \Prr[\s = s] \cdot\Prr[\ct = c] \cdot\Delta_{s, c}\\
       &\geq \sum_{s, c \notin \bbS} \Prr[\s = s] \cdot\Prr[\ct = c] \cdot\frac{1}{\sqrt{p(n)}},
   \end{align*}
    where the third step follows from the definition of $\bbS$. Therefore, $\Prr_{s,c}[s, c\notin\bbS]\leq \frac{1}{\sqrt{p(n)}}$.
\end{proof}

Define $\bbS':= \left\{s,c : s,c \in \bbS \wedge c \in \bbG_{s}\right\}$. We first note that with high probability $s,c \in \bbS'$. 
\begin{subclaim}
    \label{subclm:bbS'}
    \[
    \Prr_{s,c}[s, c\in\bbS']\geq 1 - 1/p(n)^{1/2} - 1/n^k
    \]
\end{subclaim}
\begin{proof}
    Follows from the definition of $\bbS'$ and Subclaims \ref{subclm:bbG} and \ref{subclm:bbS}.
\end{proof}

The goal of the remainder of the proof is to show that when $(s,c) \in \bbS'$, $\cB'(s,c)$ outputs a state close to $\ketbra{\$_s}$. Since w.h.p. $(s,c) \in \bbS'$, we can ignore case when $(s,c) \notin \bbS'$ at the cost of some small error probability which is incorporated in the final result. For the rest of the proof, we fix some $(s,c) \in \bbS'$, and drop the parameterization on $(s,c)$ in the notation.

Define some terms that will be useful for the proof.
\begin{itemize}
    \item Interpret $\ket{\$_{s,c}}$ as $\sum_{z \in \bin^n} a_z e^{-i\phi_z}\ket{z}$ where $a_z \geq 0$ and $\phi_z \in [0, 2\pi)$, and let $\alpha_z := a_z e^{-i\phi_z}$ and $\phi_{z'z} := \phi_{z'} - \phi_z$
    \item $\ket{\$^*_{s,c}} := \sum_{z \in \bin^n} a_z \ket{z}$. Intuitively, $\ket{\$_{s,c}^*}$ represents $\ket{\$_{s,c}}$ with the phase information removed, i.e., with real amplitudes.
\end{itemize}
We now describe how $y_0$ and $y_1$ are distributed during the execution of $\Gen$. Define $Y_0$ as the distribution on $y_0$ induced by $\Gen$.
Let $R:= \bin^n \setminus \{0^n\}$.
\begin{subclaim}
\label{subclm:y0}
    For all $y_0\in \bin^n$
    \[
    \Pr_{Y_0}[y_0] = \frac{1}{2}\left(|\alpha_{y_0}|^2\left(1-\frac{1}{|R|}\right) + \frac{1}{|R|}\right)
    \]
\end{subclaim}
\begin{proof}
    Recall that $y_0$ is generated by measuring the $F$ register in the computational basis to obtain $x_0$, and setting $y_0$ to be either $x_0$ or $x_0 \oplus r$ at random (where $r$ is uniformly sampled from $R$).
    \begin{align*} 
       \Pr_{Y_0}[y_0] &= \frac{1}{|R|} \sum_{r \neq 0} \frac{1}{2} \times \Prr[\text{measuring $F$ register gives $f_r(z) = \text{min}(y_0, y_0 \oplus r)$}] \\
       &= \frac{1}{2|R|} \sum_{r \neq 0} \left| \left(\mathbb{I} \otimes \ket{f_r(y_0)} \bra{f_r(y_0)}\right) \left(\sum_x \alpha_x \ket{x} \ket{f_r(x)} \right) \right|^2.
   \end{align*}
   Now, $f_r(y_0) = f_r(x)$ if and only if $x = {y_0}$ or $x = {y_0} \oplus r$, so only the terms, $\alpha_{y_0}\ket{{y_0}}\ket{f_r({y_0})}$ and $\alpha_{{y_0}\oplus r}\ket{{y_0}\oplus r}\ket{f_r({y_0})}$, are in the support of the projector. Therefore
   \begin{align*}
       \Pr_{Y_0}[y_0] &= \frac{1}{2|R|} \sum_{r \neq 0} \left| \alpha_{y_0} \ket{{y_0}} + \alpha_{{y_0} \oplus r} \ket{{y_0} \oplus r} \right|^2 \\
       &= \frac{1}{2|R|} \sum_{r \neq 0} \left( \left|\alpha_{y_0} \right|^2 + \left| \alpha_{{y_0} \oplus r} \right|^2\right) \\
       &= \frac{\left| \alpha_{y_0}\right|^2}{2} + \sum_{y \neq y_0} \frac{\left|\alpha_{y}\right|^2}{2|R|}.
   \end{align*}
   Since $\sum_{y} \left|\alpha_{y} \right|^2 = 1$,
   \begin{align*}
       & \Pr_{Y_0}[y_0] = \frac{\left| \alpha_{y_0} \right|^2}{2} + \frac{1 - \left|\alpha_{y_0} \right|^2}{2R},
   \end{align*}
   which after rearranging gives the statement in the claim.
\end{proof}
Define $Y^{y_0}_{1}$ as the distribution on $y_1$ induced by $\Gen$ conditioned on sampling $y_0$.
\begin{subclaim}
\label{subclm:y1}

    For all $y_1\in \bin^n$
    \[
    \Pr_{Y^{y_0}_{1}}(y_1) = \frac{|\alpha_{y_1}|^2+|\alpha_{y_0}|^2}{1+|\alpha_{y_0}|^2(|R|-1)}
    \]
\end{subclaim}
\begin{proof}
   By the definition of $Y^{y_0}_{1}$
    \begin{align*}
        \Pr_{Y^{y_0}_{1}} [y_1] &= \Prr[y_1| y_0] \\
        &= \frac{\Prr[y_1 \wedge y_0]}{\Pr_{Y_0}[y_0]}
    \end{align*}
    Recall that $y_0$ and $y_1$ are a random permutation of $x_0$ and $x_0\oplus r$ where $x_0$ is obtained by measuring the $F$ register in the computational basis and $r$ is uniformly sampled from $R$. Then, $y_0$ and $y_1$ are obtained with probability $\frac{1}{2}$ conditioned on $r = y_0 \oplus y_1$ and the measurement outcome of $F$ is $f_r(y_0) = f_r(y_0 \oplus r) = f_r(y_1)$. Thus,
    \begin{align*}
        \Pr_{Y_1^{y_0}}[y_1] &= \frac{1}{2} \cdot \frac{\Prr[\text{measuring $F$ register gives $f_r(y_0)$ } \wedge r={y_0}\oplus {y_1}]}{\Pr_{Y_0}[y_0]} \\
        &= \frac{1}{|R|} \frac{\left|\left( \mathbb{I} \otimes \ket{f_{{y_0} \oplus {y_1}}({y_0})} \bra{f_{{y_0} \oplus {y_1}}({y_0})} \right) \left(\sum_x \alpha_x \ket{x} \ket{f_{{y_0} \oplus {y_1}}(x)} \right) \right|^2}{\Pr_{Y_0}[y_0]} \\
        &= \frac{\frac{1}{|R|} \left(|\alpha_{y_0}|^2 + |\alpha_{{y_1}}|^2 \right)}{|\alpha_{y_0}|^2 \left(1 - \frac{1}{|R|} \right) + \frac{1}{|R|}} \\
        &= \frac{|\alpha_{y_0}|^2 + |\alpha_{{y_1}}|^2}{1 + |\alpha_{y_0}|^2 (|R| - 1)},
    \end{align*}
    where the third step follows from Subclaim \ref{subclm:y0}.
\end{proof}
Define $\Delta'_{y_0}$ as follows:
    \begin{itemize}
        \item Let $D_0$ be the distribution of $(\pi, \beta)$ when $\pi, \beta$ is sampled by $\Gen$ conditioned on $\Gen$ sampling $y_0$ and $s = \s, c = \ct, b_\text{mode}=1$.
        \item Let $D_1$ be the distribution of $(\pi, A(\pi, \adv{}))$ when $\pi$ is sampled by $\Gen$ conditioned on $\Gen$ sampling $y_0$ and  $s = \s, c = \ct, b_\text{mode}=1$.
        \item  $\Delta'_{y_0} := SD(D_0, D_1)$
    \end{itemize}
    Define $\Delta'_{y_0,y_1}$ as follows:
    \begin{itemize}
        \item Let $D_0$ be the distribution of $(\pi, \beta)$ when $\pi, \beta$ is sampled by $\Gen$ conditioned on $\Gen$ sampling $y_0$ and $y_1$, while $s = \s, c = \ct, b_\text{mode}=1$.
        \item Let $D_1$ be the distribution of $(\pi, A(\pi, \adv{}))$ when $\pi$ is sampled by $\Gen$ conditioned on on $\Gen$ sampling $y_0$ and $y_1$, while $s = \s, c = \ct, b_\text{mode}=1$.
        \item  $\Delta'_{y_0,y_1} := SD(D_0, D_1)$
    \end{itemize}
Intuitively, $\Delta'_{y_0}$ (and $\Delta'_{y_0,y_1}$) represent the adversary's error conditioned on sampling $y_0$ (and $y_1$).

Define the following sets
\begin{itemize}
    \item $\bbZ: =\left\{z \text{ s.t. } \Delta'_z \leq 1/p(n)^{1/4}\right\}$
    \item $\bbZ'_{z}: =\left\{z'\neq z \text{ s.t. } \Delta'_{z,z'} \leq 1/2p(n)^{1/8}\right\}$
    \item $\bbL: =\left\{z \text{ s.t. } |\alpha_z|^2 \geq \frac{1}{n^{3k}2^n}\right\}$
    \item $\bbU: =\left\{z \text{ s.t. } |\alpha_z|^2 \leq \frac{n^{3k}}{2^n}\right\}$    
\end{itemize}
Intuitively, $\bbZ$ and $\bbZ'_{z}$ are sets of strings on which the adversary's error is low, while $\bbL \cap \bbU$ is the set of strings whose probability mass in $\ket{\$_{s,c}}$ not too far from $1/2^n$. The next subclaim shows that
\begin{enumerate}
    \item[(a)] If $z$ is sampled by measuring $\ket{\$^*_{s,c}}$, with high probability $z \in \bbZ \cap \bbL \cap \bbU$ 
    \item[(b)] If $z \in \bbZ \cap \bbL \cap \bbU$ then most of the probability mass of $\ket{\$_{s,c}}$ is on strings $z'$ such that $z' \in \bbZ'_{z} \cap \bbL \cap \bbU$
\end{enumerate}
\begin{subclaim}\
\label{subclm:markov-sets}
    \begin{enumerate}
        \item  $\sum_{z\notin \bbZ} |\alpha_z|^2 \leq 3/p(n)^{1/4}$
        \item $\sum_{z\notin \bbL} |\alpha_z|^2 \leq 1/n^{3k}$
         \item  $\sum_{z\notin \bbU} |\alpha_z|^2 \leq 1/n^{k}$
         \item For all $z\in \bbU\cap \bbZ$, $\sum_{z'\notin \bbZ'_{z} } |\alpha_{z'}|^2 \leq 4n^{3k}/p(n)^{1/8}$
    \end{enumerate}

\end{subclaim}
\begin{proof}
\begin{enumerate}
    \item Recall that $\Delta_{s,c}$ is the adversary's error in sampling  when $\s = s$ and $\ct = c$. $\Delta_{s,c}$ can therefore be expressed as the expectation over $y_0$ of $\Delta'_{y_0}$, i.e. 
    \begin{align*}
        \Delta_{s, c} &\geq \sum_{y_0} \Pr_{Y_0}[y_0] \Delta'_{y_0} \\
        &\geq \sum_{y_0} \Delta'_{y_0} \left[ \frac{|\alpha_{y_0}|^2}{2} \left( 1 - \frac{1}{|R|}\right) + \frac{1}{2|R|} \right] \\
        &\geq \sum_{y_0} \Delta'_{y_0} \left[ \frac{|\alpha_{y_0}|^2}{2} \left(1 - \frac{1}{|R|} \right) \right] \\
        &\geq \sum_{y_0} \Delta'_{y_0} \frac{|\alpha_{y_0}|^2}{3},
    \end{align*}
    where the second step uses Subclaim \ref{subclm:y0}. By the definition of $\bbS'$, $\Delta_{s, c} \leq \frac{1}{\sqrt{p(n)}}$, therefore
    \begin{align*}
        \frac{1}{\sqrt{p(n)}} &\geq \sum_{y_0} \frac{\Delta'_{y_0} |\alpha_{y_0}|^2}{3} \\
        &\geq \sum_{{y_0} \notin \bbZ} \frac{\Delta'_{y_0} |\alpha_{y_0}|^2}{3} \\
        &\geq \sum_{{y_0} \notin \bbZ} \frac{ |\alpha_{y_0}|^2}{3p(n)^\frac{1}{4}},
    \end{align*}
    where the last step follows from the definition of $\bbZ$. After rearranging and relabeling,
    \begin{align*}
        &\sum_{z \notin \bbZ} |\alpha_z|^2 \leq \frac{3}{p(n)^\frac{1}{4}}.
    \end{align*}

    \item By the definition of $\bbL$
    \begin{align*}
        \sum_{z \notin \bbL} |\alpha_z|^2 &\leq \sum_{z \notin \bbL} \frac{1}{n^{3k} 2^n} \\
        &\leq \sum_{z} \frac{1}{n^{3k} 2^n} \\
        &= \frac{1}{n^{3k}}.
    \end{align*}

    \item Recall that $\bbG_{s}= \left\{c \text{ s.t. } \sum_{z: |\langle{z}|\$_{s,c}\rangle|^2 \leq n^{3k}/2^n} |\langle{z}|\$_{s,c}\rangle|^2 \leq 1/n^k \right\} $. We fixed $s,c \in \bbS'$, which by definition implies $c\in \bbG_s$. Therefore
    \begin{align*}
        \frac{1}{n^{k}} &\geq \sum_{z: |\langle z|\$_{s,c} \rangle|^2 \leq \frac{n^{3k}}{2^n}} |\langle z|\$_{s,c}\rangle|^2 \\
        &= \sum_{z \notin \bbU} |\langle z | \$_{s,c} \rangle|^2 \\
        &= \sum_{z \notin \bbU} |\alpha_z|^2
    \end{align*}

    \item By the definition of $\Delta'_{y_0}$ and $\Delta'_{y_0, y_1}$, $\Delta'_{y_0}$ can be expressed as the expectation over $y_1$ of $\Delta'_{y_0, y_1}$. Therefore
    \begin{align*}
        \Delta'_{y_0} &= \sum_{y_1} \Pr_{Y_1^{y_0}}[y_1] \cdot \Delta'_{y_0,y_1} \\
        &\geq \sum_{y_1 \notin \mathbb{Z}'_{y_0}} \Pr_{Y_1^{y_0}}[y_1] \cdot \Delta'_{y_0,y_1} \\
        &\geq \sum_{y_1 \notin \mathbb{Z}'_{y_0}} \Pr_{Y_1^{y_0}}[y_1] \cdot \frac{1}{2p(n)^\frac{1}{8}},
    \end{align*}
    where the last step follows from the definition of $\bbZ'_{y_0}$. Therefore,
    \begin{align*}
        \sum_{y_1 \notin \mathbb{Z}'_{y_0}} \Pr_{Y_1^{y_0}}[y_1] &\leq \Delta'_{y_0} \cdot 2p(n)^\frac{1}{8} \\
        &\leq \frac{2}{p(n)^\frac{1}{8}},
    \end{align*}
    where the last step follows from $y_0 \in \bbZ$. By Subclaim \ref{subclm:y1},
    \begin{align*}
        \Pr_{Y_1^{y_0}}[y_0] &= \frac{|\alpha_{y_0}|^2 + |\alpha_{y_1}|^2}{1 + |\alpha_{y_0}|^2 (|R| - 1)} \\
        &\geq \frac{|\alpha_{y_1}|^2}{1 + n^{3k} (|R| - 1) 2^{-n}} \\
        &\geq \frac{|\alpha_{y_1}|^2}{2n^{3k}},
    \end{align*}
    where the second step uses $z \in \bbU$ and the last step holds for large enough $n$. Substituting in the previous equation, rearranging, and relabeling gives
    \begin{align*}
        & \sum_{z' \in \mathbb{Z}'_{z}} |\alpha_{z'}|^2 \geq \frac{4n^{3k}}{p(n)^\frac{1}{8}}.
    \end{align*}
\end{enumerate}
\end{proof}
Next we show that $\cU_{t'}(z')$ gives a good estimate for $\phi_{z't'}$ when $t' \in \bbZ \cap \bbL \cap \bbU$ and $z' \in \bbZ'_{t'} \cap \bbL \cap \bbU$. 
 \begin{subclaim} \label{subclm:phase-estimator}For all $z \in \bbL \cap \bbU \cap \bbZ$ and for all $z'\in \bbL \cap \bbU \cap \bbZ'_{z}$\footnote{We show that the resulting complex phases are close instead of showing that the angles are close. This is because the estimate of the angle may have a $2\pi$ error, but this has no operational meaning when applying the phase.}
 \[
  \Prr\left[\left|e^{-i\cdot\cU_z(z')} - e^{-i\phi_{z'z}}\right| > \frac{8\sqrt{2}n^{6k}}{p^{1/8}} \right] \leq \negl(n)
 \]
 \end{subclaim}
\begin{proof}
  Fix any $z \in \bbL \cap \bbU \cap \bbZ$ and any $z'\in \bbL \cap \bbU \cap \bbZ'_{z}$. Recall that for all $x$, $a_x = |\alpha_x|$ and $e^{-i\phi_x} = \alpha_x/|\alpha_x|$. Also recall that $\phi_{z'z} = \phi_{z'} - \phi_z$. Consider the state $\ket{\psi_\text{post}}$ conditioned on obtaining $y_0 = z$ and $y_1 = z'$ during the execution of $\Gen$.
    \[
    \ket{\psi_\text{post}} = \frac{\alpha_z\ket{z} + \alpha_{z'}\ket{z'}}{\sqrt{|\alpha_{z}|^2 + |\alpha_{z'}|^2}} = e^{-i\phi_z}\cdot \frac{a_z\ket{z} + a_{z'}e^{-i\phi_{z'z}}\ket{z'}}{\sqrt{|\alpha_{z}|^2 + |\alpha_{z'}|^2}}
    \]
    Let $\theta$ be the unique angle in $[0,\pi]$ such that $\cos\frac{\theta}{2} = \frac{a_z}{\sqrt{a_z^2 + a_{z'}^2}}$ and $\sin\frac{\theta}{2} = \frac{a_{z'}}{\sqrt{a_z^2 + a_{z'}^2}}$. Let $\ket{\psi'} := e^{i\phi_z}\ket{\psi_\text{post}}$. Then 
    \[
    \ket{\psi'} = \cos\left(\theta/2\right)\ket{z} + e^{-i\phi_{z'z}} \sin\left(\theta/2\right)\ket{z'}
    \]
    Let $A'$ be the probability that applying $V_{z,z',0}$ to $\ket{\psi_\text{post}}$ and measuring results in output $z$. We first obtain an expression for $A'$ in terms of $\theta$ and $\phi_{z'z}$ by ignoring global phase and expanding
    \begin{align*}
        A' &= \left|\bra{z}V_{z,z',0}\ket{\psi_\text{post}}\right|^2\\
        &= \left|\bra{z}V_{z,z',0}\ket{\psi'}\right|^2\\
        &= \left|\bra{z}\left(\frac{\cos\left(\theta/2\right) + e^{-i\phi_{z'z}} \sin\left(\theta/2\right)}{\sqrt{2}} \ket{z} + \frac{\cos\left(\theta/2\right) - e^{-i\phi_{z'z}} \sin\left(\theta/2\right)}{\sqrt{2}} \ket{z'}\right)\right|^2\\
        &=\left|\frac{\cos\left(\theta/2\right) + e^{-i\phi_{z'z}} \sin\left(\theta/2\right)}{\sqrt{2}}\right|^2\\
        &=\frac{\cos\left(\theta/2\right) + e^{-i\phi_{z'z}} \sin\left(\theta/2\right)}{\sqrt{2}}\cdot \frac{\cos\left(\theta/2\right) + e^{i\phi_{z'z}} \sin\left(\theta/2\right)}{\sqrt{2}}\\
        &=\frac{\cos^2\left(\theta/2\right) + \sin^2\left(\theta/2\right) + 2\cos\left(\theta/2\right)\sin\left(\theta/2\right)\left(e^{-i\phi_{z'z}}+e^{i\phi_{z'z}}\right)}{2}\\
        &=\frac{1 + \sin\theta\left(e^{-i\phi_{z'z}}+e^{i\phi_{z'z}}\right)}{2}\\
        &=\frac{1 + \sin\theta\cos\phi_{z'z}}{2}
    \end{align*}
    Similarly, let $B'$ be the probability that applying $V_{z,z',1}$ to $\ket{\psi_\text{post}}$ and measuring results in output $z$. Then
    \begin{align*}
        B' &= \left|\bra{z}V_{z,z',1}\ket{\psi_\text{post}}\right|^2\\
        &= \left|\bra{z}V_{z,z',1}\ket{\psi'}\right|^2\\
        &= \left|\bra{z}\left(\frac{\cos\left(\theta/2\right) + ie^{-i\phi_{z'z}} \sin\left(\theta/2\right)}{\sqrt{2}} \ket{z} + \frac{\cos\left(\theta/2\right) - ie^{-i\phi_{z'z}} \sin\left(\theta/2\right)}{\sqrt{2}} \ket{z'}\right)\right|^2\\
        &=\left|\frac{\cos\left(\theta/2\right) + ie^{-i\phi_{z'z}} \sin\left(\theta/2\right)}{\sqrt{2}}\right|^2\\
        &=\frac{\cos\left(\theta/2\right) + ie^{-i\phi_{z'z}} \sin\left(\theta/2\right)}{\sqrt{2}}\cdot \frac{\cos\left(\theta/2\right) -ie^{i\phi_{z'z}} \sin\left(\theta/2\right)}{\sqrt{2}}\\
        &=\frac{\cos^2\left(\theta/2\right) + \sin^2\left(\theta/2\right) + 2i\cos\left(\theta/2\right)\sin\left(\theta/2\right)\left(e^{-i\phi_{z'z}}-e^{i\phi_{z'z}}\right)}{2}\\
        &=\frac{1 + \sin\theta\left(ie^{-i\phi_{z'z}}-ie^{i\phi_{z'z}}\right)}{2}\\
        &=\frac{1 + \sin\theta\sin\phi_{z'z}}{2}
    \end{align*}
    Let $A:=2A'-1 = \sin\theta\cos\phi_{z'z}$ and $B:=2B'-1 =\sin\theta\sin\phi_{z'z}$. We note that since $\theta \in [0,\pi)$, $\sin\theta \geq 0$ and therefore $\arctanb(B, A) = \phi_{z'z}$.

    Let $\tA' := \Pr[1 = \cA(s,c,1,z,z',0)]$ and let $\tB' := \Pr[1 = \cA(s,c,1,z,z',1)]$. Since $z' \in \bbZ'_{z}$
    \begin{align*}
        1/2p(n)^{1/8} &\geq \Delta'_{z,z'}\\
        &= \frac{1}{2} \left(\left|\tA' - A'\right| + \left|\tB' - B'\right| \right)
    \end{align*} 
    where the second step follows directly from the definition of $\Delta'_{z,z'}$. As a result
    \begin{gather*}
        \left|\tA' - A'\right| \leq 1/p(n)^{1/8}\\
        \left|\tB' - B'\right| \leq 1/p(n)^{1/8}
    \end{gather*}
    Let $\tA := 2\tA' - 1$ and $\tB:= 2\tB' - 1$. Therefore
    \begin{gather*}
        \left|\tA - A\right| \leq 2/p(n)^{1/8}\\
        \left|\tB - B\right| \leq 2/p(n)^{1/8}
    \end{gather*}
    i.e. $\tA$ is close to $A$ and $\tB$ is close to $B$.Consider the case when $\cU_z$ is run on $z'$ and internally samples $u$ and $v$. By setting $\delta = \sqrt{n}$ in the additive Chernoff bound (Theorem $\ref{thm:chernoff-additive}$), we see that $u$ and $v$ that are computed by $\cU_z$ are good approximations of $\tA$ and $\tB$, and thus of $A$ and $B$. Formally, 
    \begin{gather*}
    \Pr\left[\left|u - \tA\right| \geq \frac{2}{p(n)^{1/8}}\right]\leq 2e^{-2n}\\
        \Pr\left[\left|v - \tB\right| \geq \frac{2}{p(n)^{1/8}}\right]\leq 2e^{-2n}
         \end{gather*}
         Using the fact that $\tA$ and $\tB$ are close to $A$ and $B$ as shown above, we can bound the Euclidean distance between $(u,v)$ and $(A,B)$
        \begin{align}
        \label{eq:within-circle-whp}
   1 - 4e^{-2n} &\leq \Prr\left[\left|v - \tB\right| \leq \frac{2}{p(n)^{1/8}} \text{ and } \left|u - \tA\right| \leq \frac{2}{p(n)^{1/8}}\right]\nonumber
       \\
       &\leq \Prr\left[\left|v - B\right| \leq \frac{4}{p(n)^{1/8}} \text{ and } \left|u - A\right| \leq \frac{4}{p(n)^{1/8}}\right]\nonumber\\
       &\leq \Prr\left[\left(v - B\right)^2 + \left(u - A\right)^2 \leq \left(\frac{4\sqrt{2}}{p(n)^{1/8}}\right)^2\right]
    \end{align}
We will use the following theorem which we prove in Appendix \ref{appendix:geometric-arg}
\begin{theorem}    \label{thm:geometric-arg}
        Let $(x,y), (x^*, y^*) \in \bbR^2$ such that $\exists \gamma >0, \gamma'>0$ s.t.
        \begin{itemize}
            \item $x^2 + y^2 \geq \gamma^2$
            \item $(x-x^*)^2 + (y-y^*)^2 \leq (\gamma')^2$
            \item $\gamma' < \gamma$
        \end{itemize}
        Then $|e^{-i\cdot\arctanb(y,x)} - e^{-i\cdot\arctanb(y^*, x^*)}| \leq 2\gamma'/\gamma$
    \end{theorem}
    Using the fact that both $z$ and $z'$ belong to $\bbL \cap \bbU$, we can show that $(A,B)$ is atleast $1/n^{6k}$ far from the origin.
    \begin{align*}
        A^2 + B^2 &= \sin^2\theta \left(\cos^2\phi_{z'z} + \sin^2\phi_{z'z}\right)\\
        &= \sin^2 \theta \\
        &= \left(2 \sin(\theta/2)\cos(\theta/2)\right)^2\\
        &= \left(\frac{2 a_z a_{z'}}{a^2_z + a^2_{z'}}\right)^2\\
        &= \left(\frac{2 }{a_z/a_{z'} + a_{z'}/a_z}\right)^2\\
        &\geq \left(\frac{2 \cdot \frac{1}{n^{3k} 2^n}}{2n^{3k}/2^n}\right)^2\\
        &\geq \left(1/n^{6k}\right)^2
    \end{align*}
    where the fourth step follows from the definition of $\theta$ and the fifth step follows from bounds on $a_z$ and $a_{z'}$ implied by $z,z' \in\bbL\cap\bbU$.
Recall that $p(n) >n^{64k}$ and $k > 6$. Therefore, for large enough $n$, $p(n) > 4\sqrt{2/n^{6k}}$. Equation \eqref{eq:within-circle-whp} thus implies that we can apply Theorem \ref{thm:geometric-arg} with probability atleast $1-4e^{-2n}$ when we set $(x,y) = (A,B), (x^*,y^*) = (u,v), \gamma = 1/n^{6k}$ and $\gamma' = 4\sqrt{2}/p(n)^{1/8}$ . Formally,
\begin{gather*} 
    \Pr\left[\left|e^{-i\cdot\arctanb(v,u)}- e^{-i\cdot\arctanb(B,A)}\right| \leq \left(\frac{8\sqrt{2 }n^{6k}}{p(n)^{1/8}}\right)\right]\geq 1 - 4e^{-2n}\\
    \implies \Prr\left[\left|e^{-i\cdot\cU_z(z')} - e^{-i\cdot\phi_{z'z}}\right| >\frac{8\sqrt{2}n^{6k}}{p(n)^{1/8}} \right] \leq \negl(n)
\end{gather*}
which concludes the proof of the subclaim.
\end{proof}
Next we show that if $z\in \bbL\cap \bbU \cap \bbZ$, then the oracle responses in the phase step of Aaronson's synthesis can be answered by the estimator's outputs.
\begin{subclaim}
\label{subclm:phase-synth-for-good-sc}
    For all $z \in \bbL\cap \bbU \cap \bbZ$, for sufficiently large $n$
    \[
    \left|\ket{\$_{s,c}}_{\sZ}\ket{0}_{\sV}\adv{1}_{\sV'}-e^{-i\phi^{z,c}_z}(U_{z})^\dagger P'U_{z}\ket{\$^*_{s,c}}_{\sZ}\ket{0}_{\sV}\adv{1}_{\sV'}\right| \leq 
    \frac{8\sqrt{2}n^{6k}}{p(n)^{1/8}} + \frac{2}{\sqrt{n^k}} + \frac{2}{\sqrt{n^{3k}}} +\frac{4\sqrt{n^{3k}}}{p(n)^{1/16}}
    \]
\end{subclaim}
\begin{proof}
Applying a global phase to the difference does not alter the magnitude, so we multiply by $e^{-i\phi_z}$
\begin{align}
    \label{eq:phase-1}
        \zeta :=& \left| \ket{\$}_{\sZ}\ket{0}_{\sV}\adv{1}_{\sV'}-e^{-i\phi_z}U_{z}^\dagger P'U_{z}\ket{\$^*}_{\sZ}\ket{0}_{\sV}\adv{1}_{\sV'}\right| \nonumber\\
    =& \left| e^{i\phi_z}U_{z}\ket{\$}_{\sZ}\ket{0}_{\sV}\adv{1}_{\sV'}- P'U_{z}\ket{\$^*}_{\sZ}\ket{0}_{\sV}\adv{1}_{\sV'}\right|
\end{align}
    Next, we use the definition of $U_z$ and $P'$ to expand each term. Expanding $e^{i\phi_z}U_{z}\ket{\$}_{\sZ}\ket{0}_{\sV}\adv{1}_{\sV'}$
\begin{align}
\label{eq:phase-error-1}
    e^{i\phi_z}U_{z}\ket{\$}_{\sZ}\ket{0}_{\sV}\adv{1}_{\sV'} &= \sum_{z'}\alpha_{z'}e^{i\phi_z}\ket{z'}_\sZ\sum_v\sqrt{\Prr[v = \cU_z(z')]}\ket{v}_\sV\junk{v}_{\sV'}\nonumber\\
    &= \sum_{z'}a_{z'}e^{-i\phi_{z'z}}\ket{z'}_\sZ\sum_v\sqrt{\Prr[v = \cU_z(z')]}\ket{v}_\sV\junk{v}_{\sV'}
\end{align}
    Expanding $P'U_{z}\ket{\$^*}_{\sZ}\ket{0}_{\sV}\adv{1}_{\sV'}$
\begin{align}
\label{eq:phase-error-2}
    P'U_{z}\ket{\$^*}_{\sZ}\ket{0}_{\sV}\adv{1}_{\sV'} &= \sum_{z'}a_{z'}\ket{z'}_\sZ\sum_v e^{-iv}\sqrt{\Prr[v = \cU_z(z')]}\ket{v}_\sV\junk{v}_{\sV'}
\end{align}
Define $\bbA := \bbL\cap\bbU\cap\bbZ$ and $\bbB:= \bbL\cap\bbU\cap\bbZ'_z$. Plugging \eqref{eq:phase-error-1} and \eqref{eq:phase-error-2} into \eqref{eq:phase-1} and squaring gives
\begin{align}
\label{eq:phase-step-error-expand}
  \zeta^2 &= \left| \sum_{z'}a_{z'}\ket{z'}_\sZ\sum_v \left(e^{-i\phi_{z'z}}-e^{-iv}\right)\sqrt{\Prr[v = \cU_z(z')]}\ket{v}_\sV\junk{v}_{\sV'}\right|^2\nonumber\\
  & =\sum_{z'}a^2_{z'}\sum_v \Prr[v = \cU_z(z')]\left|e^{-i\phi_{z'z}}-e^{-iv}\right|^2\nonumber\\
  &\leq\sum_{z' \in \bbB}a^2_{z'}\sum_v \Prr[v = \cU_z(z')]\left|e^{-i\phi_{z'z}}-e^{-iv}\right|^2 + \sum_{z' \notin \bbB}4a^2_{z'}
\end{align}    
Let $\delta:= \frac{8\sqrt{2}n^{6k}}{p(n)^{1/8}}$. For all $z'\in \bbB$, let $\bbV_{z'} := \left\{v\text{ s.t. } \left|e^{-iv} - e^{-i\phi_{z'z}}\right| \leq \delta\right\}$. Then SubClaim \ref{subclm:phase-estimator} shows that the probability that when $z'\in\bbB$, the probability that $\cU_z(z')$ outputs $v\notin\bbV_{z'}$ is negligible. Therefore 
\begin{align*}
    &\sum_v\Prr[v = \cU_z(z')]\left|e^{-i\phi_{z'z}}-e^{-iv}\right|^2 \\
    &\leq \sum_{v\in\bbV_{z'}} \Prr[v = \cU_z(z')]\left|e^{-i\phi_{z'z}}-e^{-iv}\right|^2 + \sum_{v\notin\bbV_{z'}} 4\Prr[v = \cU_z(z')]\\
    &\leq \sum_{v\in\bbV_{z'}} \Prr[v = \cU_z(z')]\left|e^{-i\phi_{z'z}}-e^{-iv}\right|^2 + \negl(n)\\
    &\leq \sum_{v\in\bbV_{z'}} \Prr[v = \cU_z(z')]\delta^2 + \negl(n)\\
    &\leq \delta^2 + \negl(n)
\end{align*}
Plugging back in \eqref{eq:phase-step-error-expand} gives
\begin{align}
    \label{eq:expression-to-simplify}
    \zeta^2 &\leq \sum_{z' \in \bbB}a^2_{z'}\delta^2 + \sum_{z' \notin \bbB}4a^2_{z'}\nonumber\\
    &\leq \delta^2 + \sum_{z' \notin \bbB}4a^2_{z'}\nonumber\\
    &\leq \delta^2 + 4/n^k + 4/n^{3k} +16n^{3k}/p^{1/8} \nonumber\\
    &\leq \frac{128n^{12k}}{p(n)^{1/4}} + 4/n^k + 4/n^{3k} +\frac{16n^{3k}}{p(n)^{1/8}}
\end{align}
where the third step follows from parts 2, 3, and 4 of SubClaim \ref{subclm:markov-sets} and the last step substitutes the value of $\delta$. The final expression follows from taking the square root of both sides and noting that the square root function is subadditive.
\end{proof}

We can now begin analyzing the algorithm $\cB'(s,c)$, dropping the advice state from the input for notational convenience. For all $t$, define $\ket{\sigma_{t}}$ as follows:
\[
    \ket{\sigma_{t}} := (U_{t})^\dagger P'U_{t}\left(\tM_\text{amp} \ket{s,c}_{\sR}\ket{0^n}_{\sZ}\adv{amp}_\sA\right)\ket{0}_V\adv{1}_{V'}
\]
and define $\rho_{t}$ as the state on the $\sZ$ register of $\ketbra{\sigma_{t}}$ after tracing out the remaining registers, i.e.
\[
    \rho_{t} = \mathsf{Tr}_{\sR\sA\sV\sV'}\left(\ketbra{\sigma_{t}}\right)
\]
Note that the output of $\cB'(s,c)$ conditioned sampling $t$ is $c^\dag \rho_t c$. Next we show that $c^\dag \rho_t c$ is close to $\ketbra{\$_{s}}$ when $t \in \bbL\cap\bbU\cap\bbZ$. By $(s,c) \in \bbS'$ and the definition of $\bbS$, $\Delta_{s,c} \leq 1/\sqrt{p(n)}$. Let $\delta_\text{amp} := (\sqrt{3n^3}+2\sqrt{n})/p(n)^{1/4}$ (i.e. the error term in Claim \ref{clm:amp-synth} after substituting $\Delta_{s,c} \leq  1/\sqrt{p(n)}$) and let $\delta_\text{phase} := \frac{8\sqrt{2}n^{6k}}{p(n)^{1/8}} + \frac{2}{\sqrt{n^k}} + \frac{2}{\sqrt{n^{3k}}} +\frac{4\sqrt{n^{3k}}}{p(n)^{1/16}}$ (i.e. the error term in SubClaim \ref{subclm:phase-synth-for-good-sc}).

\begin{subclaim}
\label{subclm:cB'-for-good-t}
    If $t \in \bbL\cap\bbU\cap\bbZ$ then
\[
\bra{\$_{s}}c^\dagger\cdot \rho_{t}\cdot c\ket{\$_{s}} \geq 1-\delta_\text{amp} - \delta_\text{phase}
\]
\end{subclaim}
\begin{proof}
    Let
    \begin{align*}
        &\ket{\psi_1} := \ket{s, c} \ket{0^n} \adv{amp} \ket{0} \adv{1} \\
        &\ket{\psi_2} := \ket{s, c} \ket{\$^*_{s,c}} \adv{amp} \ket{0} \adv{1} \\
        &\ket{\psi_3} := e^{i\phi_t}\ket{s, c} \ket{\$_{s,c}}_\sZ \adv{amp} \ket{0} \adv{1} \\
        &M_1 := \tM_\text{amp} \\
        &M_2 := (U^{s,c}_t)^\dag P' U^{s,c}_t.
    \end{align*}
    By $(s,c) \in \bbS'$ and the definition of $\bbS$, $\Delta_{s,c} \leq 1/\sqrt{p(n)}$. Therefore, Claim \ref{clm:amp-synth} can be restated as
    \begin{align*}
        |M_1 \ket{\psi_1} - \ket{\psi_2}| \leq \delta_\text{amp}.
    \end{align*}
    Similarly, SubClaim \ref{subclm:phase-synth-for-good-sc} can be restated as
    \begin{align*}
        &|M_2 \ket{\psi_2} - \ket{\psi_3} | \leq \delta_\text{phase}.
    \end{align*}
    By the triangle inequality and subsituting the last two equations in the last step,
    \begin{align*}
        |M_2 M_1 \ket{\psi_1} - \ket{\psi_3}| &\leq |M_2 M_1 \ket{\psi_1} - M_2 \ket{\psi_2} | + |M_2 \ket{\psi_2} - \ket{\psi_3} | \\
        &= |M_1 \ket{\psi_1} - \ket{\psi_2} | + |M_2 \ket{\psi_2} - \ket{\psi_3}| \\
        &\leq \delta_\text{amp} + \delta_\text{phase}.
    \end{align*}
    By Theorem \ref{thm:trace-dist-and-euclidean-dist},
    \begin{align*}
        &\TD(M_2 M_1 \ketbra{\psi_1} M_1^\dag M_2^\dag, \ketbra{\psi_3}) \leq \delta_\text{amp} + \delta_\text{phase}.
    \end{align*}
    Noting that $\ket{\sigma_t}$ is defined as $M_2 M_1 \ket{\psi_1}$
    \begin{align*}
        &\TD(\ketbra{\sigma_{t}}, \ketbra{\psi_3}) \leq \delta_\text{amp} + \delta_\text{phase},
    \end{align*}
    We can trace out all but the $\sZ$ register from both states without increasing the trace distance. Therefore,
    \begin{align*}
        &\TD(\rho_{t}, \ket{\$_{s,c}} \bra{\$_{s,c}}) \leq \delta_\text{amp} + \delta_\text{phase}.
    \end{align*}
    Consider the projector $\ketbra{\$_{s,c}}$. The projection succeeds on the state $\ketbra{\$_{s,c}}$ with probability 1, so it must succeed on $\rho_{t}$ with probability atleast $1 - \TD(\rho_{t}, \ket{\$_{s,c}} \bra{\$_{s,c}}) \geq 1-\delta_\text{amp} - \delta_\text{phase}$. Therefore,
    \begin{align*}
        \bra{\$_s} c^\dag \rho_{t} c \ket{\$_s}= \bra{\$_{s,c}}\rho_{t} \ket{\$_{s,c}} \geq 1 - \delta_\text{amp} - \delta_\text{phase}.
    \end{align*}
\end{proof}
Next we show that with high probability,  $t \in \bbL\cap \bbU \cap \bbZ$. Let $\delta_\text{samp}:=3/p(n)^{1/4} - 1/n^k - 1/n^{3k}$ (i.e. the sum of error terms from  parts 1, 2, and 3 of SubClaim \ref{subclm:markov-sets}).
\begin{subclaim}
\label{subclm:sampling-t}
    Let $t$ be the outcome when measuring the $\sZ$ register of  $\tM_\text{amp} \ket{s,c}_{\sR}\ket{0^n}_{\sZ}\adv{amp}_\sA$ in the computational basis. Then 
    \[
    \Pr[t \in \bbL\cap \bbU\cap \bbZ] \geq 1- \delta_\text{samp} -\delta_\text{amp}
    \]
\end{subclaim}
\begin{proof}
   Consider the probability of obtaining a string $t'$ upon measuring the second register of $\ket{s,c}\ket{\$^*_{s,c}}\adv{amp}$.
    \[
        \Pr[t'] = \left|\bra{t'}\sum_z a_z \ket{z}\right|^2 = \left(a_z \right)^2
    \]
    Therefore,
    \begin{align*}
        \Pr[t' \in \bbL\cap \bbU\cap \bbZ] =& \sum_{z \in \bbL\cap \bbU\cap \bbZ}\left(a_z \right)^2\\
        \geq& 1 - \sum_{z \notin \bbL}\left(a_z \right)^2 - \sum_{z \notin \bbU}\left(a_z \right)^2 - \sum_{z \notin \bbZ}\left(a_z \right)^2\\
        \geq& 1- 3/p(n)^{1/4} - 1/n^k - 1/n^{3k} = 1- \delta_\text{samp}
    \end{align*}
    where the last step follows from parts 1, 2, and 3 of SubClaim \ref{subclm:markov-sets}, noting that $a_z = |\alpha_z|$.  By Claim \ref{clm:amp-synth} and Theorem \ref{thm:trace-dist-and-euclidean-dist} we know that $\tM_\text{amp} \ket{s,c}\ket{0}\adv{amp}$ and  $\ket{s,c}\ket{\$^*_{s,c}}\adv{amp}$ are atmost $\delta_\text{amp}$ apart in trace distance. Therefore, the output distributions upon measuring each state can be atmost $\delta_\text{amp}$ far. Therefore,
    \[
    \Pr[t \in \bbL\cap \bbU\cap \bbZ] \geq 1- \delta_\text{samp} -\delta_\text{amp}
    \]
\end{proof}
Putting together SubClaim \ref{subclm:sampling-t}
and SubClaim \ref{subclm:cB'-for-good-t} shows that the expected overlap of $\ketbra{\$_s}$ and $\cB'(s,c)$ is high.
\begin{subclaim}
\label{subclm:expected-overlap-cB'-for-good-sc}
    \[
    \underset{\cB'}{\bbE}\left[\bra{\$_s}\cB'(s,c)\ket{\$_s}\right] \geq 1 - \delta_\text{samp} - 2\delta_\text{amp} - \delta_\text{phase}
    \]
\end{subclaim}
\begin{proof}
    SubClaim \ref{subclm:cB'-for-good-t} shows that if $t$ sampled by $\cB'$ is in $\bbL \cap \bbU \cap \bbZ$ then the overlap is atleast $1-\delta_\text{amp} - \delta_\text{phase}$, while SubClaim \ref{subclm:sampling-t} shows that $t \in \bbL \cap \bbU \cap \bbZ$ occurs with probability atleast $1 - \delta_\text{samp} - \delta_\text{amp}$. Putting these together
    \begin{align*}
        \underset{\cB'}{\bbE}\left[\bra{\$_s}\cB'(s,c)\ket{\$_s}\right] &\geq \Pr[t \in\bbL\cap \bbU\cap \bbZ] \cdot (1-\delta_\text{amp} - \delta_\text{phase})\\
        &\geq (1 - \delta_\text{samp} - \delta_\text{amp})\cdot (1-\delta_\text{amp} - \delta_\text{phase})\\
        &\geq 1 - \delta_\text{samp} - 2\delta_\text{amp} - \delta_\text{phase}
    \end{align*}
    where the first step follows from SubClaim \ref{subclm:cB'-for-good-t} and the second step follows from SubClaim \ref{subclm:sampling-t}.
\end{proof}
Finally, we show that $\cB$ achieves the claimed bound. We note that SubClaim \ref{subclm:expected-overlap-cB'-for-good-sc} applies for arbitrary fixed $(s,c) \in\bbS'$ and SubClaim \ref{subclm:bbS'} shows that for $s$ sampled by $\cG$ and $c \leftarrow \cC_n$ the probability that $(s,c) \in \bbS$ is atleast $1- 1/p(n)^{1/2} - 1/n^k$. Putting them together
\begin{align*}
    \underset{s, \ket{\$_s} \leftarrow \cG(1^n)}{\bbE}\left[\bra{\$_s}\cB(s, \adv{synth})\ket{\$_s}\right]&\geq \Pr[(s,c) \in\bbS'] \cdot (1-\delta_\text{samp}-2\delta_\text{amp} - \delta_\text{phase})\\
        &\geq (1 - 1/p(n)^{1/2} - 1/n^k)\cdot (1-\delta_\text{samp}-2\delta_\text{amp} - \delta_\text{phase})\\
        &\geq 1 - 1/p(n)^{1/2} - 1/n^k -\delta_\text{samp}-2\delta_\text{amp} - \delta_\text{phase}
    \end{align*}
All that remains it to simplify the expression. Substituting the definitions of $\delta_\text{samp}$, $\delta_\text{amp}$, and $\delta_\text{phase}$, noting that $k > 6$ and $p(n)>n^{64k}$ and simplifying for large enough $n$
\begin{align*}
    1 - 1/p(n)^{1/2} - 1/n^k -\delta_\text{samp}-2\delta_\text{amp} - \delta_\text{phase} \geq 1- 3/\sqrt{n^k} 
\end{align*}
Finally, since $n^k \geq q(n)^3$
\begin{align*}
    \underset{s, \ket{\$_s} \leftarrow \cG(1^n)}{\bbE}\left[\bra{\$_s}\cB(s,\adv{synth})\ket{\$_s}\right]&\geq 1 - 3/q(n)^{3/2} \geq 1 - 1/q(n)
    \end{align*}

which concludes the proof of the Claim.
\end{proof}
Claim \ref{clm:state-synth} shows the existence of an algorithm that contradicts the security of $\cG$ which concludes the proof of the theorem.\end{proof}

Next, we prove that the existence of distributional one-way puzzles implies the existance of (standard) state puzzles.
\begin{theorem}
\label{thm:owp-imply-state-puzzles}
    If one-way puzzles (Definition \ref{def:owp}) exist then state puzzles (Definition \ref{def:state-puzzle}) exist.
\end{theorem}
Since Theorem \ref{thm:owp-amplification} shows that distributional one-way puzzles can be amplified to obtain (strong) one-way puzzles, the following is a corollary of Theorem \ref{thm:owp-imply-state-puzzles}.
\begin{corollary}
     If $1/q(n)$-distributional one-way puzzles (Definition \ref{def:dist-owp}) exist for some non-zero polynomial $q(\cdot)$ then state puzzles (Definition \ref{def:state-puzzle}) exist.
\end{corollary}
\begin{proof}[of Theorem \ref{thm:owp-imply-state-puzzles}]
Let $(\Gen(1^n), \Ver)$ be a one-way puzzle that samples $n$ bit puzzles and $n$ bit keys. Without loss of generality, we may assume that $\Gen(1^n)$ is the algorithm that first applies a unitary $U_n$ to $\ket{0}$ to obtain
\[
    U_n \ket{0}= \sum_{s,k} \sqrt{p_{s,k}} \ket{\mu_{s,k}}\ket{s}_\sS\ket{k}_\sK  
\]
where $\{\mu_{s,k}\}_{s,k}$ are a set of normalised states, $s$ and $k$ are an $n$-bit puzzle and and $n$-bit key respectively output by $\Gen(1^n)$ with probability $p_{s,k}$. This is followed by a classical basis measurement of registers $\sS$ and $\sK$ to obtain puzzle $s$ and key $k$.

Let $\cG(1^n)$ be the following algorithm
\begin{enumerate}
    \item Apply $U_n$ to $\ket{0}$ to obtain $\sum_{s,k} \sqrt{p_{s,k}} \ket{\mu_{s,k}}\ket{s}_\sS\ket{k}_\sK$
    \item Measure $\sS$ in the classical basis to obtain string $s$ and residual state $\ket{\psi_s}$.
    \item Return $(s, \ket{\psi_s})$.
\end{enumerate}

We will prove that $\cG(1^n)$ is a state puzzle, which sufficies to prove the theorem. Suppose for the sake of contradiction that $\cG(1^n)$ is not a state puzzle. There there exists a polynomial $p(n)$. QPT $\cA = \{\cA_\secpar\}_{\secpar \in \bbN}$, and an advice ensemble $\ket{\tau} = \{\ket{\tau_n}\}_{n \in \mathbb{N}}$ such that for all large enough $n \in \bbN$, 
\[
    \underset{(s,\ket{\psi_s})\leftarrow \cG(1^n)}{\bbE}\left[\bra{\psi_s}\cA(\ket{\tau}, s)\ket{\psi_s}\right] \geq  \frac{1}{p(n)}
\] 
Fix any such adversary $\cA$, any such advice ensemble $\ket{\tau}$, and any such large enough $n\in \bbN$. We will drop the advice from the notation since it is always implicitly provided to the adversary. 

Define the algorithm $\cA'$ that takes input $s$, computes $\cA(s)$, and outputs the outcome of measuring the $\sK$ register of $\cA(s)$ in the computational basis. We will show that $\cA'$ contradicts the security of the one-way puzzle $(\Gen(1^n), \Ver)$.

Let $p_s$ be the probability that $\Gen(1^n)$ samples puzzle $s$. Note that this is identical to the probability that $\cG(1^n)$ samples $s$ since in both cases $s$ is generated the same way. We first show that the expected trace distance between $\ketbra{\psi_s}$ and $\cA(s)$ is at most $\sqrt{1 -1/p(n)}$.
\begin{claim}
\label{clm:td-and-overlap}
    \[
    \underset{(s,\ket{\psi_s})\leftarrow \cG(1^n)}{\bbE}\left[\TD\Big(\ketbra{\psi_s}, \cA( s)\Big)\right] \leq  \sqrt{1 -1/p(n)}
\] 
\end{claim}
\begin{proof}
    
    For any $s$, the fidelity $F(\cA(s),\ket{\psi_s})$ of $\cA(s)$ and $\ket{\psi_s}$ is $\sqrt{\bra{\psi_s}\cA(s)\ket{\psi_s}}$. Therefore, by Uhlmann's Theorem, for any $s$
    \begin{align*}
        \TD\Big(\ketbra{\psi_s}, \cA( s)\Big) &\leq \sqrt{1 - F(\cA(s),\ket{\psi_s})^2}\\
        &= \sqrt{1 - \bra{\psi_s}\cA(s)\ket{\psi_s}}
    \end{align*}
    Expressing $ \underset{(s,\ket{\psi_s})\leftarrow \cG(1^n)}{\bbE}\left[\TD\Big(\ketbra{\psi_s}, \cA( s)\Big)\right]$ as a sum over $s$ and applying the above
    \begin{align*}
        \underset{(s,\ket{\psi_s})\leftarrow \cG(1^n)}{\bbE}\left[\TD\Big(\ketbra{\psi_s}, \cA( s)\Big)\right] &=\sum_s p_s \cdot \TD\Big(\ketbra{\psi_s}, \cA( s)\Big)\\
        &\leq\sum_s p_s \cdot \sqrt{1 - \bra{\psi_s}\cA(s)\ket{\psi_s}}
    \end{align*}
    Applying Jensen's inequality
    \begin{align*}
        \underset{(s,\ket{\psi_s})\leftarrow \cG(1^n)}{\bbE}\left[\TD\Big(\ketbra{\psi_s}, \cA( s)\Big)\right] 
        &\leq\sqrt{\sum_s p_s \cdot\left(1 - \bra{\psi_s}\cA(s)\ket{\psi_s}\right)}\\
        &=\sqrt{1-\sum_s p_s \cdot\bra{\psi_s}\cA(s)\ket{\psi_s}}
    \end{align*}
    Now, we can rewrite the fact that $\underset{(s,\ket{\psi_s})\leftarrow \cG(1^n)}{\bbE}\left[\bra{\psi_s}\cA(s)\ket{\psi_s}\right] \geq  \frac{1}{p(n)}$ as a sum over $s$.
    \begin{align*}
        \sum_s p_s \cdot \bra{\psi_s}\cA(s)\ket{\psi_s} \geq 1/p(n)
    \end{align*}
    which when plugged into the previous inequality implies
    \[
    \underset{(s,\ket{\psi_s})\leftarrow \cG(1^n)}{\bbE}\left[\TD\Big(\ketbra{\psi_s}, \cA( s)\Big)\right] \leq  \sqrt{1 -1/p(n)}
\] 
\end{proof}

\begin{claim}
    \[
    \Prr_{(s,k) \leftarrow \Gen(1^n)}\left[\Ver(s,\cA'(s)) = 1\right] \geq 1/3p(n) 
    \]
\end{claim}
\begin{proof}
By the correctness of the one-way puzzle
\[
    \Prr_{(s,k) \leftarrow \Gen(1^n)} [\Ver(s,k) = 1] \geq 1 - \negl(n)
\]
Let $M$ be an algorithm that takes input $\ket{\psi}$ and returns the result $k$ of measuring the $\sK$ register in the computational basis. The distribution over $s$ and $k$ obtained by sampling $(s, \ket{\psi_s})$ from $\cG(1^n)$ and sampling $k$ from $M(\ketbra{\psi_s})$ is therefore identical to the distribution obtained by sampling $(s,k)$ from $\Gen(1^n)$.
We can express the probability that $\cA'(s)$ successfully outputs a key that passes verification as a sum over $s$. Therefore
\[
    \Prr_{\substack{(s,\ket{\psi_s}) \leftarrow \cG(1^n)\\ k \leftarrow M(\ketbra{\psi_s})}} [\Ver(s,k) = 1] \geq 1 - \negl(n)
\]
which may be rewritten as
\[
    \Prr_{(s,\ket{\psi_s}) \leftarrow \cG(1^n)} [\Ver(s,M(\ketbra{\psi_s})) = 1] \geq 1 - \negl(n)
\]
and expressed as a sum over $s$ as follows.
\[
    \sum_s p_s\cdot\Prr[\Ver(s,M(\ketbra{\psi_s})) = 1] \geq 1 - \negl(n)
\]
For any $s$ and any state $\rho$, 
\[
\Big|\Pr\Big[\Ver(s,M(\ketbra{\psi_s}))=1\Big] - \Pr\Big[\Ver(s,M(\rho))=1\Big]\Big| \leq \TD(\ketbra{\psi_s}, \rho)
\]
which means that we can replace $\ketbra{\psi_s}$ with $\cA(s)$ at the cost of an error of $\TD(\ketbra{\psi_s}, \cA(s))$, i.e.
\begin{align*}
    1 - \negl(n) &\leq \sum_s p_s\cdot\Big(\Prr\Big[\Ver\Big(s,M(\cA(s))\Big) = 1\Big] + \TD\Big(\ketbra{\psi_s}, \cA(s)\Big)\Big) \\
    &\leq \sum_s p_s\cdot\Big(\Prr\Big[\Ver\Big(s,M(\cA(s))\Big) = 1\Big]\Big) + \underset{s,\ket{\psi_s}\leftarrow \cG(1^n)}{\bbE}\left[\TD\Big(\ketbra{\psi_s}, \cA( s)\Big)\right]
\end{align*}
The first term is exactly $ \Prr_{s,k \leftarrow \Gen(1^n)}\left[\Ver(s,\cA'(s)) = 1\right]$ and the second term is shown in Claim \ref{clm:td-and-overlap} upper bounded by $\sqrt{1-p(n)}$. Therefore,
\begin{align*}
    \Prr_{s,k \leftarrow \Gen(1^n)}\left[\Ver(s,\cA'(s)) = 1\right] &\geq 1-\negl(n) - \sqrt{1-p(n)}\\
    &\geq 1-\negl(n) - (1-p(n)/2)\\
    &\geq p(n)/2-\negl(n)\\
    &\geq p(n)/3
\end{align*}
where the last step holds for large enough $n$.
\end{proof}

This shows that $\cA'$ contradicts the security of the one-way puzzle, concluding the proof of the theorem.
\end{proof}

We also have the following straightforward corollary, which follows from the implication of state puzzles from quantum money mini-schemes (and other unclonable primitives).
\begin{corollary}
    Quantum money mini-schemes imply one-way puzzles and quantum bit commitments.
\end{corollary}

\section*{Acknowledgments}
We thank Scott Aaronson, Lijie Chen and William Kretschmer for helpful conversations about the (im)possibility of quantumly efficiently sampling matrices jointly with their permanents.
We thank Alexandra (Sasha) Levinshteyn for help with typesetting parts of this manuscript. Finally, we thank Daniel Apon, Tomoyuki Morimae, Barak Nehoran, Luowen Qian and Amit Sahai for useful comments on the writeup.

The authors were supported in part by AFOSR, NSF 2112890, NSF CNS-2247727 and a Google Research Scholar award. 
This material is based upon work supported by the Air Force Office of Scientific Research under award
number FA9550-23-1-0543.

\bibliographystyle{alpha}
\addcontentsline{toc}{section}{References}
\bibliography{abbrev0,crypto,custom,bib}
\appendix

\section{
Instantiating Uniform Approximation Hardness (Definition \ref{def:type-2})}
In this section we will show how to instantiate uniform approximation hardness with well studied conjectures from the sampling-based quantum advantage literature. Specifically, we will import some conjectures from the literature on BosonSampling, Random Circuit Sampling, IQP and DQC1 sampling; and will discuss why they imply Definition \ref{def:type-2}.

\subsection{Random Circuit Sampling}

The exposition in this section is primarily taken from \cite{BFNV19}. Define a circuit architecture $A := \{A_n\}_{n\in\bbN}$ as a family of graphs with $\poly(n)$ vertices, where each vertex $v$ has $\mathsf{deg}_\mathsf{in}(v) = \mathsf{deg}_\mathsf{out}(v)\in\{1,2\}$. Intuitively, the vertices of the graph denote one or two qubit gates and the edges denote wires. A quantum circuit is instantiated by specifying the gate for each vertex. Define $\cH_A$ as the distribution over circuits formed by drawing a (one or two qubit) gate independently from the Haar measure for each vertex in $A$ and assigning the gate to the vertex. 

\begin{definition}[Anticoncentration]
\label{def:rcs-anti}
    For an architecture $A$, we say that RCS anticoncentrates on $A$ if there exist constants $\kappa,\gamma > 0$ such that for all large enough $n$
    \[
        \Prr_{C\leftarrow\cH_{A_n}}\left[\Pr_C[0^n] \geq \frac{1}{\kappa2^n}\right] \geq \gamma
    \]
\end{definition}
\begin{definition}[Hiding]
\label{def:rcs-hide}
    For an architecture $A$, we say that $\cH_A$ has the hiding property if for any $C \leftarrow \cH_{A_n}$ and uniformly random $y \leftarrow \bin^n$,  $C_y$ is distributed as $\cH_{A_n}$ where $C_y$ is the circuit such that $\Pr_C[x]= \Pr_{C_y}[x\oplus y]$, i.e. the circuit $C$ with $X$ gates appended to every output wire where the value of the output bit in $y$ is $1$.
\end{definition}

\begin{definition}[Approximate Average-Case Hardness] 
\label{def:rcs-main}
An architecture $A:= \{A_n\}_{n\in\bbN}$ is said to be approximate average-case $\#P$-hard to approximate if it has the following property. There exist functions $\epsilon(n) = 1/p(n)$ and $ \delta(n) = 1/q(n)$ for some polynomials $p$ and $q$ such that for any oracle $\cO$ s.t. for all large enough $n$
\footnote{In the literature a slightly different form is often used where $\cO$ takes $1^{1/\epsilon}$ and $1^{1/\delta}$ as input and must approximate to precision $\epsilon/2^n$ with error probability at most $\delta$. The version we present is more convenient and cleaner for the purpose of building cryptography, but our results hold for both versions. See the proof of Theorem \ref{thm:type2-implies-type1} for details.}
\[
    \Prr_{C\leftarrow\cH_{A_n}}\left[\left|\cO(C) - \Pr_C[0^n]\right| \leq \epsilon(n)/2^n\right] \geq 1-\delta(n)
\]
it holds that $\mathsf{P}^{\#\mathsf{P}} \subseteq \mathsf{BPP}^\cO$.
\end{definition}

RCS is based on the conjecture that there exists an architecture $A$ that satisfies Definition \ref{def:rcs-anti}, Definition
    \ref{def:rcs-hide}, and Definition
    \ref{def:rcs-main} (see, e.g.,~\cite{BFNV19}).
We show that this implies the following conjecture, which directly leads to an instantiatiation of Definition \ref{def:type-2}.

\begin{conjecture}
\label{conj:rcs-2}
    There exists an architecture $A$ and polynomials $p$ and $q$ such that:
    \begin{enumerate}
        \item \textbf{Anticoncentration. } $$\Prr_{\substack{C \leftarrow \cH_{A_n}\\x\leftarrow\bin^n}}[\Pr_C[x] \geq 1/(p(n)\cdot2^n)] \geq 1/\gamma(n)$$
        \item \textbf{Hardness.} For any oracle $\cO$ satisfying that for all large enough $n \in \bbN$, 
        \[
            \Prr_{\substack{C \leftarrow D_n\\x\leftarrow\bin^n}}\left[ \left|\cO(C, x) - \Pr_{C}[x]\right| \leq \frac{\Pr_{C}[x]}{p(n)} \right] 
            \geq 1/\gamma(n)-\frac{1}{p(n)}
        \]
    \end{enumerate}
we have that $\mathsf{P}^{\#\mathsf{P}} \subseteq \mathsf{BPP}^{\cO}$.
\end{conjecture}

To see why the existence of an architecture satisfying Definition \ref{def:rcs-anti}, Definition
    \ref{def:rcs-hide}, and Definition
    \ref{def:rcs-main} implies Conjecture \ref{conj:rcs-2}, 
it is first observed that
anticoncentration holds directly from Definition \ref{def:rcs-anti} and Definition \ref{def:rcs-hide}. Next, the hiding property  implies that approximate average-case hardness holds even for an oracle that takes input $C\leftarrow \cH_{A_n}$ and $x\leftarrow\bin^n$ and estimates $\Pr_C(x)$. Finally, anticoncentration implies that estimating probabilities with small additive error on average implies the ability to estimate probabilities with small relative error on a large fraction of the set of anticoncentrated points.

\subsection{Boson Sampling}
This section imports conjectures that were made in~\cite{AA11} to obtain quantum advantage from Boson Sampling; and discusses why these conjectures imply Definition \ref{def:type-2}.

\begin{conjecture}
\label{con:pgc}[Permanent-of-Gaussians-Conjecture]
There exist polynomials $p(\cdot), q(\cdot)$ such that for $\epsilon = {1}/{p(n)}, \delta = {1}/{q(n)}$, if there exists an oracle $\cO$ that given as input a matrix $X \sim \cN (0, 1)^{n\times n}_{\mathbb{C}}$ of i.i.d. Gaussians, can estimate $\mathsf{Per(X)}$ to within error $\pm \epsilon(n) \cdot |\mathsf{Per}(X)|$, with probability at least $1-\delta(n)$ over $X$, then
$\mathsf{P}^{\#\mathsf{P}} \subseteq \mathsf{BPP}^{\cO}$. 
\end{conjecture}

\begin{conjecture}
\label{con:pac}
[Permanent Anti-Concentration Conjecture]
There exists a polynomial $p(\cdot)$ such that for all $n$ and $\delta > 0$,
$$ \Pr_{X \sim \cN(0,1)^{n\times n}_{\mathbb{C}}} \Big[ |\mathsf{Per}(X)| < \frac{\sqrt{n!}}{p(n,1/\delta)} \Big] < \delta$$
\end{conjecture}

\begin{theorem} 
   Conjectures \ref{con:pgc} and \ref{con:pac} imply Definition \ref{def:type-2}.
\end{theorem}
\begin{proof} (Sketched, from~\cite{AA11,grieretal})
It is shown in~\cite{AA11} how the probability that a randomly chosen linear optical network (i.e. circuit) outputs $0$ is proportional to the square of the permanent of an appropriate submatrix of a Haar random matrix; where the submatrix itself is Gaussian. They also prove a hiding property, which argues that any given Gaussian matrix can be embedded into a Haar random matrix while keeping the location of the given submatrix hidden. The embedding procedure itself is known to be in $\mathsf{FBPP}^{\mathsf{NP}}$ and is conjectured in~\cite{AA11} to be in $\mathsf{FBPP}$. Subsequent models of Boson Sampling such as (Bipartite) Gaussian Boson Sampling modify the experimental setup and work with Gaussian (as opposed to Haar random) matrices, which allow trivially embedding a Gaussian submatrix in $\mathsf{FBPP}$.

Specifically,~\cite{grieretal} efficiently (in BPP) reduce the task of computing permanents of Gaussian random matrices to the task of estimating the probabilities of (random) outputs of a bipartite Gaussian boson sampling setup. Under Conjectures \ref{con:pgc} and \ref{con:pac}, this proves that the resulting probabilities are $\#\mathsf{P}$ hard to approximate, which implies the statement of Definition \ref{def:type-2}.
\end{proof}

\subsection{IQP Circuit Sampling}

IQP refers to a class of randomly chosen commuting quantum circuits, which take as input the
state $\ket{0^n}$, whose gates are diagonal in the Pauli-X basis,
and whose $n$-qubit output is measured in the computational basis. Under {\em any one of} the two conjectures below, one coming from condensed matter physics and the other from computer science, the output distributions of IQP circuits have been proven classically hard to simulate (unless the polynomial heirarchy collapses)~\cite{shepherd2009temporally,bremneraverage}.  

\begin{conjecture}
\label{con:iqp1}
Consider the partition function of the general Ising model,
\begin{equation}
Z(\omega) = \sum_{\mathbf{z} \in \{\pm 1\}^n} \omega \exp\left(\sum_{i < j} w_{ij} z_i z_j + \sum_{k=1}^{n} v_k z_k\right),
\end{equation}
where the exponentiated sum is over the complete graph on \(n\) vertices, \(w_{ij}\) and \(v_k\) are real edge and vertex weights, and \(\omega \in \mathbb{C}\). Let the edge and vertex weights be picked uniformly at random from the set \(\{0, \dots, 7\}\).

Then it is \#P-hard to approximate \(\left| Z(e^{i\pi/8}) \right|^2\) up to multiplicative error \(1/4 + o(1)\) for a 1/24 fraction of instances, over the random choice of weights.
\end{conjecture}

\begin{conjecture}
\label{con:iqp2}
    Let \( f : \{0, 1\}^n \to \{0, 1\} \) be a uniformly random degree-3 polynomial over \( \mathbb{F}_2 \), and define \(\text{ngap}(f) := (|\{x : f(x) = 0\}| - |\{x : f(x) = 1\}|) / 2^n\). Then it is \#P-hard to approximate \(\text{ngap}(f)^2\) up to a multiplicative error of \(1/4 + o(1)\) for a 1/24 fraction of polynomials \(f\).
\end{conjecture}

Just like the case of random circuit sampling, IQP circuits satisfy the hiding property trivially due to the circuit architecture supporting the addition of random X gates. Furthermore, the structure of IQP circuits allows provable anti-concentration results for $ngap(f)$ and the partition function of the random Ising model. These properties combined help reduce the task of approximating the partition function of the general Ising model, or approximating $ngap(f)^2$ within the bounds in the conjecture, to the task of approximating output probabilities of randomly chosen IQP circuits. Thus {\em either one of} Conjectures \ref{con:iqp1} or \ref{con:iqp2} implies hardness according to Definition \ref{def:type-2}.

\section{Flatness of Unitary 2-Design Output Distributions}
\label{appendix:2-design}
\newtheorem*{flatness}{Theorem~\ref{thm:flatness-of-2-designs}}
\begin{flatness}[Flatness of 2-designs]Let $\cC$ be a unitary 2-design on $n$ qubits. Fix any $n$ qubit state $\ket{\psi}$. For any $C \in \Supp(\cC)$,  let $p_C(x):=|\bra{x}C\ket{\psi}|^2$ be the probability that measuring $C\ket{\psi}$ in the computational basis results in $x$. Then the following holds for all $k>6$ and sufficiently large $n$. Define $$\bbG := \left\{C \in \Supp(\cC) : \underset{x: p_C(x) \geq \frac{n^{3k}}{2^n}}{\sum} p_C(x) \leq 1/n^k \right\}$$ Then $$\Prr_{C\leftarrow\cC}[C\in\bbG]\geq 1 - 1/n^k$$
\end{flatness}
\begin{proof}

By the Chebyshev bound applied to random variable $p_C(x)$ for uniformly random $x$ and $C\leftarrow\cC$, for all $\alpha > 0$
\begin{align}   
\label{eq:chebyshev}
\Prr_{\substack{x\leftarrow\bin^n\\C\leftarrow\cC}}\left[|p_C(x) - 1/2^n| \geq \frac{\alpha}{2^n}\right] &\leq \frac{2^{2n}\left(\bbE_{C,x} [p_C(x)^2] - \bbE_{C,x}[p_C(x)]^2\right)}{\alpha^2}\nonumber\\
&\leq \frac{2^{2n}\left(\bbE_{C,x} [p_C(x)^2] - 1/2^{2n}\right)}{\alpha^2}
\end{align}

where the second step follows from the expectation over $x$ of $p_C(x)$ is $1/2^n$ for all $C$. Since $C$ is sampled from a unitary 2-design, we show a bound on the second moment. Fix any $x$. By the definition of $p_C(x)$,
\begin{align*}
    \bbE_C[p_C(x)^2] &= \bbE_C\left[\left(\bra{x}C^\dagger\ketbra{\psi} C \ket{x}\right)^2\right]\\
    &= \bbE_C\left[\bra{x}^{\otimes 2}C^{\otimes 2}\ketbra{\psi}^{\otimes 2} C^{\dagger\otimes 2} \ket{x}^{\otimes 2}\right]\\
    &= \bra{x}^{\otimes 2}\bbE_C\left[C^{\otimes 2}\ketbra{\psi}^{\otimes 2} C^{\dagger\otimes 2} \right]\ket{x}^{\otimes 2}
\end{align*}

By the definition of unitary 2-designs (Definition \ref{def:unitary-2-design}), $\bbE_C\left[C^{\otimes 2}\ketbra{\psi}^{\otimes 2} C^{\dagger\otimes 2} \right]$ equals the second moment operator with respect to the Haar measure applied to $\ketbra{\psi}^{\otimes 2}$, i.e. $\underset{U \leftarrow \mu_H}{\bbE}\left[U^{\otimes 2} \ketbra{\psi}^{\otimes 2} U^{\dag\otimes 2}\right]$ which is known to equal $\frac{\bbI + \bbF}{2^n(2^n+1)}$ where $\bbI$ is identity and $\bbF$ is the flip operator (see Corollary 13 in \cite{haar-intro}). Therefore

    \begin{align*}\bbE_C\left[C^{\otimes 2}\ketbra{\psi}^{\otimes 2} C^{\dagger\otimes 2} \right] = \frac{\bbI + \bbF}{2^n(2^n+1)}
\end{align*}
 We can now compute $\bbE_C[p_C(x)^2]$ as follows
\begin{align*}
     \bbE_C[p_C(x)^2] &=\bra{x}^{\otimes 2}\frac{\bbI + \bbF}{2^n(2^n+1)}\ket{x}^{\otimes 2}\\
     &= \frac{2}{2^n(2^n+1)}
\end{align*}
Since the above holds for any $x$, it also holds for uniformly sampled $x$, which implies
\begin{align*} 
\bbE_{x,C} [p_C(x)^2] &= \frac{2}{2^n(2^n+1)}
\end{align*}
Substituting in $\eqref{eq:chebyshev}$ gives
\begin{align*}
    \Prr_{\substack{x\leftarrow\bin^n\\C\leftarrow\cC}}\left[|p_C(x) - 1/2^n| \geq \frac{\alpha}{2^n}\right] \leq \frac{1}{\alpha^2}
\end{align*}
We therefore obtain a bound on the fraction of $x$ with $p_C(x)$ above a threshold. However, we need to bound the total probability mass on such $x$. To do so we leverage the fact that the above bound holds for every $\alpha$. We can therefore simultaneously bound the fraction of such $x$ for every a sequence of $\alpha$ values. This turns out to be sufficient to bound the probability mass. 

Define $f_\alpha(C) := \Prr_{x\leftarrow\bin^n}\left[p_C(x) \geq \frac{1+\alpha}{2^n}\right]$. Intuitively, this is the fraction of $x$ with $p_C(x)$ above the required threshold.
Therefore
\[
\bbE_C[f_\alpha(C)] \leq 1/\alpha^2
\]
By a Markov argument
\[
\Prr_C[f_\alpha(C) \geq 1/\alpha^{1.5}] \leq 3/\alpha^{0.5}
\]
For any $i \in \bbN$, setting $\alpha$ to $n^{2i}$
\[
\Prr_C[f_{n^{2i}}(C) \geq 1/n^{3i}] \leq 3/n^{i}
\]
Let $\bbG_i := \left\{C : f_{n^{2i}}(C) < 1/n^{3i}\right\}$. Then $\Prr_C[C\in \bbG_i] \geq 1-3/n^i$. Let $\bbG' := \bigcap_{i\geq k+3} \bbG_i$. Then 
\begin{align*}
    \Prr_C[C\in \bbG'] &\geq 1 - \left(3/n^{k+3} + 3/n^{k+4} + \ldots\right)\\
    &\geq 1-\frac{3}{n^{k+3}}(1 + 1/n + 1/n^2 + \ldots)\\
    &\geq 1 - 6/n^{k+3}
\end{align*}
Fix any $C \in \bbG'$.
\begin{align*}
\sum_{x:p_C(x) \geq \frac{1+n^{2k+6}}{2^n}} p_C(x) &= \sum_{i\geq k+3} \sum_{x:\frac{1+n^{2i+2}}{2^n} \geq p_C(x) \geq \frac{1+n^{2i}}{2^n}}p_C(x)\\
&\leq \sum_{i\geq k+3} \sum_{x:\frac{1+n^{2i+2}}{2^n}\geq p_C(x) \geq \frac{1+n^{2i}}{2^n}}\frac{1+n^{2i+2}}{2^n}\\
&\leq \sum_{i\geq k+3} \frac{2^n}{n^{3i}}\cdot \frac{1+n^{2i+2}}{2^n}\\
&\leq \sum_{i\geq k+3} \frac{2^n}{n^{3i}}\cdot \frac{1+n^{2i+2}}{2^n}\\
&\leq  \sum_{i\geq k+3} 2/n^{i-2}\\
&\leq 4/n^{k+1}
\end{align*}
where the third step uses the definitions of $\bbG'$, $\bbG_i$, and $f_{n^{2i}}$. For $C\in\bbG'$, all $k > 6$ and sufficiently large $n$
\[
    \sum_{x:p_C(x) \geq \frac{n^{3k}}{2^n}} p_C(x) \leq 1/n^k
\]
and $\Prr_C[C\in\bbG']\geq 1-1/n^k$, which concludes the proof.
\end{proof}

\section{Proof of Theorem \ref{thm:geometric-arg}}
\label{appendix:geometric-arg}
\newtheorem*{geometric}{Theorem~\ref{thm:geometric-arg}}
\begin{geometric}    
        Let $(x,y), (x^*, y^*) \in \bbR^2$ such that $\exists \gamma >0, \gamma'>0$ s.t.
        \begin{itemize}
            \item $x^2 + y^2 \geq \gamma^2$
            \item $(x-x^*)^2 + (y-y^*)^2 \leq (\gamma')^2$
            \item $\gamma' < \gamma$
        \end{itemize}
        Then $|e^{-i\cdot\arctanb(y,x)} - e^{-i\cdot\arctanb(y^*, x^*)}| \leq 2\gamma'/\gamma$
    \end{geometric}
    \begin{proof}
        We prove this theorem by a geometric argument. Let $X = (x,y)$ and $Y=(x^*,y^*)$ be points on the Cartesian plane, and let $O = (0,0)$ be the origin. Let the (smaller) angle between rays $OX$ and $OY$ be $\zeta$. Then $|\arctanb(y,x) - \arctanb(y^*, x^*)|$ is equal to $\zeta$ (upto a multiple of $2\pi$ offset). By expanding the expression we aim to bound
        \begin{align*}
            |e^{-i\cdot\arctanb(y,x)} - e^{-i\cdot\arctanb(y^*, x^*)}| &= \left|2\sin\left(\frac{\arctanb(y,x) - \arctanb(y^*, x^*)}{2}\right)\right|\\
            &=\left|2\sin\left(\frac{\zeta}{2}\right)\right|\\
            &\leq |\zeta|
        \end{align*}Therefore it suffices to show that $\zeta$ is upper bounded by $2\gamma'/\gamma$.

        Let $C$ be a circle of radius $\gamma'$ centered at $X$. Since $|OX| = \sqrt{x^2 + y^2} \geq \gamma > \gamma'$, $O$ is strictly external to the circle. Additionally, since $|XY| = \sqrt{(x-x^*)^2 + (y-y^*)^2} \leq \gamma'$, $Y$ is on the circle or internal to it. Therefore, $OY$ is either a tangent or secant line. The angle between $OX$ and $OY$ is maximized if $OY$ is tangent to the circle, in which case $OY \bot XY$ and $\sin(\zeta)= |XY|/|OX| \leq \gamma'/\gamma$. Additionally note that $\zeta$ is acute, so $\zeta \leq 2\sin(\zeta) \leq 2\gamma'/\gamma$, which concludes the proof.
    \end{proof}

\section{Quantum Advantage from One-Way Puzzles in Microcrypt}
\label{app:owpowf}
 In this section we show that if one-way functions \textit{do not} exist, the existence of one-way puzzles implies sampling based quantum advantage\footnote{The first version of this manuscript had the informal argument sketched out in the Results section, but we have added a formal proof to the current (updated) version.}. Let $(\Gen, \Ver)$ be a one-way puzzle. Then we claim that the existence of an efficient classical algorithm that  samples from a distribution less than $1/3$ statistically-far from the distribution of $\Gen(1^n)$ on all large enough $n \in \mathbb{N}$ implies the existence of one-way functions. 

Suppose there exists an efficient classical algorithm $\cS$ such that for all large enough $n$
\[
    \SD(\Gen(1^n), \cS(1^n; r)) \leq 1/3
\]
where $r$ is a uniform string of $n$-bits.
\begin{claim}Define $f(x)$ for $x\in\bin^n$ as follows:
\begin{itemize}
    \item $(s',k') \leftarrow \cS(1^n;x)$
    \item Output $s'$ 
\end{itemize}
Then $f$ is a $1/n$-distributional one-way function.
\end{claim}
\begin{proof}
Suppose not, i.e. there exists a non-uniform  PPT adversary $\cA$ such that for infinitely many $n\in\bbN$ 
\begin{equation}
\label{eq:puzz->adv}
\{s', x\} \approx_{1/n} \{s', \cA(s')\}
\end{equation}
where $x$ is a uniform $n$-bit string and $(s',k') := \cS(1^n;x)$.

\noindent Let $\cS'$ be the algorithm that on input $x$ performs the following:
\begin{itemize}
    \item $s^*, k^* \leftarrow \cS(1^n; x)$
    \item Return $k^*$
\end{itemize}
Consider the distribution $\{s', \cS'(\cA(s'))\}$. Since $\{s', \cA(s')\}$ is at most $1/n$ far from $\{s', x\}$, $\{s', \cS'(\cA(s'))\}$ is also (at most) $1/n$ far from $\{s', \cS'(x)\}$. Since $(s', k')$ is obtained by running $\cS(1^n;x)$ and $\cS'(x)$ simply outputs $k'$,  $\{s', \cS'(x)\}$ is identical to $\{s', k'\}$. This gives
\[
\{s', k'\} \approx_{1/n} \{s', \cS'(\cA(s'))\}
\]
Since the distribution over $(s',k')$ is $1/3$ close to the distribution of $\Gen(1^n)$
\[
\SD\left(\{s, k\}, \{s',k' \}\right) \leq 1/3 \text{ and } \SD\left(\{s, \cS'(\cA'(s))\},\{s', \cS'(\cA(s'))\}\right) \leq 1/3
\]
where $(s,k)\leftarrow\Gen(1^n)$. Putting these together with the previous equation we get
\[
 \SD\left(\{s, k\}, \{s, \cS'(\cA'(s))\}\right) \leq 2/3 +1/n
\]
Now, by the correctness of the one-way puzzle
\[\Prr_{s,k\leftarrow\Gen(1^n)}[\Ver(s,k) =1] \geq 1-\negl(n)\]
which implies
\[\Prr_{s,k\leftarrow\Gen(1^n)}[\Ver(s,\cS'(\cA'(x))) =1] \geq 1/3 -1/n -\negl(n)\]
which contradicts the security of the one-way puzzle.
\end{proof}

\paragraph{A Note on Quantum Advantage from One-way Puzzles.}
Since the existence of distributional one-way functions implies the existence of one-way functions \cite{STOC:ImpLevLub89}, if one-way puzzles exist but one-way functions do not, there exist distributions that can be efficiently sampled by quantum machines but cannot be efficiently sampled by classical machines. 
It was shown in~\cite{MYadv2}
that the existence of one-way functions implies (interactive, inefficiently verifiable) proofs of quantumness. Combining the previous two facts yields a direct proof that the existence of one-way puzzles implies interactive quantum advantage.
\end{document}